\documentclass[journal]{IEEEtran}
\usepackage{amsmath,amssymb,amsfonts}
\usepackage{array,amsthm}
\usepackage[caption=false,font=normalsize,labelfont=sf,textfont=sf]{subfig}
\usepackage{textcomp}
\usepackage{stfloats}
\usepackage{url}
\usepackage{verbatim}
\usepackage{graphicx}
\usepackage{cite}
\usepackage{algorithm,algpseudocode,algorithmicx}
\usepackage{xcolor}
\usepackage{booktabs}
\usepackage{comment}
\usepackage[flushleft]{threeparttable}
\usepackage[font=small,labelfont=bf]{caption}
\usepackage{tikz}
\usetikzlibrary{arrows,shapes,chains}
\usepackage{hyperref}

\newtheorem{assumption}{Assumption}

\newtheorem{lemma}{Lemma}
\newtheorem{problem}{Problem}

\newtheorem{remark}{Remark}
\newtheorem{theorem}{Theorem}
\allowdisplaybreaks[4]

\begin{document}
\title{Leader-follower Attitude Synchronization of Rigid-body Systems on $\mathrm{SO}(3)$}
\author{Yiliang Li, Jun-e Feng, Abdelhamid Tayebi, \IEEEmembership{Fellow, IEEE}
\thanks{This work was supported in part  by the National Sciences and Engineering Research Council of Canada (NSERC), under the grant NSERC-DG RGPIN 2020-06270, in part by the National Natural Science Foundation of China, under the grant 12601872, 62673293, 62350037, in part by the Research Fund for the Taishan Scholar Project of Shandong Province of China under the grant tstp20221103, in part by the Science Center Program of National Natural Science Foundation of China under Grant 62188101, in part by the China Postdoctoral Science Foundation under grant 2025M783121, and in part by Shandong Postdoctoral Science Foundation under grant SDZZ-ZR-202501341.
}
\thanks{A preliminary version of
this work has been submitted to IFAC World Congress 2026 \cite{Liyiliang2026}.}
\thanks{Yiliang Li, Jun-e Feng and Abdelhamid Tayebi are with School of Mathematics, Shandong University, Jinan, Shandong 250100, P.R. China (e-mail: liyiliang1994@126.com, fengjune@sdu.edu.cn). }
\thanks{Abdelhamid Tayebi is also with Department of Electrical Engineering, Lakehead University, Thunder Bay, Ontario, Canada (e-mail: atayebi@lakeheadu.ca).}}

\maketitle

\begin{abstract}
This paper addresses the leader-follower attitude synchronization problem on $\mathrm{SO}(3)$ for a group of heterogeneous rigid body systems.  The reference attitude, represented by a virtual leader, is accessible only to a subset of agents in the network. The follower communication graph is assumed to be undirected and acyclic, and every agent is connected to the virtual leader through a path in the corresponding augmented graph (including the virtual leader). An observer-based distributed control strategy, endowed with almost global asymptotic stability guarantees, is proposed to synchronize all rigid-body attitudes with a desired time-varying reference attitude. An observer-based distributed control, with reduced complexity, as well an observerless distributed control strategy are also developed for the constant-reference case, with almost global asymptotic stability guarantees.
Numerical simulations are presented to demonstrate the effectiveness and performance of the proposed distributed control strategies.
\end{abstract}


\section{Introduction}

The attitude synchronization of multiple rigid-body systems is an important problem from both theoretical and practical perspectives and has generated significant research interest over the past two decades \cite{abdessameud2013motion}. It plays a fundamental role in a wide range of applications involving networks of autonomous vehicles, where coordinated orientation is essential for achieving cooperative behavior, reliable sensing, and mission-level objectives.
The attitude synchronization problem for multiple rigid-body systems is commonly classified into two main categories: the leaderless and leader-follower formulations. In the leaderless case, the agents seek to reach a consensus on a common attitude through mutual coordination, whereas in the leader-follower case, the group is required to synchronize to a prescribed reference attitude (represented by a virtual leader) that is available to one agent or a subset of agents. Both formulations require the design of \textit{distributed control laws} that rely solely on local information exchange between neighboring agents over a communication graph.

Most of the existing attitude synchronization strategies for a group of rigid body systems, including both leaderless and leader-follower scenarios, were developed using the Euler angles \cite{zhao2021data,guo2022distributed}, the modified Rodrigues parameters \cite{tang2022event,ren2009distributed,jin2025event,wei2018finite,jin2021event,dimarogonas2009leader,meng2010distributed,zou2016distributed} and the unit quaternion \cite{abdessameud2009attitude,abdessameud2012attitude, huang2021global,he2025leader,cai2016leader}.
The three aforementioned attitude representations have been extensively used in the literature because they enable classical synchronization techniques developed for multi-agent systems in Euclidean spaces to be adapted to attitude synchronization problems.
However, all these representations fail to provide a unique and global representation of the attitude and most of the developed attitude synchronization strategies, relying on these representations, are local and/or with limited stability guarantees. The rotation matrix on the special orthogonal group $\mathrm{SO}(3)$ is the only attitude representation that is both global and unique.
Unfortunately, synchronization algorithms developed in Euclidean spaces cannot be directly extended to the compact matrix Lie group $\mathrm{SO}(3)$, since such approaches typically rely on linear operations that are not suitable on curved manifolds.
Consequently, attitude synchronization on $\mathrm{SO}(3)$ necessitates dedicated control techniques that respect the geometric properties of the manifold.
In addition, the topological properties of $\mathrm{SO}(3)$ preclude the existence of any smooth feedback control law that stabilizes globally the attitude \cite{koditschek1989application,Bhat_SCL2000}.
These inherent geometric and topological constraints introduce fundamental complexity to the attitude synchronization on $\mathrm{SO}(3)$, thereby motivating the design of specialized control laws that explicitly account for the intrinsic structure of $\mathrm{SO}(3)$.
Attitude synchronization schemes  on $\mathrm{SO}(3)$, with local convergence results, have been proposed in \cite{sarlette2009autonomous,thunberg2016consensus,deng2021attitude,peng2018specified,zou2018velocity}. For the leaderless synchronization problem on $\mathrm{SO}(3)$, almost global synchronization strategies were reported in \cite{Tron_TAC2012,markdahl2021synchronization}, while hybrid feedback techniques enabling global asymptotic synchronization on $\mathrm{SO}(3)$ were developed in \cite{Bou_Ber_Tayebi_TAC2026}.
In the preliminary conference version of this work \cite{Liyiliang2026}, we proposed a distributed control law endowed with almost global asymptotic stability guarantees, allowing to synchronize the attitude of all rigid body systems to a constant reference attitude accessible solely to a single agent.
Nevertheless, the existing leader-follower attitude synchronization strategies on $\mathrm{SO}(3)$ for a time-varying reference attitude merely yield local convergence, and fail to provide rigorous stability guarantees.

In the present paper, we address the leader-follower attitude synchronization problem for a network of rigid-body systems on $\mathrm{SO}(3)$, where the reference attitude, described by a virtual leader, is either constant or time-varying. The reference attitude is only available to a subset of agents in the network, and the local information exchange is carried according to an undirected and acyclic graph (without the leader), wherein each vertex is connected to the leader via a path in the augmented graph including the virtual leader vertex.
The main contributions of this work are summarized as follows:

\begin{enumerate}
    \item We propose an observer-based distributed control strategy on $\mathrm{SO}(3)$, that synchronizes the attitude of all rigid body systems to constant and time-varying reference attitudes with almost global asymptotic stability guarantees\footnote{An equilibrium point is almost globally asymptotically stable if it is stable and attractive from all initial conditions except from a set of zero Lebesgue measure.}. To the best of our knowledge, there is no work in the available literature that achieves such strong stability results for the leader-follower problem on $\mathrm{SO}(3)$ with time-varying reference trajectory.
    \item We propose an observerless distributed control law endowed with almost global asymptotic stability guarantees, driving the attitude of all rigid body systems to align with a constant reference attitude.
\end{enumerate}

The remainder of this paper is organized as follows.
Section \ref{sec2} presents notations and fundamental graph theory concepts adopted throughout this paper.
The problem formulation is provided in Section \ref{sec3}, and two useful lemmas related to the interconnection graph are given in Section \ref{lemmas}.
The attitude synchronization problems subject to time-varying and constant reference attitudes are addressed in Section \ref{timevarying} and Section \ref{constant}, respectively.
Section \ref{simulation} demonstrates the performance of the proposed distributed control law via some simulations.
Finally, concluding remarks are drawn in Section \ref{conclusion}.

\section{Preliminaries}\label{sec2}

\subsection{Notation}
The sets of real numbers, $n$-dimensional vectors and $n$-by-$m$ matrices are denoted by $\mathbb{R},\mathbb{R}^{n}$ and $\mathbb{R}^{n\times m}$, respectively.
The set of unit vectors in $\mathbb{R}^n$ is defined as $\mathbb{S}^{n-1}=\{x\in\mathbb{R}^n|\|x\|=1\}$.
The $i$-th row of a given matrix $M\in\mathbb{R}^{n\times m}$ is denoted by $\mathrm{Row}_i(M)$.
We define the Euclidean norm of $x\in\mathbb{R}^{n}$ as $\|x\|=\sqrt{x^{\top}x}$.
The Euclidean inner product of $A,B\in\mathbb{R}^{n\times m}$ is defined as $\langle\langle A,B\rangle\rangle=\mathbf{tr}(A^{\top}B)$, and the Frobenius norm of $A\in\mathbb{R}^{n\times m}$ is defined as $\|A\|_F=\sqrt{\mathbf{tr}(A^{\top}A)}$.
The $n$-dimensional null vector and the $n$-by-$m$ null matrix are denoted as $\mathbf{0}_n$ and $\mathbf{0}_{n\times m}$, respectively.
An $n$-dimensional column vector with each element being $1$ is denoted as $\mathbf{1}_n$.
An $n$-dimensional identity matrix is denoted as $I_n$.
The Kronecker product is denoted by $\otimes$.

The special orthogonal group is defined as $\mathrm{SO}(3)=\{R\in\mathbb{R}^{3\times 3}|R^{\top}R=RR^{\top}=I_3,\mathrm{det}(R)=1\}$, where $\mathrm{det}(R)$ stands for the determinant of $R$.
The \textit{Lie\ algebra} of $\mathrm{SO}(3)$ is $\mathfrak{so}(3)=\{\Omega\in\mathbb{R}^{3\times 3}|\Omega=-\Omega^{\top}\}$. We define the map $(\cdot)^{\times}:\mathbb{R}^3\rightarrow\mathfrak{so}(3)$ such that $x^{\times}y=x\times y$ for any $x,y\in\mathbb{R}^3$, where $\times$ denotes the vector cross-product on $\mathbb{R}^3$. The matrix $x^{\times}$ is a $3$-by-$3$ skew-symmetric matrix, generated by $x=[x_1\ x_2\ x_3]^{\top}$, as follows:
\[
x^{\times}=\left[
\begin{array}{ccc}
0&-x_3&x_2\\
x_3&0&-x_1\\
-x_2&x_1&0\\
\end{array}
\right]_.
\]
Let $\mathbf{vex}:\mathfrak{so}(3)\rightarrow\mathbb{R}^3$ be a map satisfying $(\mathbf{vex}(\Omega))^{\times}=\Omega, \forall\Omega\in\mathfrak{so}(3)$ and $\mathbf{vex}(x^{\times})=x,\forall x\in\mathbb{R}^3$.
The projection map $\mathbb{P}_a:\mathbb{R}^{3\times 3}\rightarrow\mathfrak{so}(3)$ on the Lie algebra $\mathfrak{so}(3)$ is defined by $\mathbb{P}_a(A):=\frac{1}{2}(A-A^{\top})$.
The composition map $\psi:=\mathbf{vex}\circ\mathbb{P}_a$ is defined as $\psi(A)=\mathbf{vex}(\mathbb{P}_a(A))=\frac{1}{2}(A-A^{\top})=\frac{1}{2}[a_{32}-a_{23}\ a_{13}-a_{31}\ a_{21}-a_{12}]^{\top}$ for $A=[a_{i,j}]\in\mathbb{R}^{3\times 3}$.
The normalized attitude norm on $\mathrm{SO}(3)$ is defined as $|R|_I=\frac{1}{2}\sqrt{\mathbf{tr}(I_3-R)}$.
For $A\in\mathbb{R}^{3\times 3}$, we define the matrix $\mathbf{E}(A):=\frac{1}{2}(\mathbf{tr}(A)I_3-A^{\top})$.
We denote by \(\mathcal{E}_v(A)\) the set containing exactly one representative from each antipodal pair of unit eigenvectors of a symmetric matrix $A$. Throughout this paper, for notational convenience and simplicity, the time argument is omitted whenever no ambiguity arises.
\subsection{Graph Theory}
We formally define a graph as $\mathcal{G}=(\mathcal{V},\mathcal{E})$, where $\mathcal{V}=\{1,\ldots,N\}$ is a set of vertices, $\mathcal{E}\subset\mathcal{V}\times\mathcal{V}$ is a set of edges of $\mathcal{G}$ with an edge between vertices $i,j$ to be denoted as $(i,j)$.
The neighborhood $\mathcal{N}_{\mathcal{G}}(i)$ of the vertex $i$ in the graph $\mathcal{G}$ is defined as $\mathcal{N}_{\mathcal{G}}(i)=\{j\in\mathcal{V}|(j,i)\in\mathcal{E}\}$.
A sequence of distinct vertices $(i_1,i_2,\ldots,i_m)$ is said to be a path of length $m$ in $\mathcal{G}$ if $(i_k,i_{k+1})\in\mathcal{E}$ for $k\in\{1,\ldots,m-1\}$.
A graph $\mathcal{G}$ is undirected if for any two vertices $i,j\in\mathcal{V}$, $(i,j)\in\mathcal{E}$ implies $(j,i)\in\mathcal{E}$.
If a path $(i_1,\ldots,i_m)$ in the graph $\mathcal{G}$ satisfies $i_1=i_m$, then $(i_1,\ldots,i_m)$ is called a cycle.
A graph is said to be acyclic if it contains no cycles.

\section{Problem formulation}\label{sec3}

Consider a group of $N$ heterogeneous\footnote{The rigid body systems do not necessarily have the same inertia matrix.} rigid body systems, the rotational dynamics of which is as follows:
\begin{equation}\label{followeri}
\left\{
\begin{aligned}
\dot{R}_i=&R_i\omega_i^{\times},\\
J_i\dot{\omega}_i=&-\omega_i\times J_i\omega_i+\tau_i,\\
\end{aligned}
\right.
\end{equation}
for $i=1,\ldots,N$,
where $R_i\in \mathrm{SO}(3)$ is the attitude of the $i$-th rigid body, $\omega_i\in\mathbb{R}^3$ is the  angular velocity of the $i$-th rigid body with respect to the inertial frame expressed in the body-attached frame, $J_i\in\mathbb{R}^{3\times 3}$ is the symmetric positive definite inertia matrix of the $i$-th rigid body, and $\tau_i\in\mathbb{R}^3$ is the control torque of the $i$-th rigid body.
Let the interconnection between the rigid body systems be described by an undirected graph $\mathcal{G}=(\mathcal{V},\mathcal{E})$, where $\mathcal{V}=\{1,\ldots,N\}$ stands for the set of rigid body systems, $\mathcal{E}\subset\mathcal{V}\times\mathcal{V}$ represents the edge set with $(i,j)\in\mathcal{E}$ indicating that the $j$-th rigid body receives the information from the $i$-th rigid body.

Note that for any two vertices $i,j$ in an undirected graph $\mathcal{G}$, $(i,j)\in\mathcal{E}$ implies $(j,i)\in\mathcal{E}$, and vice versa. Since the convergence of $R_iR_j^{\top}$ (respectively $R^{\top}_iR_j$) implies the convergence of $R_jR_i^{\top}$ (respectively  $R_j^{\top}R_i$), the stability analysis could be simplified by specifying an arbitrary virtual direction for any two interconnected agents in the graph $\mathcal{G}$ (see, for instance, \cite{Arcak2008,Bough_Tay_2025}).
Let $\overline{\mathcal{E}}$ denote the edge set obtained by arbitrarily removing one of the edges between any two interconnected agents, and let  $\overline{\mathcal{G}}=(\mathcal{V},\overline{\mathcal{E}})$ denote the graph with a virtual orientation specified in $\mathcal{G}$.

Our objective is to design a control torque $\tau_i$ such that the attitude of the $i$-th rigid body is synchronized to the desired attitude $R_0(t)$.
The desired attitude $R_0$ is represented by a virtual leader $0$, and the rigid body systems $i=1,2,\ldots,N$ are regarded as followers.
We assume that the desired attitude $R_0$ is available to the rigid body systems $i=1,\ldots,n$, where $n<N$. Let $\mathcal{V}_0=\mathcal{V}\cup\{0\}$, $\mathcal{E}_0=\mathcal{E}\cup\{(0,i)\subset\mathcal{V}_0\times\mathcal{V}_0|i=1,\ldots,n\}$ and $\overline{\mathcal{E}}_0=\overline{\mathcal{E}}\cup\{(0,i)\subset\mathcal{V}_0\times\mathcal{V}_0|i=1,\ldots,n\}$. Let
 $\mathcal{G}_0=(\mathcal{V}_0,\mathcal{E}_0)$ denote the augmented graph including the leader $0$ and let $\overline{\mathcal{G}}_0=(\mathcal{V}_0,\overline{\mathcal{E}}_0)$ be the augmented oriented graph.

 In this paper, the interconnection graph $\mathcal{G}$ and the augmented graph $\mathcal{G}_0$ must satisfy the following assumption.
\begin{assumption}\label{assumpgraph_original}
The interconnection graph $\mathcal{G}$ is  undirected and acyclic, and each vertex of the graph $\mathcal{G}$ is reachable from the leader $0$ through a path in the augmented graph $\mathcal{G}_0$.
\end{assumption}

Assumption \ref{assumpgraph_original} ensures that the augmented oriented graph $\overline{\mathcal{G}}_0$ possesses exactly $N$ edges, \emph{i.e.}, $|\overline{\mathcal{E}}_0|=N$.


In this work we aim to solve the following problems:
\begin{problem}\label{pro1}
Consider a network of $N$ rigid body systems governed by the rotational dynamics in \eqref{followeri}.
The desired attitude and the desired angular velocity are generated by
\begin{equation}\label{leader}
\setlength{\arraycolsep}{1pt}
\begin{array}{rcl}
\dot{R}_0&=&R_0\omega_0^{\times},\\
\dot{\omega}_0&=&\sigma_0,\\
\dot{\sigma}_0&=&P\sigma_0,\\
\end{array}
\end{equation}
where the matrix $P\in\mathbb{R}^{3\times 3}$ has simple non-zero imaginary eigenvalues, and all other eigenvalues are with negative real parts.
These constraints on $P$ guarantee that the angular acceleration $\sigma_0$ and the angular velocity $\omega_0$ are bounded. Under Assumption \ref{assumpgraph_original}, we aim to design a distributed control torque $\tau_i$, $i\in\mathcal{V}$, ensuring almost global asymptotic synchronization of all attitudes of the rigid body systems to the time-varying reference attitude $R_0(t)$.
\end{problem}

\begin{problem}\label{pro1}
Consider a network of $N$ rigid body systems governed by the rotational dynamics in \eqref{followeri}.
Under Assumption \ref{assumpgraph_original}, we aim to design a distributed control torque $\tau_i$, $i\in\mathcal{V}$, such that all the attitudes of the rigid body systems synchronize to a constant reference attitude $R_0$, with almost global asymptotic stability.
\end{problem}

\section{Useful Lemmas}\label{lemmas}
In this section, we establish several useful lemmas related to the interconnection graph for subsequent stability analysis.

The first lemma is related to the Laplacian matrix of the interconnection graph $\mathcal{G}$.
Let  $\mathcal{M}_{\overline{\mathcal{G}}_0}=\{1,\ldots,N\}$ denote the set of edges in the graph $\overline{\mathcal{G}}_0$, and let $\mathcal{M}_{\overline{\mathcal{G}}}=\{n+1,n+2,\ldots,N\}$ denote the set of edges in the graph $\overline{\mathcal{G}}$.
Let $(\mathcal{M}_{\overline{\mathcal{G}}})_i^+\subset\mathcal{M}_{\overline{\mathcal{G}}}$ be the set of edges in the graph $\overline{\mathcal{G}}$ with $i$ being the head of the edges and $(\mathcal{M}_{\overline{\mathcal{G}}})_i^-\subset\mathcal{M}_{\overline{\mathcal{G}}}$ be the set of edges in the graph $\overline{\mathcal{G}}$ with $i$ being the tail of the edges.
We construct the matrix $L\in\mathbb{R}^{N\times (N-n)}$, according to the oriented graph $\overline{\mathcal{G}}$, as follows:
\begin{equation}\label{LG}
L_{ik}=\left\{
\begin{array}{ll}
1, & n+k\in(\mathcal{M}_{\overline{\mathcal{G}}})_i^-, \\
-1,     & n+k\in(\mathcal{M}_{\overline{\mathcal{G}}})_i^+,\\
0,&\mathrm{otherwise}.
\end{array}
\right.
\end{equation}
It is clear that $LL^{\top}$ is the Laplacian matrix of the graph $\mathcal{G}$.

\begin{lemma}\label{le1}
Let Assumption \ref{assumpgraph_original} hold and let $\Xi=\mathrm{diag}(\Xi_1,\ldots,\Xi_N)$ with
\begin{equation}\label{alpha}
\Xi_i=\left\{
\begin{array}{ll}
  1,   & i\in\{1,\ldots,n\}, \\
  0,   & i\in\{n+1,\ldots,N\}.
\end{array}
\right.
\end{equation}
Then, the matrix $\Xi+LL^{\top}$ is symmetric positive definite.
\end{lemma}

\begin{proof}
See Appendix \ref{appenA}.
\end{proof}

We consider the following left-invariant and right-invariant attitude errors,
\[
\check{R}_{i,j}^l=R_iR_j^{\top},\quad\check{R}^r_{i,j}=R_j^{\top}R_i.
\]
Depending on which attitude error is used, we define the matrices $\check{\mathcal{L}}^l,\check{\mathcal{L}}^r\in\mathbb{R}^{3(N+1)\times 3N}$ with their respective $3$-by-$3$ blocks $(\check{\mathcal{L}}^l)_{ik},(\check{\mathcal{L}}^r)_{ik}$ defined as follows:
\begin{equation}\label{checkL}
(\check{\mathcal{L}}^l)_{ik}=\left\{
\begin{array}{ll}
 -I_3,    & k\in(\mathcal{M}_{\overline{\mathcal{G}}_0})_{i-1}^+,  \\
  \check{R}_k^l,   & k\in(\mathcal{M}_{\overline{\mathcal{G}}_0})_{i-1}^-,\\
  \mathbf{0}_{3\times 3},&\mathrm{otherwise}.
\end{array}
\right.
\end{equation}
\begin{equation}\label{checkLe}
(\check{\mathcal{L}}^r)_{ik}=\left\{
\begin{array}{ll}
-\check{R}^r_k, & k\in(\mathcal{M}_{\overline{\mathcal{G}}_0})^{+}_{i-1},  \\
I_3, & k\in(\mathcal{M}_{\overline{G}_0})^{-}_{i-1},\\
\mathbf{0}_{3\times 3}, &\mathrm{otherwise}.\\
\end{array}
\right.
\end{equation}
where $(\mathcal{M}_{\overline{\mathcal{G}}_0})_i^+\subset\mathcal{M}_{\overline{\mathcal{G}}_0}$ and $(\mathcal{M}_{\overline{\mathcal{G}}_0})_i^-\subset\mathcal{M}_{\overline{\mathcal{G}}_0}$ are the sets of edges in the graph $\overline{\mathcal{G}}_0$ with $i$ being the head and the tail of the edges respectively.

Now, we show that the matrices, obtained from $\check{\mathcal{L}}^l$ and $\check{\mathcal{L}}^r$ by removing their respective first three rows, are non-singular.

\begin{lemma}\label{le2}
Consider the matrices $\check{\mathcal{L}}^l,\check{\mathcal{L}}^r\in\mathbb{R}^{3(N+1)\times 3N}$, associated to the augmented graph $\mathcal{G}_0$, with their $3$-by-$3$ blocks $(\check{\mathcal{L}}^l)_{ik},(\check{\mathcal{L}}^r)_{ik}$ defined in \eqref{checkL} and \eqref{checkLe} respectively.
Under Assumption \ref{assumpgraph_original}, $\check{\mathcal{L}}_2^l$ and $\check{\mathcal{L}}^r_2$ are non-singular, where
\[
\check{\mathcal{L}}_2^l=\left[
\begin{array}{cccc}
(\check{\mathcal{L}}^l)_{21}&(\check{\mathcal{L}}^l)_{22}&\cdots&(\check{\mathcal{L}}^l)_{2N}\\
(\check{\mathcal{L}}^l)_{31}&(\check{\mathcal{L}}^l)_{32}&\cdots&(\check{\mathcal{L}}^l)_{3N}\\
\vdots&\vdots&\ddots&\vdots\\
(\check{\mathcal{L}}^l)_{N+1,1}&(\check{\mathcal{L}}^l)_{N+1,2}&\cdots&(\check{\mathcal{L}}^l)_{N+1,N}
\end{array}
\right]_,
\]
\[
\check{\mathcal{L}}_2^r=\left[
\begin{array}{cccc}
(\check{\mathcal{L}}^r)_{21}&(\check{\mathcal{L}}^r)_{22}&\cdots&(\check{\mathcal{L}}^r)_{2N}\\
(\check{\mathcal{L}}^r)_{31}&(\check{\mathcal{L}}^r)_{32}&\cdots&(\check{\mathcal{L}}^r)_{3N}\\
\vdots&\vdots&\ddots&\vdots\\
(\check{\mathcal{L}}^r)_{N+1,1}&(\check{\mathcal{L}}^r)_{N+1,2}&\cdots&(\check{\mathcal{L}}^r)_{N+1,N}
\end{array}
\right]_.
\]
\end{lemma}

\begin{proof}
See Appendix \ref{appenB}.
\end{proof}

\section{Observer-based Attitude Synchronization to A Time-varying Reference Trajectory}\label{timevarying}

This section is devoted to the design of a distributed control scheme that allows to synchronize the attitudes of the rigid body systems, governed by \eqref{followeri}, to a time-varying reference attitude, \textit{i.e.}, to solve \textit{Problem 1}.
It is worth mentioning that each agent exchanges information with its neighboring agents over the communication graph, and some agents have no direct access to the time-varying reference attitude.
This requires each agent to estimate the time-varying reference attitude and subsequently track its local estimate, so as to achieve the attitude synchronization to a time-varying reference attitude.
Hence, we propose the following observer-based distributed control algorithm:
\begin{enumerate}
    \item A distributed observer is assigned to each rigid body to estimate the desired attitude, the desired angular velocity and the desired angular acceleration.
    \item A distributed control law is formulated to enable each rigid body to track its local estimates.
\end{enumerate}

Figure \ref{observerbasedtimevarying} depicts the block diagram of the proposed observer-based control law \eqref{control}, incorporating the distributed observer \eqref{observer}, to achieve the attitude synchronization to a time-varying reference trajectory.

\begin{figure}[!htbp]
\footnotesize
\centering
\tikzstyle{test}=[diamond,aspect=2,draw,thin]
\tikzstyle{line} = [draw, ->]
\tikzstyle{point}=[coordinate,on grid,]
\begin{tikzpicture}
\node[rectangle, minimum width=8.5cm, minimum height=2cm, draw = black,fill=orange!5] (outer1) {\shortstack[1]{\ \ \ \ \ \ \ \ \ \ \ \ \ \ \ \ \ \ \ \ \ \ \ \ \ \ \ \ \ \ \ \ \ \ \ \ \ \ \ \ \ \ \ \ \ \ \ \ \ \ \ Rigid body $i$\\ \ \\ \ \\ \ \\ \ \\ \ \\ \ \\ \ \\ \ \\ \ \\ \ \\ \ \\ \ \\ \ \\ \ \\ \ \\ \ \\ \ \\ \ \\ \ \\ \ \\ \ \\ \ \\ \ \\ \ \\ \ \\ \ \\ \ \\ \ \\ \ \\ \ \\ \ \\ \ \\ \ \\ \ \\ \ \\ \ \\ \ \\ \ \\ \ \\ \ \\ \ \\ \ \\ \ \\ \ \\ \ }};
\node[rectangle, minimum width=5cm, minimum height=2cm,xshift=-1.6cm, yshift=0.2cm, draw = black,fill=pink!10] (outer3) {\shortstack[1]{\\ \ \\ \ \\ \ \\ \ \\ \ \\ \ \\ \ \\ \ \\ \ \\ \ \\ \ \\ \ \\ \ \\ \ \\ \ \\ \ \\ \ \\ \ \\ \ \\ \ \\ \ \\ \ \\ \ \\ \ \\ \ \\ \ \\ \ \\ \ \\ \ \\ \ \\ \ \\ \ \\ \ \\ \ \\ \ \\ \ \\ \ \\ \ \\ \ \\ \ \\ \ \\ \ \\ \ \\ \ \\ \ \\ \ \\ \ \\ \ \\ \ \\ \ \\ \ \\ \ \\ \ \\ \ \\ \ \\ \ \\ \ \\ \ \\ \ \\ \ \\ \ \\ \ \\ \ \\ \ \\ \ \\ \ \\ \ \\ \ \\ \ \\ \ \\ \ }};
\node[rectangle, minimum width=4.6cm, minimum height=1cm,   xshift=-1.6cm, yshift=0.2cm, draw = green!10,fill=green!10] (mid1) {\shortstack[1]{Estimation\\ \ \\ \ \\ \ \\ \ \\ \ \\ \ \\ \ \\ \ \\ \ \\ \ \\ \ \\ \ \\ \ \\ \ \\ \ \\ \ \\ \ \\ \ \\ \ \\ \ \\ \ \\ \ \\ \ \\ \ \\ \ \\ \ \\ \ \\ \ \\ \ \\ \ \\ \ \\ \ \\ \ \\ \ \\ \ \\ \ \\ \ \\ \ }};
\node[rectangle,minimum width=1.2cm,minimum height=0.5cm,xshift=-2.9cm,yshift=1.4cm,draw=yellow!20,fill=yellow!20] (inner11) {\shortstack[1]{Angular\\ acceleration\\ observer \eqref{7a}}};
\node[rectangle,minimum width=1.2cm,minimum height=0.5cm,xshift=-1.9cm,yshift=0cm,draw=yellow!20,fill=yellow!20](inner12) {\shortstack[1]{Angular\\ velocity\\ observer \eqref{7b}}};
\node[point,right of=inner11,node distance=1.1cm](pr2inner12){};
\node[point,below of=pr2inner12,node distance=0.81cm](prb2inner12){};
\draw[line] (pr2inner12)--(prb2inner12);
\node[rectangle,minimum width=1.2cm,minimum height=0.5cm,xshift=-0.9cm,yshift=-1.2cm,draw=yellow!20,fill=yellow!20](inner13) {\shortstack[1]{Attitude\\ observer \eqref{7c}}};
\node[point,right of=inner13,node distance=1.5cm](prinner13){};
\draw[-] (inner13)--node[above]{$\hat{R}_i$}(prinner13);
\node[point,right of =inner12,node distance=1.1cm](pr2inner12){};
\node[point,below of=pr2inner12,node distance=0.81cm](prb2inner12){};
\draw[line] (pr2inner12)--(prb2inner12);
\node[point,right of =inner13,node distance=1.1cm](pr2inner13){};
\node[point,below of=pr2inner13,node distance=0.6cm](prb2inner13){};
\draw[-] (pr2inner13)--(prb2inner13);
\node[point,left of=prb2inner13,node distance=2.2cm](prbl2inner13){};
\draw[-] (prb2inner13)--(prbl2inner13);
\node[point,above of=prbl2inner13,node distance=1.2cm](prbla2inner13){};
\draw[line] (prbl2inner13)--(prbla2inner13);
\node[rectangle, minimum width=2.9cm, minimum height=1cm,   xshift=2.65cm, yshift=-0.2cm, draw = blue!10,fill=blue!10] (mid2) {\shortstack[1]{Distributed control\\ \ \\ \ \\ \ \\ \ \\ \ \\ \ \\ \ \\ \ \\ \ \\ \ \\ \ \\ \ \\ \ \\ \ \\ \ \\ \ \\ \ \\ \ \\ \ \\ \ \\ \ \\ \ \\ \ \\ \ \\ \ \\ \ \\ \ \\ \ \\ \ \\ \ \\ \ \\ \ \\ \ \\ \ \\ \ \\ \ \\ \ }};
\node[rectangle,minimum width=1.2cm,minimum height=0.5cm,xshift=2.5cm,yshift=0cm,draw=brown!50,fill=brown!50] (inner2){\shortstack[1]{Distributed \\control \eqref{control}}};
\node[point,right of=inner11,node distance=5cm](prinner11){};
\draw[-] (inner11)--node[above]{$\hat{\sigma}_i$}(prinner11);
\node[point,below of=prinner11,node distance=1.02cm](prrbinner11){};
\draw[line] (prinner11)--(prrbinner11);
\node[point,right of=inner12,node distance=3.6cm](prinner12){};
\draw[line] (inner12)--node[above]{$\hat{\omega}_i$}(prinner12);
\node[point,right of=prinner13,node distance=1.5cm](prinner132){};
\draw[-] (prinner13)--(prinner132);
\node[point,above of=prinner132,node distance=0.81cm](prainner132){};
\draw[line] (prinner132)--(prainner132);
\node[point,right of=prainner132,node distance=0.9cm](prinner12r){};
\node[rectangle,minimum width=0.3cm,minimum height=0.3cm,xshift=3cm,yshift=-1.8cm,draw=blue!10,fill=blue!10] (R1){$R_i$};
\draw[line] (R1)--(prinner12r);
\node[rectangle,minimum width=0.3cm,minimum height=0.3cm,xshift=3cm,yshift=1.1cm,draw=blue!10,fill=blue!10] (omega1){$\omega_i$};
\node[point,above of=prinner12r,node distance=0.77cm](prinner12ra){};
\draw[line] (omega1)--(prinner12ra);
\node[rectangle,minimum width=0.3cm,minimum height=0.3cm,xshift=3.8cm,yshift=0cm,draw=blue!10,fill=blue!10] (tau1){$\tau_i$};
\draw[line] (inner2)--(tau1);
\node[rectangle, minimum width=4.6cm, minimum height=0.8cm, xshift=-1.6cm, yshift=3.5cm,  draw = green!10,fill=green!10] (outer2) {The interconnected graph $\mathcal{G}$};
\node[point,left of=outer2,node distance=0.5cm](prouter21){};
\node[point,below of=prouter21,node distance=1.1cm](prbouter21){};
\node[point,below of=prouter21,node distance=0.4cm](prbouter22){};
\draw[line] (prbouter22)--node[right]{\shortstack[1]{$\hat{R}_j,\hat{\omega}_j,\hat{\sigma}_j,j\in\mathcal{N}_{\mathcal{G}(i)}$}}(prbouter21);
\node[point,left of=outer2,node distance=0.8cm](prouter21){};
\node[point,below of=prouter21,node distance=1.1cm](prbouter21){};
\node[point,below of=prouter21,node distance=0.4cm](prbouter22){};
\draw[line] (prbouter21)--node[left]{\shortstack[1]{$\hat{R}_i,\hat{\omega}_i,\hat{\sigma}_i$}}(prbouter22);
\node[rectangle, minimum width=4.6cm, minimum height=0.8cm, xshift=-1.6cm, yshift=-3.1cm,  draw = green!10,fill=green!10] (leader) {\shortstack[1]{Time-varying \\reference trajectory \eqref{leader}}};
\draw[line] (leader)--node[left]{$R_0,\omega_0,\sigma_0$}node[right]{if $(0,i)\in\mathcal{E}_0$}(mid1);
\end{tikzpicture}
\caption{The implementation block diagram of the proposed observer-based control law \eqref{control} incorporating the distributed observer \eqref{observer}}\label{observerbasedtimevarying}
\end{figure}
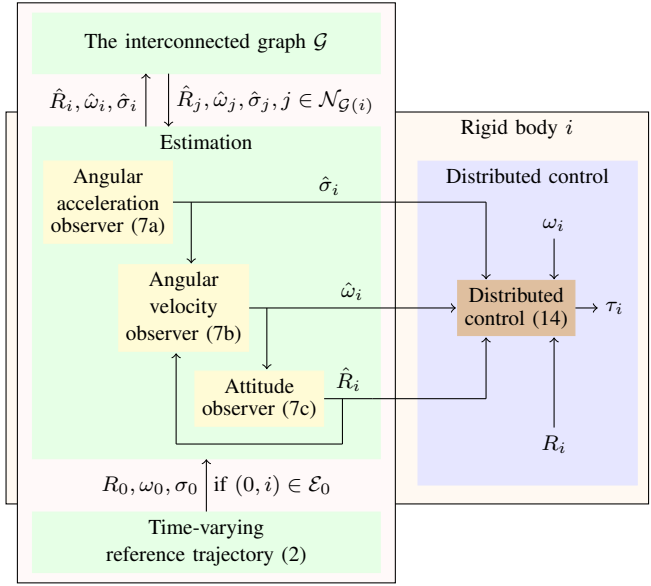

\subsection{Distributed observer design}\label{sec5a}

In this subsection, we propose a distributed observer assigned to each rigid body, and derive the dynamics of the estimation errors.

Let $\hat{R}_i\in \mathrm{SO}(3)$,  $\hat{\omega}_i\in\mathbb{R}^3$, and $\hat{\sigma}_i\in\mathbb{R}^3$ denote the $i$-th rigid body's estimates of the desired attitude $R_0$, the desired angular velocity $\omega_0$, and the desired angular acceleration $\sigma_0$, respectively, where $i\in\mathcal{V}$.
We propose the following distributed observer:
\begin{subequations}\label{observer}
\begin{align}
\dot{\hat{R}}_i=&\hat{R}_i\hat{\omega}_i^{\times},\label{7a}\\
\dot{\hat{\omega}}_i=&\hat{\sigma}_i-k_{R^l}R_0^{\top}\psi(A_{i,0}^l\hat{R}_iR_0^{\top})-k_{\omega^l}\Xi_i(\hat{\omega}_i-\omega_0)\label{7b}\\
 & -k_{R^l}\sum\limits_{j\in\mathcal{N}_{\mathcal{G}}(i)}\hat{R}_j^{\top}\psi(A_{i,j}^l\hat{R}_{i}\hat{R}_j^{\top})-k_{\omega^l}\sum\limits_{j\in\mathcal{N}_{\mathcal{G}}(i)}(\hat{\omega}_i-\hat{\omega}_j),\nonumber\\
\dot{\hat{\sigma}}_i=&
P\hat{\sigma}_i-k_{\tilde{\sigma}}\sum\limits_{j\in\mathcal{N}_{\mathcal{G}}(i)}(\hat{\sigma}_i-\hat{\sigma}_j)-k_{\tilde{\sigma}}\Xi_i(\hat{\sigma}_i-\sigma_0),\label{7c}
\end{align}
\end{subequations}
for $i\in\mathcal{V}$,
where $k_{R^l},k_{\omega^l},k_{\tilde{\sigma}}>0$, $A_{i,j}^l=(A_{i,j}^l)^{\top}>0$ with three distinct eigenvalues and $A_{i,j}^l=A_{j,i}^l$, $A_{i,0}^l=(A_{i,0}^l)^{\top}>0$ with three distinct eigenvalues, $A_{i,0}^l=\mathbf{0}_{3\times 3},\forall i\in\mathcal{V}\setminus\{1,\ldots,n\}$.

The proposed distributed observer \eqref{7a} for the desired attitude reveals that the desired attitude estimate is corrected by updating  the desired angular velocity estimate.
The proposed distributed observer \eqref{7b} for the desired angular velocity is composed of four innovation terms.
The first two terms are activated solely for the rigid body systems with access to the leader's information, driving the attitude and angular velocity estimates toward their respective desired values.
The remaining two terms leverage relative attitude and relative angular velocity information, ensuring that the estimates of all rigid body systems converge to a common attitude.
The proposed distributed observer \eqref{7c} for the desired angular acceleration consists of two innovation terms. One is a relative angular acceleration term that aligns the angular acceleration estimate to a common value, whereas the other is an absolute angular acceleration term activated only for the rigid body systems with access to the leader's information, steering the angular acceleration toward the desired value.

The dynamics of estimation errors are derived as follows.
First, we characterize the dynamics of the absolute estimation errors.
For $i\in\mathcal{V}$, let us define the absolute attitude error, absolute angular velocity error, and absolute angular acceleration error between the $i$-th observer and the leader $0$ as $\tilde{R}_{i,0}^l=\hat{R}_iR_0^{\top}$,  $\tilde{\omega}_{i,0}^l=\hat{\omega}_i-\omega_0$, and $\tilde{\sigma}_{i,0}=\hat{\sigma}_i-\sigma_0$, respectively.
Combining the leader dynamics \eqref{leader} with the proposed distributed observer \eqref{observer}, we obtain the absolute estimation error system
\begin{equation}\label{absoluteerror}
\setlength{\arraycolsep}{1pt}
\begin{array}{rcl}
\dot{\tilde{R}}_{i,0}^l&=&\tilde{R}_{i,0}^l(R_0\tilde{\omega}_{i,0}^l)^{\times},\\
\dot{\tilde{\omega}}_{i,0}^l&=&\tilde{\sigma}_{i,0}-k_{R^l}\sum\limits_{j\in\mathcal{N}_{\mathcal{G}}(i)}\hat{R}_j^{\top}\psi(A_{i,j}^l\hat{R}_{i}\hat{R}_j^{\top})-k_{\omega^l}\Xi_i\tilde{\omega}_{i,0}^l\\
\ &\ &-k_{R^l}R_0^{\top}\psi(A_{i,0}^l\tilde{R}_{i,0}^l)-k_{\omega^l}\sum\limits_{j\in\mathcal{N}_{\mathcal{G}}(i)}(\hat{\omega}_i-\hat{\omega}_j),\\
\dot{\tilde{\sigma}}_{i,0}&=&P\tilde{\sigma}_{i,0}
-k_{\tilde{\sigma}}\sum\limits_{j\in\mathcal{N}_{\mathcal{G}}(i)}(\hat{\sigma}_i-\hat{\sigma}_j)-k_{\tilde{\sigma}}\Xi_i\tilde{\sigma}_{i,0}.
\end{array}
\end{equation}



Next we derive the dynamics of relative attitude estimation errors.
The interconnection topology among the observers of rigid body systems coincides with that of the rigid body systems, which is depicted by the undirected graph $\mathcal{G}$.
Accordingly, the augmented graph $\mathcal{G}_0$ characterizes the interconnection topology of the observer network including the leader.
For all $i,j\in\mathcal{V}$, if $(j,i)\in\overline{\mathcal{E}}$, then
the relative attitude error between the observers $i,j$ is defined as $\tilde{R}_{i,j}^l=\hat{R}_i\hat{R}_j^{\top}$,
yielding $\dot{\tilde{R}}_{i,j}^l=\tilde{R}_{i,j}^l(\hat{R}_j\tilde{\omega}_{i,j}^l)^{\times}$, where $\tilde{\omega}_{i,j}^l=\hat{\omega}_i-\hat{\omega}_j$ stands for the relative angular velocity error.
In the resulting oriented graph $\overline{\mathcal{G}}_0$, if the edge $(j,i)\in\overline{\mathcal{E}}_0$ is labeled as the $k$-th edge, then we define  $\tilde{R}_k^l=\tilde{R}_{i,j}^l$, $\tilde{\omega}_k^l=\tilde{\omega}_{i,j}^l$.
Consequently, the dynamics of attitude estimation errors can be expressed in the following compact form:
\begin{equation}\label{attitudeerror}
\dot{\tilde{R}}^l=\tilde{R}^l\tilde{W}^l,
\end{equation}
where $\tilde{R}^l=\mathrm{diag}(\tilde{R}_1^l,\ldots,\tilde{R}_N^l)$, $\tilde{W}^l\in\mathbb{R}^{3N\times 3N}$ with each $3$-by-$3$ block $(\tilde{W}^l)_{kl}$ being defined as
\[
(\tilde{W}^l)_{kl}=\left\{
\begin{array}{ll}
(\hat{R}_j\tilde{\omega}_k^l)^{\times},&k=l\ \mathrm{and}\ k\in(\mathcal{M}_{\overline{\mathcal{G}}_0})_j^+,  \\
\mathbf{0}_{3\times 3},     &k\neq l,
\end{array}
\right.
\]
where $\hat{R}_0=R_0$. With the matrix $L$ defined in \eqref{LG}, one obtains
\begin{equation}\label{compactomegasigma}
\begin{array}{rcl}
\sum\limits_{j\in\mathcal{N}_{\mathcal{G}}(i)}(\hat{\omega}_i-\hat{\omega}_j)
&=&(\mathrm{Row}_i(LL^{\top})\otimes I_3)\tilde{\omega}_0^l,\\
\sum\limits_{j\in\mathcal{N}_{\mathcal{G}}(i)}(\hat{\sigma}_i-\hat{\sigma}_j)&=&(\mathrm{Row}_i(LL^{\top})\otimes I_3)\tilde{\sigma}_0,\\
\end{array}
\end{equation}
where $\tilde{\omega}_0^l=[(\tilde{\omega}_{1,0}^l)^{\top}\ \cdots\ (\tilde{\omega}_{N,0}^l)^{\top}]^{\top}$, $\tilde{\sigma}_0=[\tilde{\sigma}_{1,0}^{\top}\ \cdots\ \tilde{\sigma}_{N,0}^{\top}]^{\top}$.
As proven in Lemma \ref{le3} in Appendix \ref{appenC}, one has
\begin{align}\label{attitudecompact}
&\sum\limits_{j\in\mathcal{N}_{\mathcal{G}}(i)}\hat{R}_j^{\top}\psi(A_{i,j}^l\tilde{R}_{i,j}^l)+R_0^{\top}\psi(A_{i,0}^l\tilde{R}_{i,0}^l)\notag\\
&=\hat{R}_i^{\top}(\sum\limits_{k=1}^{N}(\tilde{\mathcal{L}}^l_2)_{i,k}\psi(A_k^l\tilde{R}_{k}^l)),
\end{align}
where $A_k^l=A_{i,j}^l$ for $i,j\in\mathcal{V}_0$ with the edge between $i$ and $j$ being described by edge $k$, $\tilde{\mathcal{L}}_2^l\in\mathbb{R}^{3N\times 3N}$ is obtained from $\tilde{\mathcal{L}}^l$ by removing the first three rows,  $\tilde{\mathcal{L}}^l\in\mathbb{R}^{3(N+1)\times 3N}$ is constructed, according to the oriented graph $\overline{\mathcal{G}}_0$, with each $3$-by-$3$ block $(\tilde{\mathcal{L}}^l)_{ik}$ being defined as
\begin{equation}\label{LG0}
(\tilde{\mathcal{L}}^l)_{ik}=\left\{
\begin{array}{ll}
 -I_3,    & k\in(\mathcal{M}_{\overline{\mathcal{G}}_0})_{i-1}^+,  \\
  \tilde{R}_k^l,   & k\in(\mathcal{M}_{\overline{\mathcal{G}}_0})_{i-1}^-,\\
  \mathbf{0}_{3\times 3},&\mathrm{otherwise}.
\end{array}
\right.
\end{equation}
Combining \eqref{compactomegasigma}-\eqref{attitudecompact}, the compact form of the dynamics of the estimation errors is derived as follows:
\begin{equation}\label{errorcompact}
\setlength{\arraycolsep}{1pt}
\begin{array}{rcl}
\dot{\tilde{R}}^l&=&\tilde{R}^l\tilde{W}^l,\\
\dot{\tilde{\omega}}^l_0&=&\tilde{\sigma}_0^l-k_{R^l}\hat{R}^{\top}\tilde{\mathcal{L}}_2^l\Psi^l-k_{\omega^l}((\Xi+LL^{\top})\otimes I_3)\tilde{\omega}_0^l,\\
\dot{\tilde{\sigma}}_0&=&(I_N\otimes P)\tilde{\sigma}_0
-k_{\tilde{\sigma}}((\Xi+LL^{\top})\otimes I_3)\tilde{\sigma}_0,
\end{array}
\end{equation}
where $\Psi^l=[\psi(A_1^l\tilde{R}_1^l)^{\top}\ \cdots\ \psi(A_N^l\tilde{R}_N^l)^{\top}]^{\top}$, $\hat{R}=\mathrm{diag}(\hat{R}_1,$
$\ldots,\hat{R}_N)$.

\subsection{Distributed control design}

In this subsection, we propose a distributed control scheme for each rigid body to track its local estimate, and further derive the dynamics of tracking errors.

Using the estimated attitude $\hat{R}_i$ and the estimated angular velocity $\hat{\omega}_i$, the control torque of the $i$-th rigid body system is designed as follows:
\begin{equation}\label{control}
    \tau_i=-\Sigma_i\bar{\omega}_i-\Gamma_i-k_{\bar{R}}\psi(\bar{A}_i\bar{R}_i^r)-k_{\bar{\omega}}\bar{\omega}_i,
\end{equation}
for $i\in\mathcal{V}$, where $k_{\bar{R}},k_{\bar{\omega}}>0$, $\bar{A}_i=\bar{A}_i^{\top}>0$ with three distinct eigenvalues, $\bar{R}_i^r=\hat{R}_i^{\top}R_i$ is the attitude tracking error between the $i$-th rigid body and its corresponding observer, $\bar{\omega}_i=\omega_i-(\bar{R}_i^r)^{\top}\hat{\omega}_i$ is the angular velocity tracking error,
\begin{equation}\label{sigmai}
\begin{array}{rcl}
\Sigma_i&=&(J_i\bar{\omega}_i)^{\times}-(J_i(\bar{R}_i^r)^{\top}\hat{\omega}_i)^{\times}\\
\ &\ &-((\bar{R}_i^r)^{\top}\hat{\omega}_i)^{\times}J_i-J_i((\bar{R}_i^r)^{\top}\hat{\omega}_i)^{\times},
\end{array}
\end{equation}
\begin{equation}\label{gammai}
\Gamma_i=(\bar{R}_i^r)^{\top}\hat{\omega}_i\times J_i(\bar{R}_i^r)^{\top}\hat{\omega}_i-J_i(\bar{R}_i^r)^{\top}\dot{\hat{\omega}}_i,
\end{equation}
with $\dot{\hat{\omega}}_i$ being defined in \eqref{7b}.
In the proposed control law \eqref{control}, the first two terms are devised to cancel the nonlinear effects arising in the dynamics of the tracking errors, while the remaining two terms drive the actual attitude and angular velocity to converge to their respective estimates.

The tracking error dynamics are given as follows:
\begin{equation}\label{controlerror1}
\setlength{\arraycolsep}{1pt}
\begin{array}{rcl}
\dot{\bar{R}}_i^r&=&\bar{R}_i^r\bar{\omega}_i^{\times},\\
J_i\dot{\bar{\omega}}_i&=&\Sigma_i\bar{\omega}_i+\Gamma_i+\tau_i.\\
\end{array}
\end{equation}
Substituting the designed control law \eqref{control} into \eqref{controlerror1} yields
\begin{equation}\label{controlerror2}
\setlength{\arraycolsep}{1pt}
\begin{array}{rcl}
\dot{\bar{R}}_i^r&=&\bar{R}_i^r\bar{\omega}_i^{\times},\\
J_i\dot{\bar{\omega}}_i&=&-k_{\bar{R}}\psi(\bar{A}_i\bar{R}_i^r)-k_{\bar{\omega}}\bar{\omega}_i, \\
\end{array}
\end{equation}
which can be expressed in the following compact form:
\begin{equation}\label{controlerrorcompact}
\setlength{\arraycolsep}{1pt}
\begin{array}{rcl}
\dot{\bar{R}}^r&=&\bar{R}^r\bar{W},\\
J\dot{\bar{\omega}}&=&-k_{\bar{R}}\bar{\Psi}-k_{\bar{\omega}}\bar{\omega},\\
\end{array}
\end{equation}
where $\bar{\Psi}=[\psi(\bar{A}_1\bar{R}_1^r)^{\top}\ \cdots\ \psi(\bar{A}_N\bar{R}_N^r)^{\top}]^{\top}$, $\bar{R}^r=\mathrm{diag}(\bar{R}_1^r,$
$\ldots,\bar{R}_N^r)$, $\bar{W}=\mathrm{diag}(\bar{\omega}_1^{\times},\ldots,\bar{\omega}_N^{\times})$, $J=\mathrm{diag}(J_1,\ldots,J_N)$, and
$\bar{\omega}=[\bar{\omega}_1^{\top}\ \cdots\ \bar{\omega}_N^{\top}]^{\top}$.

\begin{remark}\label{re10}
It is worth pointing out that the left-invariant attitude error $\tilde{R}_{i,j}^l=\hat{R}_i\hat{R}_j^{\top}$ is used for the observer design and the right-invariant attitude error $\bar{R}_i^r=\hat{R}_i^{\top}R_i$ is used for the control design. The choice of the right-invariant attitude error for the control design allows to obtain an autonomous closed-loop system \eqref {controlerrorcompact}, which facilitates the almost global asymptotic stability proof. The adoption of the left-invariant attitude estimation error facilitates considerably the observer design. While a uniform use of either the left-invariant or right-invariant attitude error for both observer and controller design might be feasible, it may result in a more complicated design and stability analysis.
\end{remark}

\subsection{Stability analysis}
To address \textit{Problem 1}, we aim to demonstrate that the proposed distributed control law \eqref{control} with the distributed observer \eqref{observer} enables the attitudes of the rigid body systems to achieve almost global asymptotic synchronization to a time-varying reference attitude.
To this end, we state the stability result of the error systems \eqref{errorcompact} and \eqref{controlerrorcompact} in this subsection.  Before stating our theorem, we define the state $\chi^l=(\tilde{R}^l,\tilde{\omega}_0^l,\tilde{\sigma}_0,\bar{R}^r,\bar{\omega})\in \mathcal{S}^l$, with $\mathcal{S}^l:=\mathrm{SO}(3)^N\times\mathbb{R}^{3N}\times\mathbb{R}^{3N}\times\mathrm{SO}(3)^N\times\mathbb{R}^{3N}$ and $\mathrm{SO}(3)^N=\{R=\mathrm{diag}(R_1,\ldots,R_N)|R_i\in\mathrm{SO}(3),\forall i=1,\ldots,N\}$, and the following sets:
\[
\begin{array}{l}
\mathcal{C}^{l,d}=\left\{\chi^l\in \mathcal{S}^l~|~\chi^l=(I_{3N},\mathbf{0}_{3N},\mathbf{0}_{3N},I_{3N},\mathbf{0}_{3N})\right\},
\end{array}
\]
\[\mathcal{C}^{l,u}=\left\{\chi^l\in \mathcal{S}^l~|~ \chi^l=(\tilde{R}^{l,\ast},\mathbf{0}_{3N},\mathbf{0}_{3N},\bar{R}^{r,\ast},\mathbf{0}_{3N})\right\},
\]
where
$\tilde{R}^{l,\ast}=\mathrm{diag}(\tilde{R}_1^{l,\ast},\ldots,\tilde{R}_N^{l,\ast})$, $\bar{R}^{r,\ast}=\mathrm{diag}(\bar{R}_1^{r,\ast},\ldots,$
$\bar{R}_N^{r,\ast})$, with $\tilde{R}_{p}^{l,\ast}=I_3,\forall p\in\mathcal{M}_{\overline{\mathcal{G}}_0}^{I}$, $\tilde{R}_{q}^{l,\ast}=\mathcal{R}(\pi,\tilde{u}_{q}^l),\forall q\in\mathcal{M}_{\overline{\mathcal{G}}_0}^{\pi}$, $\bar{R}_{p}^{r,\ast}=I_3,\forall p\in\mathcal{V}^{I}$, $\bar{R}_{q}^{r,\ast}=\mathcal{R}(\pi,\bar{u}_{q}),\forall q\in\mathcal{V}^{\pi}$. The sets  $\mathcal{M}_{\overline{\mathcal{G}}_0}^{I}$ and $\mathcal{M}_{\overline{\mathcal{G}}_0}^{\pi}$ satisfy $\mathcal{M}_{\overline{\mathcal{G}}_0}^{I}\cup\mathcal{M}_{\overline{\mathcal{G}}_0}^{\pi}=\mathcal{M}_{\overline{\mathcal{G}}_0}$, $\mathcal{M}_{\overline{\mathcal{G}}_0}^{\pi}\neq\emptyset$, and the sets $\mathcal{V}^I$ and $\mathcal{V}^{\pi}$ satisfy $\mathcal{V}^I\cup\mathcal{V}^{\pi}=\mathcal{V}$, $\mathcal{V}^{\pi}\neq\emptyset$.
The vectors $\tilde{u}_{q}^l$, and $\bar{u}_{q}$ are the unit eigenvectors associated to the eigenvalues of $A_{q}^l$ and $\bar{A}_{q}$, respectively.

\begin{theorem}\label{estimationstability}
Consider a network of $N$ rigid body systems governed by the rotational dynamics \eqref{followeri}.
The desired attitude and the desired angular velocity trajectories are generated by \eqref{leader}.
Let the distributed observer and control torque be designed as \eqref{observer} and \eqref{control}, respectively.
Under Assumption \ref{assumpgraph_original}, the following statements hold.
\begin{enumerate}
    \item The trajectories of the error systems \eqref{errorcompact} and $\eqref{controlerrorcompact}$ converge to the set of equilibria $\mathcal{C}^l=\mathcal{C}^{l,d}\cup\mathcal{C}^{l,u}$.
    \item The undesired equilibria in $\mathcal{C}^{l,u}$ are unstable, and the desired equilibrium $\mathcal{C}^{l,d}$ is almost globally asymptotically stable.
    \item The attitudes of the rigid body systems synchronize (almost\footnote{From all initial conditions except from a set of zero Lebesgue measure} globally)
    to the desired attitude $R_0(t)$, \emph{i.e.}, $\lim_{t\rightarrow\infty}R_i(t)R_0(t)^{\top}=I_3$ and $\lim_{t\rightarrow\infty}(\omega_i(t)-\omega_0(t))=\mathbf{0}_3$, for all $i\in\mathcal{V}$, for almost all initial conditions.
\end{enumerate}
\end{theorem}

\begin{proof}
See Appendix \ref{appenD}.
\end{proof}

\begin{remark}\label{re1}
The control torque $\tau_i$ formulated in \eqref{control} compensates the nonlinear terms $\Sigma_i\bar{\omega}_i$ and $\Gamma_i$
arising in the angular velocity tracking error dynamics.
Note that the term $\Sigma_i\bar{\omega}_i$ is a passive term in the sense that $\bar{\omega}_i^\top\Sigma_i\bar{\omega}_i=0$ since $\Sigma_i$ is a skew symmetric matrix. Hence, this term does not affect the Lyapunov time-derivative that is used in the proof of Theorem \ref{estimationstability}. The cancellation of the term $\Sigma_i\bar{\omega}_i$ is performed only for technical reasons that will lead to an autonomous tracking error subsystem that allows to prove almost global asymptotic stability.
Without the compensation of this term, the control torque given by
\begin{equation}\label{controlconstantestimate}
\tau_i=-\Gamma_i-k_{\bar{R}}\psi(\bar{A}_i\bar{R}_i^r)-k_{\bar{\omega}}\bar{\omega}_i,
\end{equation}
is still valid and we can prove that the desired equilibrium $\mathcal{C}^{l,d}$ is locally exponentially stable.
\end{remark}

\section{Attitude Synchronization to A Constant Reference Attitude}\label{constant}

To address \textit{Problem 2}, this section develops two distributed control schemes that enable the attitudes of the rigid body systems to synchronize to a constant reference attitude.
The first scheme is an observer-based distributed control, which is derived from the observer-based distributed control law proposed in the previous section for the attitude synchronization to a time-varying reference attitude.
Note that a constant reference attitude leads to a zero desired angular velocity, allowing each rigid body to implement a control law without estimating the reference attitude.
Hence, we propose an observerless distributed control scheme as the second strategy dedicated to the attitude synchronization to a constant reference attitude.

Figures \ref{observerbasedconstant}-\ref{observerlessconstant} depict the block diagrams for implementing the proposed distributed control schemes that achieve the attitude synchronization to a constant reference attitude.
Figure \ref{observerbasedconstant} displays the structure of the proposed observer-based distributed control law \eqref{observerconstant} incorporating the distributed observer \eqref{controlconstantestimate}, while Figure \ref{observerlessconstant} shows the implementation of the proposed observerless distributed control law \eqref{districontrol}.

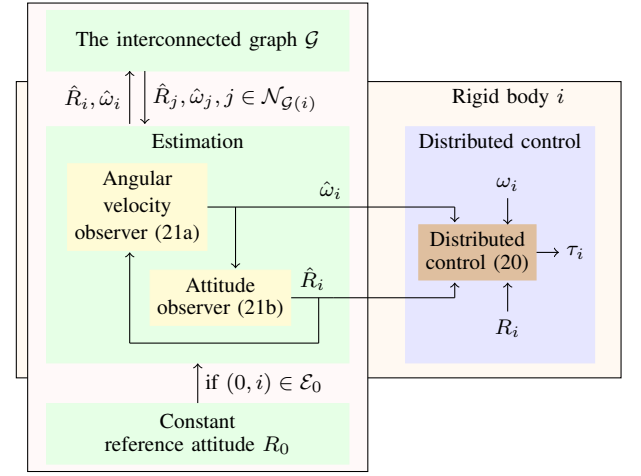
\begin{figure}[!htbp]
\footnotesize
\centering
\tikzstyle{test}=[diamond,aspect=2,draw,thin]
\tikzstyle{line} = [draw, ->]
\tikzstyle{point}=[coordinate,on grid,]
\begin{tikzpicture}
\node[rectangle, minimum width=8cm, minimum height=2cm, draw = black,fill=orange!5] (outer1) {\shortstack[1]{\ \ \ \ \ \ \ \ \ \ \ \ \ \ \ \ \ \ \ \ \ \ \ \ \ \ \ \ \ \ \ \ \ \ \ \ \ \ \ \ \ \ \ \ \ \ \ \ \ \ \ Rigid body $i$\\ \ \\ \ \\ \ \\ \ \\ \ \\ \ \\ \ \\ \ \\ \ \\ \ \\ \ \\ \ \\ \ \\ \ \\ \ \\ \ \\ \ \\ \ \\ \ \\ \ \\ \ \\ \ \\ \ \\ \ \\ \ \\ \ \\ \ \\ \ \\ \ \\ \ \\ \ \\ \ \\ \ }};
\node[rectangle, minimum width=4.5cm, minimum height=2cm,xshift=-1.6cm, yshift=-0.1cm, draw = black,fill=pink!10] (outer3) {\shortstack[1]{\\ \ \\ \ \\ \ \\ \ \\ \ \\ \ \\ \ \\ \ \\ \ \\ \ \\ \ \\ \ \\ \ \\ \ \\ \ \\ \ \\ \ \\ \ \\ \ \\ \ \\ \ \\ \ \\ \ \\ \ \\ \ \\ \ \\ \ \\ \ \\ \ \\ \ \\ \ \\ \ \\ \ \\ \ \\ \ \\ \ \\ \ \\ \ \\ \ \\ \ \\ \ \\ \ \\ \ \\ \ \\ \ \\ \ \\ \ \\ \ \\ \ \\ \ \\ \ \\ \ \\ \ \\ \ \\ \ \\ \ \\ \ }};
\node[rectangle, minimum width=4cm, minimum height=1cm,   xshift=-1.6cm, yshift=-0.2cm, draw = green!10,fill=green!10] (mid1) {\shortstack[1]{Estimation\\ \ \\ \ \\ \ \\ \ \\ \ \\ \ \\ \ \\ \ \\ \ \\ \ \\ \ \\ \ \\ \ \\ \ \\ \ \\ \ \\ \ \\ \ \\ \ \\ \ \\ \ \\ \ \\ \ \\ \ \\ \ \\ \ }};
\node[rectangle,minimum width=1.2cm,minimum height=0.5cm,xshift=-2.4cm,yshift=0.3cm,draw=yellow!20,fill=yellow!20](inner12) {\shortstack[1]{Angular\\ velocity\\ observer \eqref{16a}}};
\node[rectangle,minimum width=1.2cm,minimum height=0.5cm,xshift=-1.3cm,yshift=-0.9cm,draw=yellow!20,fill=yellow!20](inner13) {\shortstack[1]{Attitude\\ observer \eqref{16b}}};
\node[point,right of=inner13,node distance=1.5cm](prinner13){};
\draw[-] (inner13)--node[above]{$\hat{R}_i$}(prinner13);
\node[point,right of =inner12,node distance=1.3cm](pr2inner12){};
\node[point,below of=pr2inner12,node distance=0.81cm](prb2inner12){};
\draw[line] (pr2inner12)--(prb2inner12);
\node[point,right of =inner13,node distance=1.3cm](pr2inner13){};
\node[point,below of=pr2inner13,node distance=0.6cm](prb2inner13){};
\draw[-] (pr2inner13)--(prb2inner13);
\node[point,left of=prb2inner13,node distance=2.5cm](prbl2inner13){};
\draw[-] (prb2inner13)--(prbl2inner13);
\node[point,above of=prbl2inner13,node distance=1.2cm](prbla2inner13){};
\draw[line] (prbl2inner13)--(prbla2inner13);
\node[rectangle, minimum width=2.5cm, minimum height=1cm,   xshift=2.4cm, yshift=-0.2cm, draw = blue!10,fill=blue!10] (mid2) {\shortstack[1]{Distributed control\\ \ \\ \ \\ \ \\ \ \\ \ \\ \ \\ \ \\ \ \\ \ \\ \ \\ \ \\ \ \\ \ \\ \ \\ \ \\ \ \\ \ \\ \ \\ \ \\ \ \\ \ \\ \ \\ \ \\ \ \\ \ \\ \ }};
\node[rectangle,minimum width=1.2cm,minimum height=0.5cm,xshift=2.1cm,yshift=-0.3cm,draw=brown!50,fill=brown!50] (inner2){\shortstack[1]{Distributed \\control \eqref{controlconstantestimate}}};
\node[point,right of=inner12,node distance=4.2cm](prinner12){};
\draw[-] (inner12)--node[above]{$\hat{\omega}_i$}(prinner12);
\node[point,below of=prinner12,node distance=0.2cm](prbinner12){};
\draw[line] (prinner12)--(prbinner12);
\node[point,right of=prinner13,node distance=1.6cm](prinner132){};
\draw[-] (prinner13)--(prinner132);
\node[point,above of=prinner132,node distance=0.2cm](prainner132){};
\draw[line] (prinner132)--(prainner132);
\node[point,right of=prainner132,node distance=0.7cm](prinner12r){};
\node[rectangle,minimum width=0.3cm,minimum height=0.3cm,xshift=2.5cm,yshift=-1.3cm,draw=blue!10,fill=blue!10] (R1){$R_i$};
\draw[line] (R1)--(prinner12r);
\node[rectangle,minimum width=0.3cm,minimum height=0.3cm,xshift=2.5cm,yshift=0.6cm,draw=blue!10,fill=blue!10] (omega1){$\omega_i$};
\node[point,above of=prinner12r,node distance=0.77cm](prinner12ra){};
\draw[line] (omega1)--(prinner12ra);
\node[rectangle,minimum width=0.3cm,minimum height=0.3cm,xshift=3.4cm,yshift=-0.3cm,draw=blue!10,fill=blue!10] (tau1){$\tau_i$};
\draw[line] (inner2)--(tau1);
\node[rectangle, minimum width=4cm, minimum height=0.8cm, xshift=-1.6cm, yshift=2.5cm,  draw = green!10,fill=green!10] (outer2) {The interconnected graph $\mathcal{G}$};
\node[point,left of=outer2,node distance=0.7cm](prouter21){};
\node[point,below of=prouter21,node distance=1.1cm](prbouter21){};
\node[point,below of=prouter21,node distance=0.4cm](prbouter22){};
\draw[line] (prbouter22)--node[right]{\shortstack[1]{$\hat{R}_j,\hat{\omega}_j,j\in\mathcal{N}_{\mathcal{G}(i)}$}}(prbouter21);
\node[point,left of=outer2,node distance=0.9cm](prouter21){};
\node[point,below of=prouter21,node distance=1.1cm](prbouter21){};
\node[point,below of=prouter21,node distance=0.4cm](prbouter22){};
\draw[line] (prbouter21)--node[left]{\shortstack[1]{$\hat{R}_i,\hat{\omega}_i$}}(prbouter22);
\node[rectangle, minimum width=4cm, minimum height=0.8cm, xshift=-1.6cm, yshift=-2.7cm,  draw = green!10,fill=green!10] (leader) {\shortstack[1]{Constant\\ reference attitude $R_0$}};
\draw[line] (leader)--node[right]{if $(0,i)\in\mathcal{E}_0$}(mid1);
\end{tikzpicture}
\caption{The implementation block diagram of the observer-based control law \eqref{controlconstantestimate} incorporating the distributed observer \eqref{observerconstant}}\label{observerbasedconstant}
\end{figure}

\begin{figure}[!htbp]
\footnotesize
\centering
\tikzstyle{test}=[diamond,aspect=2,draw,thin]
\tikzstyle{line} = [draw, ->]
\tikzstyle{point}=[coordinate,on grid,]
\begin{tikzpicture}
\node[rectangle, minimum width=6cm, minimum height=1.5cm, draw =green!10,fill=green!10] (outer1) {\shortstack[1]{Rigid body $i$\\ \ \\ \ \\ \ \\ \ \\ \ \\ \ \\ \ \\ \ \\ \ \\ \ \\ \ \\ \ }};
\node[rectangle,minimum width=2cm,minimum height=1cm,xshift=0cm,yshift=-0.2cm,draw=yellow!20,fill=yellow!20] (inner2){\shortstack[1]{Distributed \\control \eqref{districontrol}}};
\node[rectangle,minimum width=0.3cm,minimum height=0.3cm,xshift=-2cm,yshift=0cm,draw=green!10,fill=green!10] (omega1){$\omega_i$};
\node[point,right of=omega1,node distance=1cm](promega1){};
\draw[line] (omega1)--(promega1);
\node[rectangle,minimum width=0.3cm,minimum height=0.3cm,xshift=-2cm,yshift=-0.4cm,draw=green!10,fill=green!10] (R1){$R_i$};
\node[point,right of=R1,node distance=1cm](prR1){};
\draw[line] (R1)--(prR1);
\node[rectangle,minimum width=0.3cm,minimum height=0.3cm,xshift=2cm,yshift=-0.2cm,draw=green!10,fill=green!10] (tau1){$\tau_i$};
\node[point,left of=tau1,node distance=1cm](pltau1){};
\draw[line] (pltau1)--(tau1);
\node[rectangle, minimum width=6cm, minimum height=0.8cm, xshift=0cm, yshift=1.85cm,  draw = green!10,fill=green!10] (outer2) {The interconnected graph $\mathcal{G}$};
\node[point,right of=outer2,node distance=0.1cm](prouter21){};
\node[point,below of=prouter21,node distance=1cm](prbouter21){};
\node[point,below of=prouter21,node distance=0.4cm](prbouter22){};
\draw[line] (prbouter22)--node[right]{\shortstack[1]{$R_j,j\in\mathcal{N}_{\mathcal{G}(i)}$}}(prbouter21);
\node[point,left of=outer2,node distance=0.1cm](prouter21){};
\node[point,below of=prouter21,node distance=1cm](prbouter21){};
\node[point,below of=prouter21,node distance=0.4cm](prbouter22){};
\draw[line] (prbouter21)--node[left]{\shortstack[1]{$R_i$}}(prbouter22);
\node[rectangle, minimum width=6cm, minimum height=0.8cm, xshift=0cm, yshift=-1.85cm,  draw = green!10,fill=green!10] (leader) {\shortstack[1]{Constant reference attitude $R_0$}};
\node[point,above of=leader,node distance=1cm](paleader){};
\draw[line] (leader)--node[right]{if $(0,i)\in\mathcal{E}_0$}node[left]{$R_0$}(paleader);
\end{tikzpicture}
\caption{The implementation block diagram of the observerless control law \eqref{districontrol}}\label{observerlessconstant}
\end{figure}
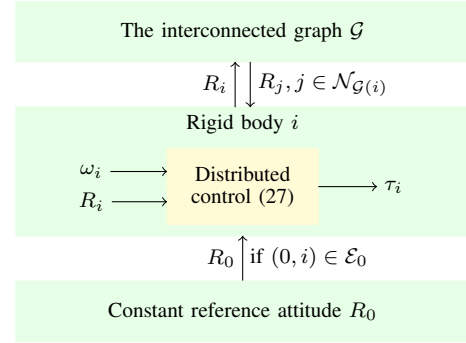

\subsection{Distributed observer-based attitude synchronization to a constant reference attitude}

This subsection deals with the design of an observer-based distributed control scheme for the attitude synchronization to a constant reference attitude.  Since the constant desired attitude leads to a zero desired angular velocity, as well as a zero desired angular acceleration, innovation terms are designed directly such that the angular velocity estimate converges to zero.
This eliminates the necessity of estimating the desired angular acceleration.
Hence, we propose the following distributed observer:
\begin{subequations}\label{observerconstant}
\begin{align}
&\dot{\hat{R}}_i=\hat{R}_i\hat{\omega}_i^{\times},\label{16a}\\
&\dot{\hat{\omega}}_i=-k_{R^r}\sum\limits_{j\in\mathcal{N}_{\mathcal{G}}(i)}\psi(A^r_{i,j}\hat{R}_j^{\top}\hat{R}_i)-k_{\omega^r}\sum\limits_{j\in\mathcal{N}_{\mathcal{G}(i)}}(\hat{\omega}_i-\hat{\omega}_j)\nonumber\\
&\ \ \ \ \ \ \ \ -k_{R^r}\psi(A^r_{i,0}R_0^{\top}\hat{R}_i)-k_{\omega^r}\Xi_i\hat{\omega}_i,\label{16b}
\end{align}
\end{subequations}
for $i\in\mathcal{V}$,
where $k_{R^r},k_{\omega^r}>0$, $A^r_{i,j}=(A^r_{i,j})^{\top}>0$ with three distinct eigenvalues and $A^r_{i,j}=A^r_{j,i}$, $A^r_{i,0}=(A^r_{i,0})^{\top}>0$ with three distinct eigenvalues, $A^r_{i,0}=\mathbf{0}_{3\times 3},\forall i\in\mathcal{V}\setminus\{1,\ldots,n\}$.
In view of the proposed distributed observer \eqref{16a} for the desired attitude, the desired attitude estimate is adjusted by correcting the desired angular velocity estimate.
The proposed distributed observer \eqref{16b} for the desired angular velocity includes four innovation terms.
The first two terms leverage relative attitude and relative angular velocity information, enabling the estimates of all rigid body systems to converge to a common attitude.
The remaining two terms are activated solely for the rigid body systems with access to the leader's information, allowing to drive the attitude and the angular velocity estimates toward their respective desired values.

Now, we derive the dynamics of the estimation errors.
The absolute attitude error between the $i$-th observer and the leader $0$ is defined as $\tilde{R}^r_{i,0}=R_0^{\top}\hat{R}_i$, whose dynamics satisfy  $\dot{\tilde{R}}^r_{i,0}=\tilde{R}^r_{i,0}(\tilde{\omega}^r_{i,0})^{\times}$ for $i\in\mathcal{V}$, where $\tilde{\omega}^r_{i,0}=\hat{\omega}_i-(\tilde{R}^r_{i,0})^{\top}\omega_0$ is the absolute angular velocity error between the $i$-th observer and the leader $0$.
Since the desired attitude $R_0$ is constant, one has $\omega_0=\mathbf{0}_3$, as well as $\tilde{\omega}^r_{i,0}=\hat{\omega}_i$ for $i\in\mathcal{V}$.
Hence, the resulting dynamics for the absolute angular velocity error are as follows:
\[
\setlength{\arraycolsep}{1pt}
\begin{array}{rcl}
\dot{\tilde{\omega}}^r_{i,0}&=&-k_{R^r}\sum\limits_{j\in\mathcal{N}_{\mathcal{G}}(i)}\psi(A^r_{i,j}\hat{R}_j^{\top}\hat{R}_i)-k_{\omega^r}(\tilde{\omega}^r_{i,0}-\tilde{\omega}^r_{j,0})\\
\ &\ &-k_{R^r}\psi(A^r_{i,0}R_0^{\top}\hat{R}_i)-k_{\omega^r}\Xi_i\tilde{\omega}^r_{i,0}.
\end{array}
\]
The relative attitude error between the observers $i,j\in\mathcal{V}$ satisfying $(j,i)\in\overline{\mathcal{E}}$ is defined as $\tilde{R}^r_{i,j}=\hat{R}_j^{\top}\hat{R}_i$.
Then its dynamics are  $\dot{\tilde{R}}^r_{i,j}=\tilde{R}^r_{i,j}(\tilde{\omega}^r_{i,j})^{\times}$, where the relative angular velocity error is given by $\tilde{\omega}^r_{i,j}=\hat{\omega}_i-(\tilde{R}^r_{i,j})^{\top}\hat{\omega}_j=\tilde{\omega}^r_{i,0}-(\tilde{R}^r_{i,j})^{\top}\tilde{\omega}^r_{j,0}$.
In the obtained oriented graph $\overline{\mathcal{G}}_0$, we index edge $(j,i)\in\overline{\mathcal{E}}_0$ as the $k$-th edge, and
express $\tilde{R}^r_{i,j}$ as $\tilde{R}^r_k$, \textit{i.e.}, $\tilde{R}^r_k=\tilde{R}^r_{i,j}$, and $\tilde{\omega}^r_{i,j}$ as $\tilde{\omega}^r_k$, \textit{i.e.}, $\tilde{\omega}^r_k=\tilde{\omega}^r_{i,j}$.
Let $\tilde{R}^r=\mathrm{diag}(\tilde{R}^r_1,\ldots,\tilde{R}^r_N)$ and $\tilde{W}^r=\mathrm{diag}(\tilde{\omega}^r_1)^{\times},\ldots,(\tilde{\omega}^r_N)^{\times})$.
The compact form of the dynamics of attitude errors is given by
\begin{equation}\label{constantattitudeerror}
\dot{\tilde{R}}^r=\tilde{R}^r\tilde{W}^r.
\end{equation}
As per Lemma \ref{le3} in Appendix \ref{appenC}, one has
\[
\sum\limits_{j\in\mathcal{N}_{\mathcal{G}}(i)}\psi(A^r_{i,j}\tilde{R}_{i,j}^{r})+\psi(A^r_{i,0}\tilde{R}_{i,0}^{r})=\sum\limits_{k=1}^N(\tilde{\mathcal{L}}^r_2)_{i,k}\psi(A^r_{k}\tilde{R}^r_k),
\]
where $A^r_k=A^r_{i,j}$ for $i,j\in\mathcal{V}_0$ with the edge between $i$ and $j$ being described by edge $k$, $\tilde{\mathcal{L}}_2^r$ is obtained from $\tilde{\mathcal{L}}^r$ by removing the first three rows, $\tilde{\mathcal{L}}^r\in\mathbb{R}^{3(N+1)\times 3N}$ is constructed according to the obtained oriented graph $\overline{\mathcal{G}}_0$, with each $3$-by-$3$ block $(\tilde{\mathcal{L}}^r)_{ik}$ being defined as
\begin{equation}\label{tildeL}
(\tilde{\mathcal{L}}^r)_{ik}=\left\{
\begin{array}{ll}
-\tilde{R}^r_k, & k\in(\mathcal{M}_{\overline{\mathcal{G}}_0})^{+}_{i-1},  \\
I_3, & k\in(\mathcal{M}_{\overline{G}_0})^{-}_{i-1},\\
\mathbf{0}_{3\times 3}, &\mathrm{otherwise}.\\
\end{array}
\right.
\end{equation}
It is clear that $\sum\limits_{j\in\mathcal{N}_{\mathcal{G}}(i)}(\tilde{\omega}^r_{i,0}-\tilde{\omega}^r_{j,0})=(\mathrm{Row}_i(LL^{\top})\otimes I_3)\tilde{\omega}^r_0$, where $\tilde{\omega}^r_0=[(\tilde{\omega}^r_{1,0})^{\top}\ \cdots\ (\tilde{\omega}^r_{N,0})^{\top}]^{\top}$.
Hence, the compact form of the estimation error system is as follows:
\begin{equation}\label{estimatedcompactconstant}
\begin{array}{rcl}
\dot{\tilde{R}}^r&=&\tilde{R}^r\tilde{W}^r,\\
\dot{\tilde{\omega}}^r_0&=&-k_{R^r}\tilde{\mathcal{L}}^r_2\Psi^r-k_{\omega^r}((\Xi+LL^{\top})\otimes I_3)\tilde{\omega}^r_0,\\
\end{array}
\end{equation}
where $\Psi^r=[\psi(A^r_1\tilde{R}^r_1)^{\top}\ \cdots\ \psi(A^r_N\tilde{R}^r_N)^{\top}]$.

To drive each rigid body to track its own estimates, we design the control torque $\tau_i$ for the $i$-th rigid body system as in \eqref{controlconstantestimate},
where $\Gamma_i$ is defined in \eqref{gammai} with $\dot{\hat{\omega}}_i$ being defined in \eqref{16b}.
The first term in the proposed distributed control law \eqref{controlconstantestimate} compensates for nonlinear effects arising in the tracking error dynamics. The last two terms act to drive the actual attitude towards its estimate.

Substituting the designed control \eqref{controlconstantestimate}  into \eqref{controlerror1}, one can derive the following closed-loop system
\begin{equation}\label{controlerror21}
\begin{array}{rcl}
\dot{\bar{R}}_i^r&=&\bar{R}_i^r\bar{\omega}_i^{\times},\\
J_i\dot{\bar{\omega}}_i&=&\Sigma_i\bar{\omega}_i-k_{\bar{R}}\psi(\bar{A}_i\bar{R}_i^r)-k_{\bar{\omega}}\bar{\omega}_i, \\
\end{array}
\end{equation}
which can be expressed in the following compact form:
\begin{equation}\label{controlerrorcompact1}
\setlength{\arraycolsep}{1pt}
\begin{array}{rcl}
\dot{\bar{R}}^r&=&\bar{R}^r\bar{W},\\
J\dot{\bar{\omega}}&=&\Sigma\bar{\omega}-k_{\bar{R}}\bar{\Psi}-k_{\bar{\omega}}\bar{\omega},\\
\end{array}
\end{equation}
where $\Sigma=\mathrm{diag}(\Sigma_1,\ldots,\Sigma_N)$ with $\Sigma_i$ being defined in \eqref{sigmai}.

Now we state the stability results of the error systems \eqref{estimatedcompactconstant} and \eqref{controlerrorcompact1}.
The obtained stability results allow us to conclude that the proposed distributed control law \eqref{controlconstantestimate}, integrated with the distributed observer \eqref{observerconstant}, steers the attitudes of the rigid body systems to synchronize to the desired constant attitude, \textit{i.e.}, \textit{Problem 2} is solved.
Before stating our theorem, we define the state $\chi^r=(\tilde{R}^r,\tilde{\omega}_0^r,\bar{R}^r,\bar{\omega})\in \mathcal{S}^r$, with $\mathcal{S}^r:=\mathrm{SO}(3)^N\times\mathbb{R}^{3N}\times\mathrm{SO}(3)^N\times\mathbb{R}^{3N}$, and the following sets:
\[
\begin{array}{l}
\mathcal{C}^{r,d}=\left\{\chi^r\in \mathcal{S}^r~|~\chi^r=(I_{3N},\mathbf{0}_{3N},I_{3N},\mathbf{0}_{3N})\right\},
\end{array}
\]
\[\mathcal{C}^{r,u}=\left\{\chi^r\in \mathcal{S}^r~|~ \chi^r=(\tilde{R}^{r,\ast},\mathbf{0}_{3N},\bar{R}^{r,\ast},\mathbf{0}_{3N})\right\},
\]
where
$\tilde{R}^{r,\ast}=\mathrm{diag}(\tilde{R}_1^{r,\ast},\ldots,\tilde{R}_N^{r,\ast})$, $\bar{R}^{r,\ast}=\mathrm{diag}(\bar{R}_1^{r,\ast},\ldots,$
$\bar{R}_N^{r,\ast})$, with $\tilde{R}_{p}^{r,\ast}=I_3,\forall p\in\mathcal{M}_{\overline{\mathcal{G}}_0}^{I}$, $\tilde{R}_{q}^{r,\ast}=\mathcal{R}(\pi,\tilde{u}_{q}^r),\forall q\in\mathcal{M}_{\overline{\mathcal{G}}_0}^{\pi}$, $\bar{R}_{p}^{r,\ast}=I_3,\forall p\in\mathcal{V}^{I}$, $\bar{R}_{q}^{r,\ast}=\mathcal{R}(\pi,\bar{u}_{q}),\forall q\in\mathcal{V}^{\pi}$. The sets  $\mathcal{M}_{\overline{\mathcal{G}}_0}^{I}$ and $\mathcal{M}_{\overline{\mathcal{G}}_0}^{\pi}$ satisfy $\mathcal{M}_{\overline{\mathcal{G}}_0}^{I}\cup\mathcal{M}_{\overline{\mathcal{G}}_0}^{\pi}=\mathcal{M}_{\overline{\mathcal{G}}_0}$, $\mathcal{M}_{\overline{\mathcal{G}}_0}^{\pi}\neq\emptyset$, and the sets $\mathcal{V}^I$ and $\mathcal{V}^{\pi}$ satisfy $\mathcal{V}^I\cup\mathcal{V}^{\pi}=\mathcal{V}$, $\mathcal{V}^{\pi}\neq\emptyset$.
The vectors $\tilde{u}_{q}^r$, and $\bar{u}_{q}$ are the unit eigenvectors associated to the eigenvalues of $A_{q}^r$ and $\bar{A}_{q}$, respectively.

\begin{theorem}\label{stabilityforestimateconstant}
Consider a network of $N$ rigid body systems governed by the rotational dynamics \eqref{followeri}.
The desired attitude $R_0$ is constant.
Let the distributed observer and control torque be designed as \eqref{observerconstant} and \eqref{controlconstantestimate}, respectively.
Under Assumption \ref{assumpgraph_original}, the following statements hold.
\begin{enumerate}
    \item The trajectories of the error systems \eqref{estimatedcompactconstant} and \eqref{controlerrorcompact1} converge to the set of equilibria $\mathcal{C}^r=\mathcal{C}^{r,d}\cup\mathcal{C}^{r,u}$.
    \item The undesired equilibria in $\mathcal{C}^{r,u}$ are unstable, and the desired equilibrium $\mathcal{C}^{r,d}$ is almost globally asymptotically stable.
    \item  All the rigid body attitudes synchronize (almost globally) to the desired attitude $R_0$, \textit{i.e.}, $\lim_{t\rightarrow\infty}R_i(t)=R_0$ and $\lim_{t\rightarrow\infty}\omega_i(t)=\mathbf{0}_3$ for all $i\in\mathcal{V}$, for almost all initial conditions.
\end{enumerate}
\end{theorem}

\begin{proof}
See Appendix \ref{appenE}.
\end{proof}

\begin{remark}\label{re2}
Unlike the observer-based attitude synchronization strategy for the time-varying reference attitude, the control torque $\tau_i$ formulated in \eqref{controlconstantestimate} for the constant-reference case does not compensate the nonlinear term $\Sigma_i\bar{\omega}_i$, and still ensures that the desired equilibrium $\mathcal{C}^{r,d}$ is almost globally asymptotically stable.
In addition, the constant reference attitude yields a zero desired angular velocity and a zero desired angular acceleration.
This allows innovation terms to be designed to converge to zero, rendering the estimate of the desired angular acceleration unnecessary.
These two distinctions collectively reduce the structural complexity of the observer-based attitude synchronization strategy for the constant reference attitude, as compared with its counterpart for the time-varying reference case.
\end{remark}

\subsection{Observerless distributed attitude synchronization to a constant reference attitude}

A constant reference attitude yields a zero desired angular velocity, which allows to design an observerless distributed control law to achieve the attitude synchronization.
Hence, we propose the following distributed control torque $\tau_i$:
\begin{equation}\label{districontrol}
\tau_i=-k_{\check{R}}\psi(\check{A}_{i,0}R_0^{\top}R_i)-k_{\check{R}}\sum\limits_{j\in\mathcal{N}_{\mathcal{G}}(i)}\psi(\check{A}_{i,j}R_j^{\top}R_i)-k_{\check{\omega}}\omega_i
\end{equation}
for $i\in\mathcal{V}$, where $k_{\check{R}},k_{\check{\omega}}>0$,  $\check{A}_{i,j}=\check{A}_{i,j}^{\top}>0$ with three distinct eigenvalues and $\check{A}_{i,j}=\check{A}_{j,i}$, $\check{A}_{i,0}=\check{A}_{i,0}^{\top}>0$ with three distinct eigenvalues, $\check{A}_{i,0}=\mathbf{0}_{3\times 3},\forall i\in\mathcal{V}\setminus\{1,\ldots,n\}$.
The first term of the proposed observerless distributed control law \eqref{districontrol}, is activated only for the rigid body with access to the leader's information, allowing its attitude to converge to the desired value.
The second term represents the relative attitude, allowing the attitude to align to a common value.
The last term is dedicated to driving the angular velocity to zero.

Next we derive the dynamics of the synchronized errors.
The absolute attitude error between the $i$-th rigid body and the leader $0$ is defined as $\check{R}^r_{i,0}=R_0^{\top}R_i$. Its dynamics satisfy  $\dot{\check{R}}^r_{i,0}=\check{R}^r_{i,0}\check{\omega}_{i,0}^{\times}$, where $\check{\omega}_{i,0}=\omega_i-(\check{R}^r_{i,0})^{\top}\omega_0$ is the absolute angular velocity error between the $i$-th rigid body and the leader $0$.
Since the desired reference attitude $R_0$ is constant, it follows that $\omega_0=\mathbf{0}_{3}$, which further simplifies $\check{\omega}_{i,0}=\omega_i$.
Together with \eqref{districontrol}, the dynamics of absolute angular velocity errors are as follows:
\begin{equation}\label{angularerror}
\begin{array}{rcl}
J_i\dot{\check{\omega}}_{i,0}&=&-\check{\omega}_{i,0}^{\times}J_i\check{\omega}_{i,0}-k_{\check{R}}\psi(\check{A}_{i,0}R_0^{\top}R_i)\\
\ &\ &-k_{\check{R}}\sum\limits_{j\in\mathcal{N}_{\mathcal{G}}(i)}\psi(\check{A}_{i,j}R_j^{\top}R_i)-k_{\check{\omega}}\check{\omega}_{i,0}.
\end{array}\end{equation}
The relative attitude error between the rigid body systems $i,j\in\mathcal{V}$ with $(j,i)\in\overline{\mathcal{E}}$ is defined as $\check{R}^r_{i,j}=R_j^{\top}R_i$.
Its dynamics satisfy $\dot{\check{R}}^r_{i,j}=\check{R}^r_{i,j}\check{\omega}_{i,j}^{\times}$, where $\check{\omega}_{i,j}=\omega_{i}-(\check{R}^r_{i,j})^{\top}\omega_j=\check{\omega}_{i,0}-(\check{R}^r_{i,j})^{\top}\check{\omega}_{j,0}$ is the relative angular velocity error.

In terms of the obtained oriented graph $\overline{\mathcal{G}}_0$, if the edge $(j,i)\in\overline{\mathcal{E}}_0$ is denoted as the $k$-edge, then we define $\check{R}^r_k=\check{R}^r_{i,j}$ and $\check{\omega}_k=\check{\omega}_{i,j}$.
It follows that $\dot{\check{R}}^r=\check{R}^r\check{W}$, where $\check{R}^r=\mathrm{diag}(\check{R}^r_1,\ldots,\check{R}^r_N)$ and $\check{W}=\mathrm{diag}(\check{\omega}_1^{\times},\ldots,\check{\omega}_N^{\times})$.
Hence,
the compact form of the synchronization error system is as follows:
\begin{equation}\label{compactobserverless}
\setlength{\arraycolsep}{1pt}
\begin{array}{rcl}
\dot{\check{R}}^r&=&\check{R}^r\check{W},\\
J\dot{\check{\omega}}_0&=&-\check{W}_0J\check{\omega}_0-k_{\check{R}}\check{\mathcal{L}}^r_2\check{\Psi}-k_{\check{\omega}}\check{\omega}_0,
\end{array}
\end{equation}
where $\check{\Psi}=[\psi(\check{A}_1\check{R}^r_1)^{\top}\ \cdots\ \psi(\check{A}_N\check{R}^r_N)^{\top}]^{\top}$, $\check{\omega}_0=[\check{\omega}_{1,0}^{\top}$
$\cdots\ \check{\omega}_{N,0}^{\top}]^{\top}$, $\check{W}_0=\mathrm{diag}(\check{\omega}_{1,0}^{\times},\ldots,\check{\omega}_{N,0}^{\times})$,
$\check{\mathcal{L}}^r_2$ and $\check{\mathcal{L}}^r$ are defined in Section \ref{lemmas}, and as per Lemma \ref{le3} in Appendix \ref{appenC}, we used the following identity
\[
\psi(\check{A}_{i,0}\check{R}_{i,0}^r)+\sum\limits_{j\in\mathcal{N}_{\mathcal{G}}(i)}\psi(\check{A}_{i,j}\check{R}_{i,j}^r)=\sum\limits_{k=1}^{N}(\check{\mathcal{L}}^r_2)_{i,k}\psi(\check{A}_k\check{R}^r_k),
\]
with $\check{A}_k=\check{A}_{i,j}$ for $i,j\in\mathcal{V}_0$ with the edge between $i$ and $j$ being denoted as edge $k$.

To address \textit{Problem 2}, we should demonstrate that the proposed observerless distributed control law \eqref{districontrol} drives the attitude of the rigid body systems to converge to the desired reference.
To this end, we provide the stability properties of the synchronization error system \eqref{compactobserverless} in the following theorem.
We define the sets $\check{\mathcal{C}}^{d}$ and $\check{\mathcal{C}}^{u}$ as $\check{\mathcal{C}}^{d}=\{
(\check{R}^r=I_{3N},\check{\omega}_0=\mathbf{0}_{3N})\}$ and $\check{\mathcal{C}}^{u}=\{
(\check{R}^{r,\ast},\mathbf{0}_{3N})\in\mathrm{SO}(3)^N\times\mathbb{R}^{3N}|
\check{R}^{r,\ast}=\mathrm{diag}(\check{R}_1^{r,\ast},\ldots\check{R}_N^{r,\ast})\}$
with $\check{R}_{p}^{r,\ast}=I_3,\forall p\in\mathcal{M}_{\overline{\mathcal{G}}_0}^{I}$, $\check{R}_{q}^{r,\ast}=\mathcal{R}(\pi,\check{u}_{q}),\forall q\in\mathcal{M}_{\overline{\mathcal{G}}_0}^{\pi}$.
The sets  $\mathcal{M}_{\overline{\mathcal{G}}_0}^{I}$ and $\mathcal{M}_{\overline{\mathcal{G}}_0}^{\pi}$ satisfy $\mathcal{M}_{\overline{\mathcal{G}}_0}^{I}\cup\mathcal{M}_{\overline{\mathcal{G}}_0}^{\pi}=\mathcal{M}_{\overline{\mathcal{G}}_0}$, $\mathcal{M}_{\overline{\mathcal{G}}_0}^{\pi}\neq\emptyset$. The vector $\check{u}_{q}$ is the unit eigenvector associated to the eigenvalues of $\check{A}_{q}$.

\begin{theorem}\label{stabilityobserverless}
Consider a network of $N$ rigid body systems governed by the rotational dynamics \eqref{followeri}.
The desired attitude $R_0$ is constant.
Let the control torque be designed as \eqref{districontrol}.
Under Assumptions \ref{assumpgraph_original}, the following statements hold.
\begin{enumerate}
    \item The trajectories of the error system \eqref{compactobserverless} converge to the set of equilibria $\check{\mathcal{C}}=\check{\mathcal{C}}^{d}\cup\check{\mathcal{C}}^{u}$.
    \item The undesired equilibria in $\check{\mathcal{C}}^{u}$ are unstable, and the desired equilibrium $\check{\mathcal{C}}^{d}$ is almost globally asymptotically stable.
    \item All the rigid body attitudes synchronize (almost globally) to the desired attitude $R_0$, \textit{i.e.}, $\lim_{t\rightarrow\infty}R_i(t)=R_0$ and $\lim_{t\rightarrow\infty}\omega_i(t)=\mathbf{0}_3$ for all $i\in\mathcal{V}$, for almost all initial conditions.
\end{enumerate}
\end{theorem}

\begin{proof}
See Appendix \ref{appenF}.
\end{proof}

\section{Simulations}\label{simulation}
Consider a group of 7 rigid body systems, the interconnected topology of which is shown in Fig. \ref{interconnectiongraph}.
The desired attitude $R_0$ is accessible to rigid bodies $1$ and $2$.
The neighbor sets of agents are given as $\mathcal{N}_{\mathcal{G}}(1)=\{3\},\mathcal{N}_{\mathcal{G}}(2)=\{7\},\mathcal{N}_{\mathcal{G}}(3)=\{1,5,6\},\mathcal{N}_{\mathcal{G}}(4)=\{5\},\mathcal{N}_{\mathcal{G}}(5)=\{3,4\},\mathcal{N}_{\mathcal{G}}(6)=\{3\},$
$\mathcal{N}_{\mathcal{G}}(7)=\{2\}$.
A virtual leader rigid body capable of providing the desired attitude $R_0$ is augmented to the original graph $\mathcal{G}$, yielding an augmented graph $\mathcal{G}_0$.
In the augmented graph $\mathcal{G}_0$, a virtual direction is assigned to each pair of interconnected rigid bodies.
It is clear that the considered graphs satisfy Assumption \ref{assumpgraph_original}.


\begin{figure}[!htbp]
\centering
\footnotesize
\tikzstyle{startstop} = [rectangle, rounded corners, minimum width = 0.3cm, minimum height=0.3cm,text centered, draw = black]
\tikzstyle{test}=[diamond,aspect=2,draw,thin]
\tikzstyle{point}=[coordinate,on grid,]
\tikzstyle{line} = [draw, -latex']
\subfloat[An interconnection graph $\mathcal{G}$]{
\begin{tikzpicture}
\node[startstop,draw=red,fill=red](1){\includegraphics[width=0.02\textwidth]{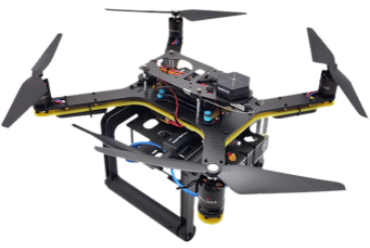}};
\node[startstop,right of=1,node distance=9mm,draw=red,fill=red](2){\includegraphics[width=0.02\textwidth]{uav.png}};
\node[startstop,right of=2,node distance=10mm,draw=white](3){\includegraphics[width=0.02\textwidth]{uav.png}};
\node[point,right of=3,node distance=8mm](point3){};
\node[startstop,below of=point3,node distance=6mm,draw=white](4){\includegraphics[width=0.02\textwidth]{uav.png}};
\node[point,left of=4,node distance=8mm](point4){};
\node[startstop,below of=point4,node distance=7mm,draw=white](5){\includegraphics[width=0.02\textwidth]{uav.png}};
\node[startstop,left of=5,node distance=9mm,draw=white](6){\includegraphics[width=0.02\textwidth]{uav.png}};
\node[startstop,left of=6,node distance=10mm,draw=white](7){\includegraphics[width=0.02\textwidth]{uav.png}};
\draw[-] (2)--(3);
\draw[-] (3)--(5);
\draw[-] (3)--(6);
\draw[-] (1)--(7);
\draw[-] (5)--(4);
\node[startstop,left of=1,node distance=4.5mm,draw=white](11){{\color{red}2}};
\node[startstop,above of=2,node distance=4.5mm,draw=white](22){{\color{red}1}};
\node[startstop,above of=3,node distance=3.5mm,draw=white](33){3};
\node[startstop,above of=4,node distance=3.5mm,draw=white](44){4};
\node[startstop,right of=5,node distance=3.5mm,draw=white](55){5};
\node[startstop,above of=6,node distance=3.5mm,draw=white](66){6};
\node[startstop,left of=7,node distance=3.5mm,draw=white](77){7};
\end{tikzpicture}
\label{interconnectiongraph}
}\quad\quad
\subfloat[An augmented oriented graph $\mathcal{G}_0$]{
\begin{tikzpicture}
\node[startstop,draw=white](0){\includegraphics[width=0.02\textwidth]{uav.png}};
\node[point,left of=0,node distance=6mm](point1){};
\node[startstop,below of=point1,node distance=8mm,draw=red,fill=red](1){\includegraphics[width=0.02\textwidth]{uav.png}};
\node[point,right of=0,node distance=6mm](point2){};
\node[startstop,below of=point2,node distance=8mm,draw=red,fill=red](2){\includegraphics[width=0.02\textwidth]{uav.png}};
\node[startstop,right of=2,node distance=9mm,draw=white](3){\includegraphics[width=0.02\textwidth]{uav.png}};
\node[point,right of=3,node distance=10mm](point3){};
\node[startstop,below of=point3,node distance=6mm,draw=white](4){\includegraphics[width=0.02\textwidth]{uav.png}};
\node[point,left of=4,node distance=9mm](point4){};
\node[startstop,below of=point4,node distance=6mm,draw=white](5){\includegraphics[width=0.02\textwidth]{uav.png}};
\node[startstop,left of=5,node distance=10mm,draw=white](6){\includegraphics[width=0.02\textwidth]{uav.png}};
\node[startstop,left of=6,node distance=12mm,draw=white](7){\includegraphics[width=0.02\textwidth]{uav.png}};
\draw[->] (0)--(1);
\draw[->] (0)--(2);
\draw[->] (2)--(3);
\draw[->] (3)--(5);
\draw[->] (3)--(6);
\draw[->] (5)--(4);
\draw[->] (1)--(7);
\node[startstop,above of=0,node distance=4mm,draw=white](00){0};
\node[startstop,left of=1,node distance=4.5mm,draw=white](11){2};
\node[startstop,above of=2,node distance=4.5mm,draw=white](22){1};
\node[startstop,above of=3,node distance=3.5mm,draw=white](33){3};
\node[startstop,above of=4,node distance=3.5mm,draw=white](44){4};
\node[startstop,right of=5,node distance=3.5mm,draw=white](55){5};
\node[startstop,above of=6,node distance=3.5mm,draw=white](66){6};
\node[startstop,left of=7,node distance=3.5mm,draw=white](77){7};
\end{tikzpicture}
\label{augmentedgraph}
}
\caption{Interconnection graphs of $7$ rigid body systems and $1$ leader }
\end{figure}

The inertia matrix $J_i$ of each rigid body is given by $J_i=\frac{1}{10}\mathrm{diag}(i,i+2,2i),i\in\{1,\ldots,7\}$.
In the subsequent simulations, the initial conditions are set as $\hat{R}_i(0)=\mathcal{R}(-\frac{i}{10}\pi,u_2)$, $\hat{\omega}_i(0)=[i\ i+0.5\ i+1]^{\top}$, $\hat{\sigma}_i(0)=[i+1\ i+0.5\ i]^{\top}$, $R_i(0)=\mathcal{R}(\frac{i}{10}\pi,u_2)$, $\omega_i(0)=[i+0.5\ i\ i+1]^{\top}$,
where $i\in\{1,\ldots,7\}$, $u_1=[0\ 1\ 0]^{\top},u_2=\frac{1}{\sqrt{2}}[1\ 1\ 0]^{\top}$.

\subsection{Observer-based attitude synchronization to a time-varying reference trajectory}

In this subsection, the desired attitude $R_0$ and the desired angular velocity $\omega_0$ are generated by \eqref{leader}, with $P$ being taken as
\[
P=\left[
\begin{array}{ccc}
-1&0&0\\
0&0&\cos(\frac{3\pi}{4})\\
0&\sin(\frac{3\pi}{4})&0\\
\end{array}
\right]_.
\]
The initial conditions of the reference trajectory are set as $\omega_0(0)=[1\ 2\ 0]^{\top},\sigma_0(0)=[1\ 0\ 0]^{\top},R_0(0)=\mathcal{R}(0.8\pi,u_0)$, where $u_0=\frac{1}{\sqrt{21}}[1\ 4\ 2]^{\top}$.
The parameters are chosen as $k_{R^l}=9,k_{\omega^l}=14,k_{\tilde{\sigma}}=12,k_{\bar{R}}=8,k_{\bar{\omega}}=13$ and
$A_{i,0}^l=\mathrm{diag}(14+i,15+i,16+i)$, $A_{i,j}^l=\mathrm{diag}(10+\max\{i,j\},11+\max\{i,j\},12+\max\{i,j\})$, $\bar{A}_i=\mathrm{diag}(10+i,11+i,12+i)$.

Fig. \ref{observer-based-time-varying} displays the simulation results of the proposed observer-based control law \eqref{control}, incorporating the distributed observer \eqref{observer}. As shown in Fig. \ref{observer-based-time-varying-a}, the attitude error between the desired and actual orientations of each rigid-body system asymptotically converges to zero. Meanwhile, Fig. \ref{observer-based-time-varying-b} validates the convergence of the angular velocity error. Collectively, these results confirm that the proposed observer-based control law \eqref{control} incorporating the distributed observer \eqref{observer} allows all rigid-body systems to achieve attitude synchronization with the desired reference trajectory.


\begin{figure}[!htbp]
\subfloat[Attitude error]{\includegraphics[width=0.25\textwidth]{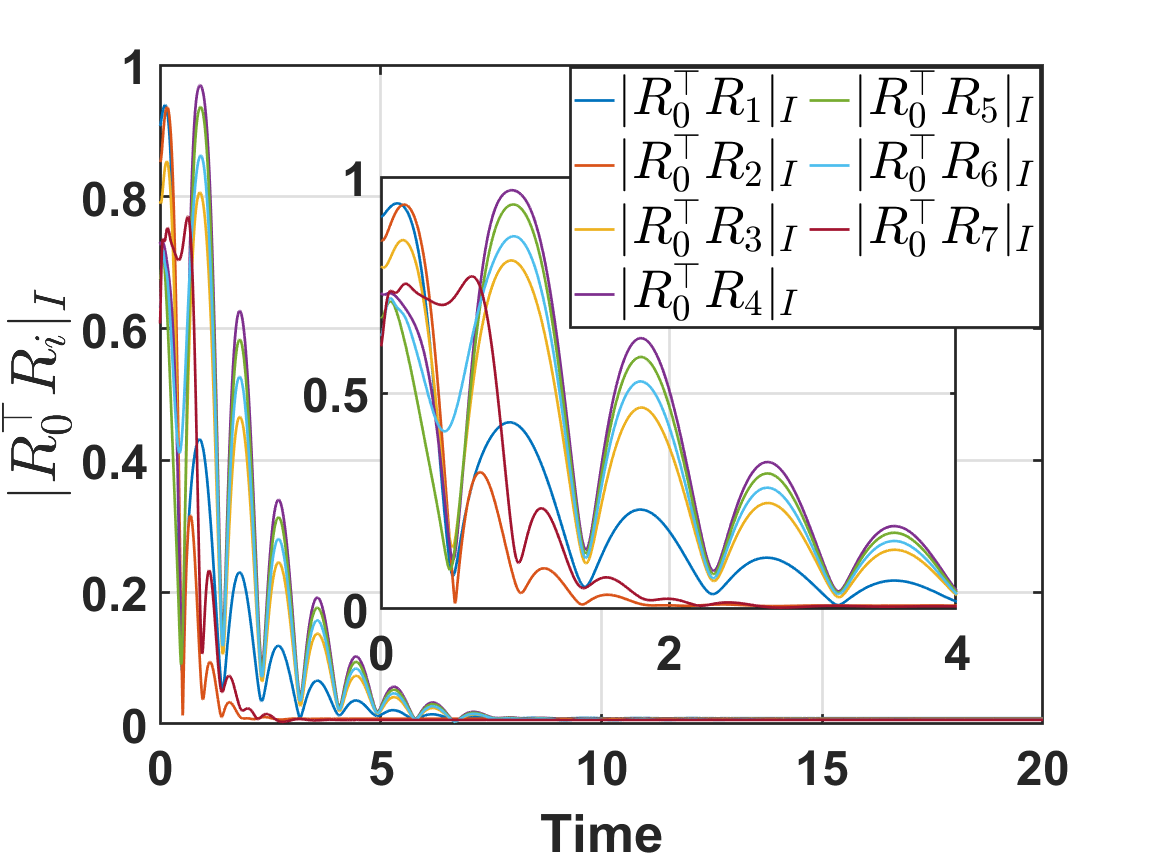}\label{observer-based-time-varying-a}}
\subfloat[Angular velocity error]{\includegraphics[width=0.25\textwidth]{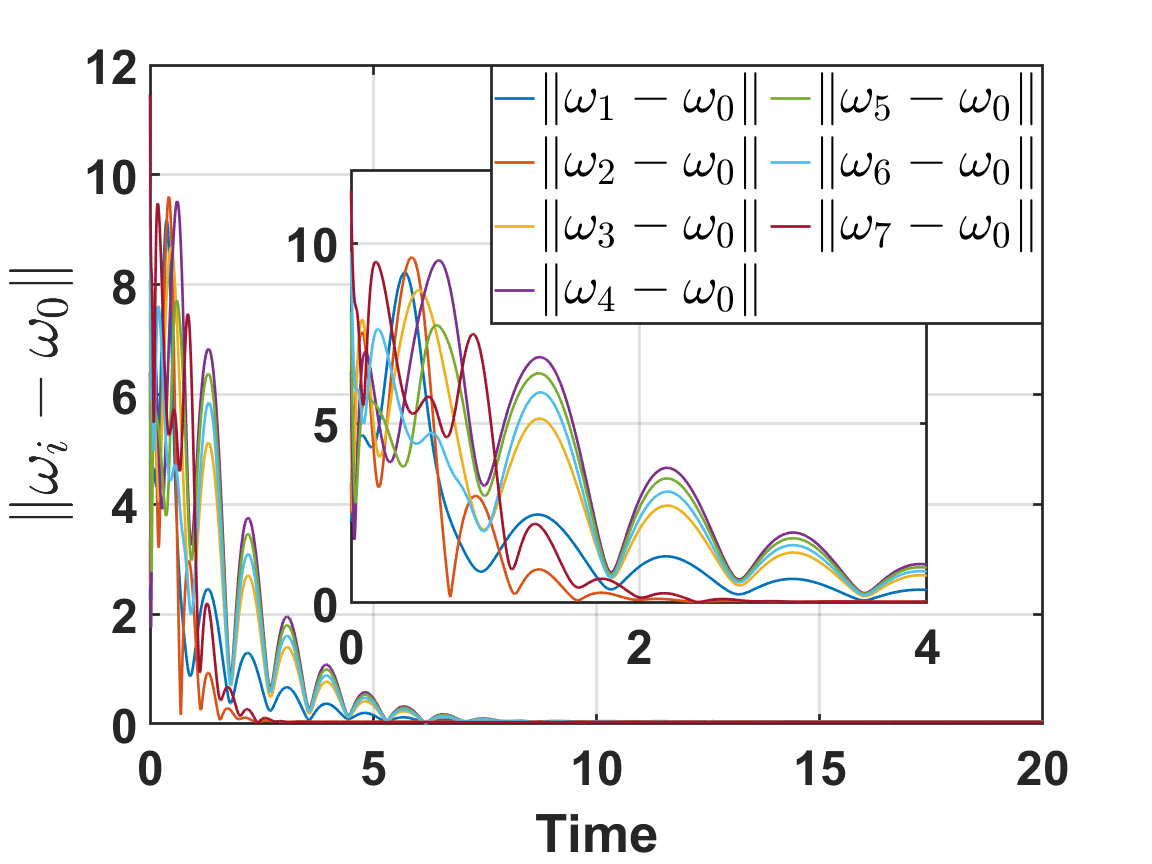}\label{observer-based-time-varying-b}}
\caption{Time evolution of the error between the leader and the follower with a time-varying reference trajectory
}
\label{observer-based-time-varying}
\end{figure}

\subsection{Attitude synchronization to a constant reference attitude}

In this subsection, the desired constant attitude $R_0$ is taken as $R_0=\mathcal{R}(0.8\pi,u_0)$.
Two groups of parameters are adopted for comparative studies, which further verifies the superior performance of the observerless distributed control scheme over the distributed observer-based control scheme.

The first group of parameters are taken as $A^r_{i,j}=A^l_{i,j}$ for $(j,i)\in\mathcal{E}_0$, $k_{R^r}=14,k_{\omega^r}=12,k_{\bar{R}}=13,k_{\bar{\omega}}=11$.
Fig. \ref{observer-based-constant} presents the simulation results of the proposed observer-based distributed control law \eqref{observerconstant} incorporating the distributed observer \eqref{controlconstantestimate}. The convergence of attitude and angular velocity errors is illustrated in Fig. \ref{observer-based-constant-a} and Fig. \ref{observer-based-constant-b}, respectively.
This sufficiently demonstrates that the proposed observer-based distributed control law \eqref{observerconstant} incorporating the distributed observer \eqref{controlconstantestimate} can drive the attitudes of all rigid-body systems to synchronize to the desired constant attitude $R_0$.


\begin{figure}[!htbp]
\subfloat[Attitude error]{\includegraphics[width=0.25\textwidth]{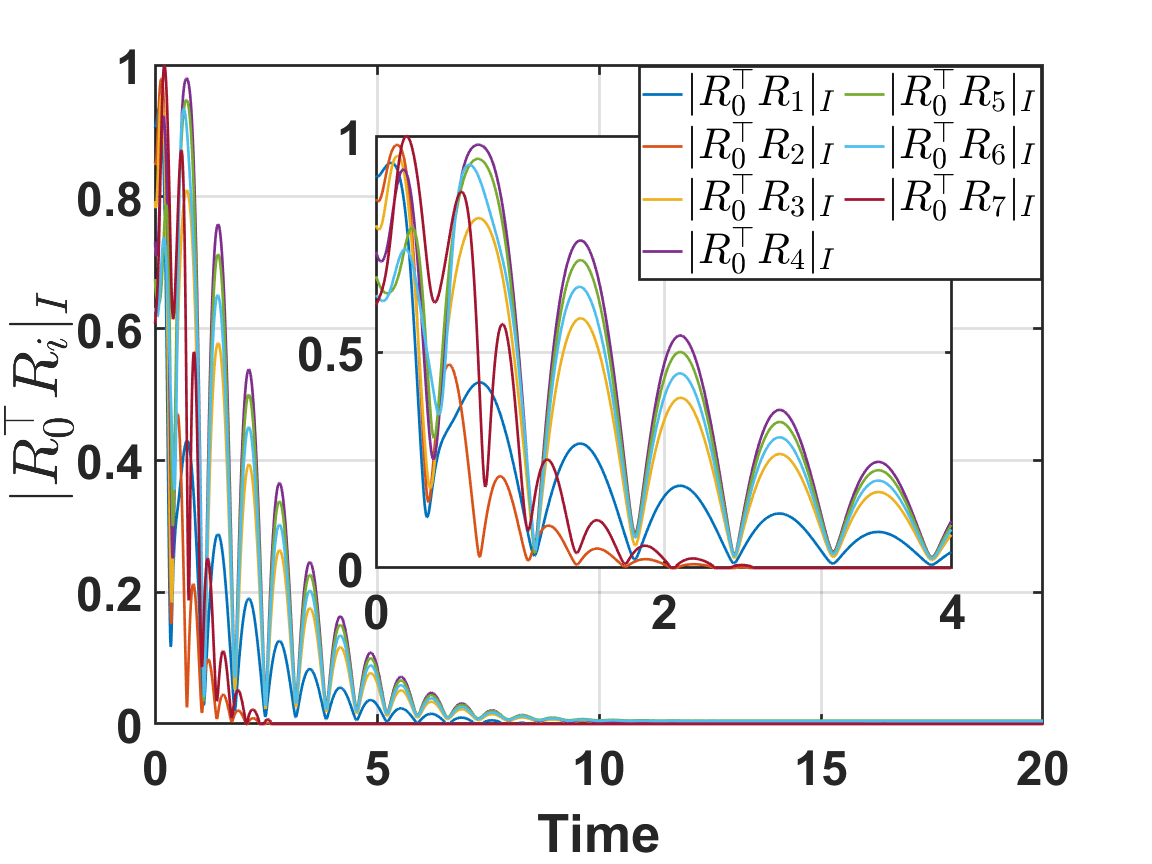}\label{observer-based-constant-a}}
\subfloat[Angular velocity error]{\includegraphics[width=0.25\textwidth]{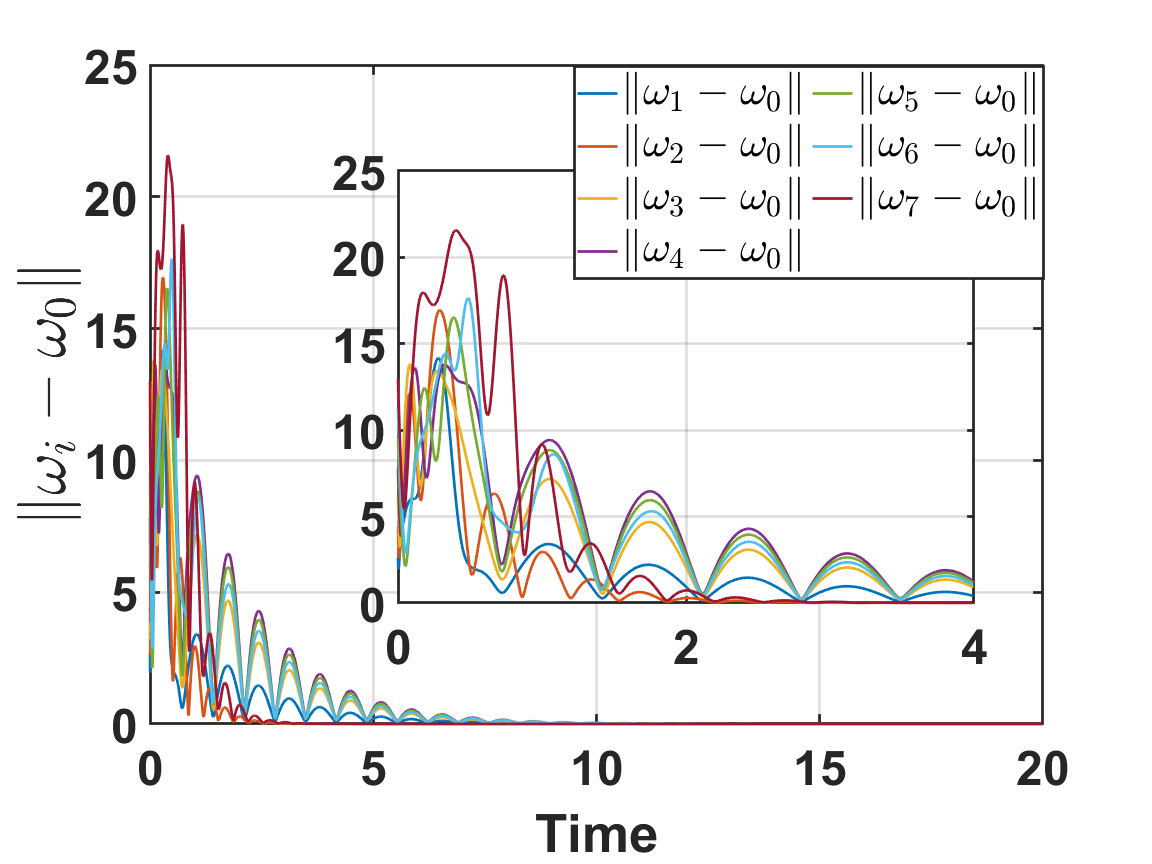}\label{observer-based-constant-b}}
\caption{Time evolution of the error between the leader and the follower with a constant reference attitude, using the proposed observer-based distributed control scheme \eqref{controlconstantestimate}}
\label{observer-based-constant}
\end{figure}

The second group of parameters are taken as $k_{\check{R}}=1.9,k_{\check{\omega}}=1.9$ and $\check{A}_{i,0}=\mathrm{diag}(5+i,6+i,7+i)$, $\check{A}_{i,j}=\mathrm{diag}(10-\min\{i.j\},9-\min\{i,j\},8-\min\{i,j\})$.
Fig. \ref{observerless-constant} illustrates the simulation results for the proposed distributed observerless control law \eqref{districontrol}.
Specifically, Fig. \ref{observerless-constant-a} and Fig. \ref{observerless-constant-b} depict the convergence of the attitude and the angular velocity errors, respectively.
The simulation results verify that the proposed distributed observerless control law \eqref{districontrol} enables all rigid body systems' orientations to synchronize to the desired constant attitude $R_0$.

\begin{figure}[!htbp]
\subfloat[Attitude error]{\includegraphics[width=0.25\textwidth]{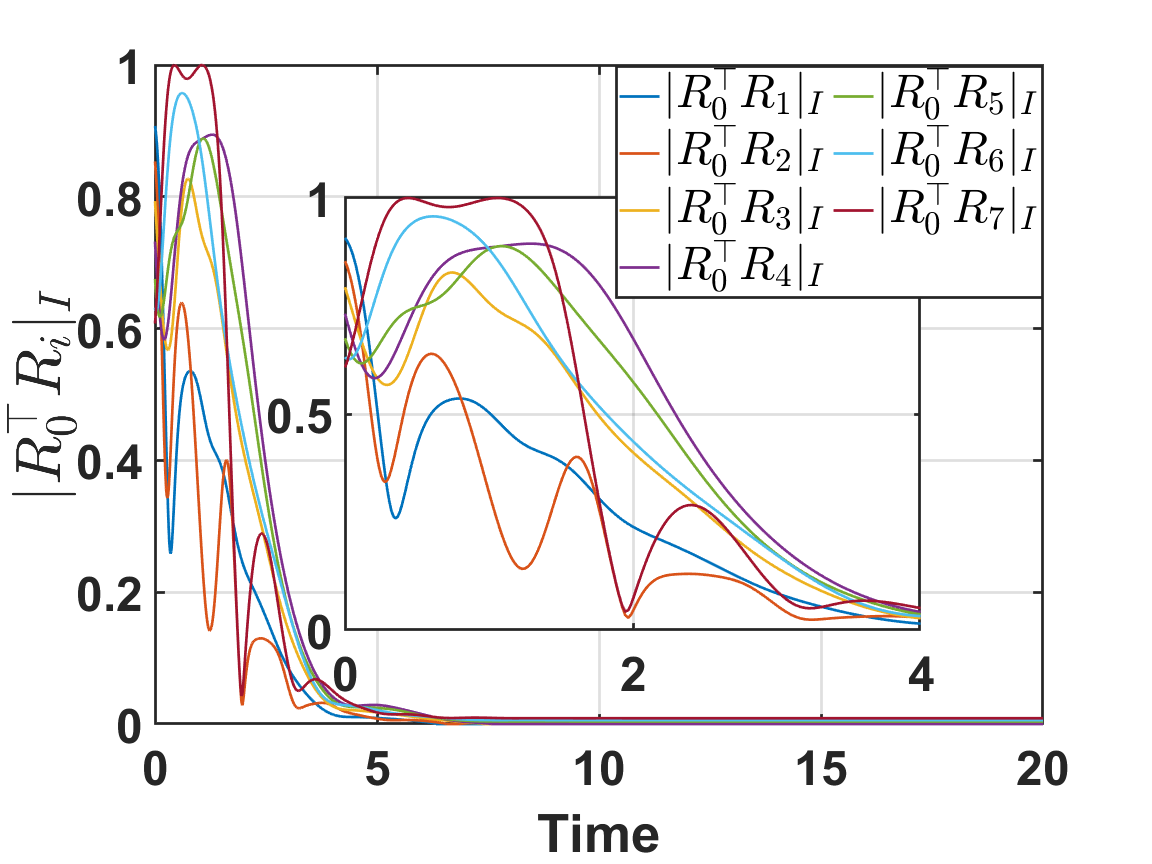}\label{observerless-constant-a}}
\subfloat[Angular velocity error]{\includegraphics[width=0.25\textwidth]{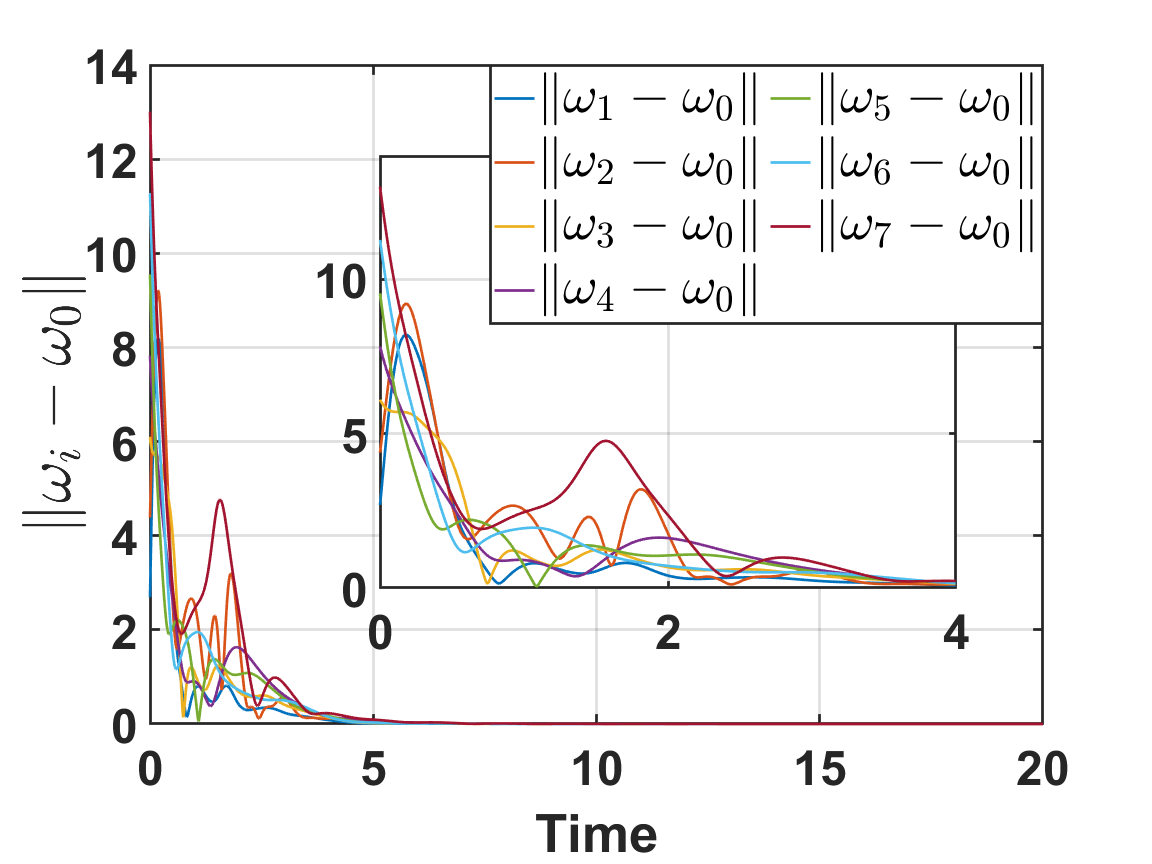}\label{observerless-constant-b}}
\caption{Time evolution of the error between the leader and the follower with a constant reference attitude, using the proposed distributed observerless control scheme \eqref{districontrol}}
\label{observerless-constant}
\end{figure}

Comparative analysis between Fig. \ref{observer-based-constant} and Fig. \ref{observerless-constant} demonstrates the superiority of the observerless control protocol \eqref{districontrol} over the observer-based counterpart \eqref{observerconstant} equipped with the distributed observer \eqref{controlconstantestimate}.
Notably, these two schemes achieve identical convergence performance, even though the observerless control protocol \eqref{districontrol} adopts significantly lower control gains.

\section{Conclusion}\label{conclusion}

We proposed two types of distributed control for the attitude synchronization on $\mathrm{SO}(3)$ for a group of heterogeneous rigid body systems over an undirected and acyclic communication graph topology: an observerless control strategy that allows to align all the attitudes to a constant reference attitude, and an observer-based control strategy that allows to handle both constant and time-varying reference attitudes. The reference attitude, represented by a virtual leader, is only available to a subset of the follower agents and the augmented graph including the virtual leader contains a directed spanning tree rooted at the virtual leader.
Since only a subset of the agents have access to the reference trajectory, all remaining agents need to estimate the desired reference states (attitude, angular velocity and angular acceleration) in the time-varying reference case. In the constant reference case, the estimation of the desired states by the followers is not necessary, and the design of a simplified observerless distributed control scheme is possible.
As mentioned in Remark \ref{re1}, the compensation of the nonlinear term $\Sigma_i\bar{\omega}_i$ in the control torque $\tau_i$ in \eqref{control} is not necessary since the matrix $\Sigma_i$ is skew symmetric, resulting in a passive term that does not affect the Lyapunov time-derivative. This compensation is performed for technical reasons that allow to prove our claimed almost global asymptotic stability. Without this compensations, the resulting control torque in \eqref{controlconstantestimate} would guarantee local exponential stability. Almost global asymptotic stability with the control torque in \eqref{controlconstantestimate} is not excluded and a formal proof will be part of our future investigations.
It is worth noting that almost global asymptotic stability is the strongest stability performance achievable via smooth feedback on $\mathrm{SO}(3)$.
This is fundamentally attributed to the topological properties of $\mathrm{SO}(3)$,  precluding the existence of any smooth feedback control law that stabilizes globally the attitude. Another future research direction consists in the design of a distributed hybrid control strategy for the leader-follower attitude synchronization, endowed with global asymptotic stability guarantees.

\appendices

\section{Proof of Lemma \ref{le1}} \label{appenA}

This result is proven by showing that the eigenvalues of $\Xi+LL^{\top}$ are positive.
To this end, we first prove that $\Xi+LL^{\top}$ is nonsigular using the result proven in \cite{shivakumar1974sufficient}, and then point out that the eigenvalues of $\Xi+LL^{\top}$ are nonnegative using Ger{\v{s}}gorin disc theorem \cite{horn2012matrix}.

It is clear that $(LL^{\top})_{ii}=\sum\limits_{i\neq j}|(LL^{\top})_{ij}|$ for $i\in\mathcal{V}$ since $LL^{\top}$ is the Laplacian matrix of the graph $\mathcal{G}$.
According to \eqref{alpha}, one has $(\Xi+LL^{\top})_{ii}>\sum\limits_{i\neq j}|(\Xi+LL^{\top})_{ij}|$ for $i\in\{1,\ldots,n\}$ and $(\Xi+LL^{\top})_{ii}=\sum\limits_{i\neq j}|(\Xi+LL^{\top})_{ij}|$ for $i\in\{n+1,\ldots,N\}$.
Hence, $\Xi+LL^{\top}$ is diagonally dominant with at least one $i\in\mathcal{V}$ such that $(\Xi+LL^{\top})_{ii}>\sum\limits_{i\neq j}|(\Xi+LL^{\top})_{ij}|$.

For each $j\in\{n+1,\ldots,N\}$, if there exists $k_1\in\{1,\ldots,n\}$ such that $j\in\mathcal{N}_{\mathcal{G}}(k_1)$, then $(\Xi+LL^{\top})_{jk_1}=-1\neq0$ due to the fact that $LL^{\top}$ is the Laplacian matrix and $\Xi$ is a diagonal matrix.
If no $i\in\{1,\ldots,n\}$ satisfying $j\in\mathcal{N}_{\mathcal{G}}(i)$ exists, then according to Assumption \ref{assumpgraph_original}, one can find a path $(0,k_1,k_2,\ldots,k_m,j)$, in the oriented graph $\overline{\mathcal{G}}_0$, from the leader $0$ to the $j$-th rigid body system, where $k_1\in\{1,\ldots,n\}$.
This means that, in the graph $\mathcal{G}$, there exists a path $(k_1,k_2,\ldots,k_m,j)$ from the $k_1$-th rigid body system to the $j$-th rigid body system, leading to
\begin{align*}
&(\Xi+LL^{\top})_{k_1k_2}=\cdots=(\Xi+LL^{\top})_{k_{m-1}k_m}\\
&=(\Xi+LL^{\top})_{k_mj}=-1\neq0.
\end{align*}
This, together with the fact that $\Xi+LL^{\top}$ is diagonally dominant with at least one $i\in\mathcal{V}$ such that $(\Xi+LL^{\top})_{ii}>\sum\limits_{i\neq j}|(\Xi+LL^{\top})_{ij}|$, implies that $\Xi+LL^{\top}$ is non-singular as per the result proven in \cite{shivakumar1974sufficient}.

As per Ger{\v{s}}gorin disc theorem, all eigenvalues of $\Xi+LL^{\top}$ are in the union of Ger{\v{s}}gorin disc
\[\bigcup_{i=1}^{N}\{\lambda\in\mathbb{R}||\lambda-(\Xi+LL^{\top})_{ii}|\leq\sum_{i\neq j}|(\Xi+LL^{\top})_{ij}|\}.\]
Due to the fact that $(\Xi+LL^{\top})_{ii}\geq \sum_{i\neq j}|(\Xi+LL^{\top})_{ij}|$, one can conclude that the eigenvalues of $\Xi+LL^{\top}$ are nonnegative.
Therefore, the eigenvalues of $\Xi+LL^{\top}$ are positive since $\Xi+LL^{\top}$ is non-singular.

\section{Proof of Lemma \ref{le2}}\label{appenB}

To prove that $\check{\mathcal{L}}_2^l$ is non-singular, we show that $\check{\mathcal{L}}_2^lx=\mathbf{0}_{3N}$ has unique solution $x=\mathbf{0}_{3N}$.
Assume that $x=[x_1^{\top}\ \cdots\ x_{N}^{\top}]^{\top}\neq\mathbf{0}_{3N}$ is one of the solutions of $(\check{\mathcal{L}}_2^l)^{\top}x=\mathbf{0}_{3N}$, where $x_i\in\mathbb{R}^3$, $i\in\mathcal{V}$.
Since the desired attitude $R_0$ is available to the rigid body systems $i=1,\ldots,n$, one has $(\check{\mathcal{L}}^l)_{i+1,k_i}\neq\mathbf{0}_{3\times 3}$ and $(\check{\mathcal{L}}^l)_{l+1,k_i}=\mathbf{0}_{3\times 3}$ for $l\in\mathcal{V}\setminus\{i\}$, which in turn implies that $\sum\limits_{l=1}^{N}((\check{\mathcal{L}}^l)_{l+1,k_i})^{\top}x_l=((\check{\mathcal{L}}^l)_{i+1,k_i})^{\top}x_i=\mathbf{0}_{3}$ for $i\in\{1,\ldots,n\}$, where the edge $(0,i)$ is denoted by  $k_i$.
From \eqref{checkL}, the block $(\check{\mathcal{L}}^l)_{jk}$ with $(\check{\mathcal{L}}^l)_{jk}\neq\mathbf{0}_{3\times 3}$ is non-singular.
Hence, one has $x_i=\mathbf{0}_3$ for $i\in\{1,\ldots,n\}$.

According to Assumption \ref{assumpgraph_original}, for any $i\in\{n+1,\ldots,N\}$, there exists a path $(0,j_1,\ldots,j_m,i)$ from the leader $0$ to the $i$-th rigid body system, where $j_1\in\{1,\ldots,n\}$.
Denoting the edge $(j_{d},j_{d+1})$ by $k_{j_{d+1}}$ for $d\in\{1,\ldots,m-1\}$, one has $(\check{\mathcal{L}}^l)_{j_d+1,k_{j_{d+1}}}\neq\mathbf{0}_{3\times3},(\check{\mathcal{L}}^l)_{j_{d+1}+1,k_{j_{d+1}}}\neq\mathbf{0}_{3\times 3}$, and $(\check{\mathcal{L}}^l)_{l+1,k_{j_{d+1}}}=\mathbf{0}_{3\times3}$ for $l\in\mathcal{V}\setminus\{j_d,j_{d+1}\}$, implying that
\[
\begin{array}{ll}
\sum\limits_{l=1}^{N}((\check{\mathcal{L}}^l)_{l+1,k_{j_{d+1}}})^{\top}x_{l}=&((\check{\mathcal{L}}^l)_{j_d+1,k_{j_{d+1}}})^{\top}x_{j_d}\\
\ &+((\check{\mathcal{L}}^l)_{j_{d+1}+1,k_{j_{d+1}}})^{\top}x_{j_{d+1}}=\mathbf{0}_3.
\end{array}
\]
Using the fact again that the block $(\check{\mathcal{L}}^l)_{jk}$ with $(\check{\mathcal{L}}^l)_{jk}\neq\mathbf{0}_{3\times 3}$ is non-singular, one has
\[
x_{j_{d+1}}=-(\check{\mathcal{L}}^l)_{j_{d+1}+1,k_{j_{d+1}}}((\check{\mathcal{L}}^l)_{j_d+1,k_{j_{d+1}}})^{\top}x_{j_d}
\]
for $d\in\{1,\ldots,m-1\}$,
leading to
\begin{align*}
&x_{j_{m}}=(-1)^{m+1}(\check{\mathcal{L}}^l)_{j_{m}+1,k_{j_{m}}}\cdots((\check{\mathcal{L}}^l)_{j_{1}+1,k_{j_{2}}})^{\top}x_{j_1}.
\end{align*}
This, with $x_{j_1}=\mathbf{0}_3$, leads to $x_{j_m}=\mathbf{0}_3$.
Let the edge $(j_m,i)$ be denoted as the $k_i$-th edge.
It is clear that $(\check{\mathcal{L}}^l)_{j_m+1,k_i}\neq\mathbf{0}_{3\times 3},(\check{\mathcal{L}}^l)_{i+1,k_i}\neq\mathbf{0}_{3\times 3}$, and $(\check{\mathcal{L}}^l)_{l+1,k_i}=\mathbf{0}_{3\times 3}$ for $l\in\mathcal{V}\setminus\{j_m,i\}$.
Then one obtains
\[
\begin{array}{rcl}
\sum\limits_{l=1}^N((\check{\mathcal{L}}^l)_{l+1,k_i})^{\top}x_l&=&((\check{\mathcal{L}}^l)_{j_m+1,k_i})^{\top}x_{j_{m}}\\
\ &\ &+((\check{\mathcal{L}}^l)_{i+1,k_i})^{\top}x_i=\mathbf{0}_3,
\end{array}
\]
 resulting in $((\check{\mathcal{L}}^l)_{i+1,k_i})^{\top}x_i=\mathbf{0}_3$. It follows that $x_{i}=\mathbf{0}_3$ due to the fact that the block $(\check{\mathcal{L}}^l)_{jk}$ with $(\check{\mathcal{L}}^l)_{jk}\neq\mathbf{0}_{3\times 3}$ is non-singular.
Consequently, one concludes that $x=\mathbf{0}_{3N}$, which is a contradiction.
Therefore, $x=\mathbf{0}_{3N}$ is the unique solution of $\check{\mathcal{L}}^l_2x=\mathbf{0}_{3N}$, implying that $\check{\mathcal{L}}_2^l$ is non-singular.

Since the same process can be conducted to show that $\check{\mathcal{L}}_2^r$ is non-singular, we omit the details here.

\section{\ \ \ \ \ \ Lemma \ref{le3}}\label{appenC}

\begin{lemma}\label{le3}
Let $\check{\mathcal{L}}^l,\check{\mathcal{L}}^r\in\mathbb{R}^{3(N+1)\times 3N}$ be associated to the augmented graph $\mathcal{G}_0$ with the blocks $\check{\mathcal{L}}_{ik}^l,\check{\mathcal{L}}_{ik}^r$ defined in \eqref{checkL} and \eqref{checkLe} respectively.
Then the following results hold:
\begin{enumerate}
\item Let $\check{A}_k^l=\check{A}_{i,j}^l$ for $i,j\in\mathcal{V}_0$ with the edge between the $i$-th rigid body system and the $j$-th rigid body system being described by edge $k$. One has
\begin{align*}
&R_0^{\top}\psi(\check{A}_{i,0}^lR_iR_0^{\top})+\sum\limits_{j\in\mathcal{N}_{\mathcal{G}}(i)}R_j^{\top}\psi(\check{A}_{i,j}^lR_iR_j^{\top})\\
&=R_i^{\top}(\sum\limits_{k=1}^{N}(\check{\mathcal{L}}^l_2)_{i,k}\psi(\check{A}_k^l\check{R}_{k}^l)), ~i\in\mathcal{V}.
\end{align*}
\item Let $\check{A}_k^r=\check{A}^r_{i,j}$ for $i,j\in\mathcal{V}_0$ with the edge between the $i$-th rigid body system and the $j$-th rigid body system being described by edge $k$. One has
\begin{align*}
&\psi(\check{A}^r_{i,0}R_0^{\top}R_i)+\sum\limits_{j\in\mathcal{N}_{\mathcal{G}}(i)}\psi(\check{A}_{i,j}^rR_j^{\top}R_i)\\
&=\sum\limits_{k=1}^{N}(\check{\mathcal{L}}^r_2)_{i,k}\psi(\check{A}^r_k\check{R}^r_k), ~i\in\mathcal{V}.
\end{align*}
\end{enumerate}
\end{lemma}

\begin{proof}
Let us prove the first item.
Denote $\mathcal{I}_{\overline{\mathcal{G}}_0}(i)=\{j\in\mathcal{N}_{\mathcal{G}_0}(i)|(j,i)\in\overline{\mathcal{E}}_0\}$ and $\mathcal{O}_{\overline{\mathcal{G}}_0}(i)=\{j\in\mathcal{N}_{\mathcal{G}_0}(i)|(i,j)\in\overline{\mathcal{E}}_0\}$.
Using the facts $\psi(AR)=-\psi(R^{\top}A)$ and $\psi(AR)=R^{\top}\psi(RA)$, $A=A^{\top}\in\mathbb{R}^{3\times 3}, R\in \mathrm{SO}(3)$, we handle the term
$R_0^{\top}\psi(\check{A}_{i,0}^lR_iR_0^{\top})+\sum\limits_{j\in\mathcal{N}_{\mathcal{G}}(i)}R_j^{\top}\psi(\check{A}_{i,j}^lR_iR_j^{\top})$ as in \cite{bough_Berk_Tay_Arxiv2025},
\[\small
\begin{array}{l}
R_0^{\top}\psi(\check{A}_{i,0}^lR_iR_0^{\top})+\sum\limits_{j\in\mathcal{N}_{\mathcal{G}}(i)}R_j^{\top}\psi(\check{A}_{i,j}^lR_iR_j^{\top})\\
=\sum\limits_{j\in\mathcal{I}_{\overline{\mathcal{G}}_0}(i)}R_j^{\top}\psi(\check{A}_{i,j}^lR_iR_j^{\top})+\sum\limits_{j\in\mathcal{O}_{\overline{\mathcal{G}}_0}(i)}R_j^{\top}\psi(\check{A}_{i,j}^lR_iR_j^{\top})\\
=\sum\limits_{j\in\mathcal{I}_{\overline{\mathcal{G}}_0}(i)}R_j^{\top}\psi(\check{A}_{i,j}^lR_iR_j^{\top})-\sum\limits_{j\in\mathcal{O}_{\overline{\mathcal{G}}_0}(i)}R_j^{\top}\psi(R_jR_i^{\top}\check{A}_{i,j}^l)\\
=\sum\limits_{j\in\mathcal{I}_{\overline{\mathcal{G}}_0}(i)}R_j^{\top}\psi(\check{A}_{i,j}^lR_iR_j^{\top})-\sum\limits_{j\in\mathcal{O}_{\overline{\mathcal{G}}_0}(i)}R_j^{\top}R_jR_i^{\top}\psi(\check{A}_{i,j}^lR_jR_i^{\top})\\
=R_i^{\top}(\sum\limits_{j\in\mathcal{I}_{\overline{\mathcal{G}}_0}(i)}R_iR_j^{\top}\psi(\check{A}_{i,j}^lR_iR_j^{\top})-\sum\limits_{j\in\mathcal{O}_{\overline{\mathcal{G}}_0}(i)}\psi(\check{A}_{i,j}^l\hat{R}_jR_i^{\top}))\\
=R_i^{\top}(\sum\limits_{k=1}^{N}(\check{\mathcal{L}}^l_2)_{i,k}\psi(\check{A}_k^l\check{R}_{k}^l))
\end{array}
\]
for $i\in\mathcal{V}$.

Since the proof of the second item is similar to that of the first item, we omit the detials here.
\end{proof}

\section{Proof of Theorem \ref{estimationstability}}\label{appenD}

The following lemma is required for the stability proof.
\begin{lemma}\label{le2.1}\cite{BerkaneSoulaimanephdtheis}
Consider the trajectory $\dot{R}=R\omega^{\times}$ with $R(0)\in \mathrm{SO}(3)$ and $\omega\in\mathbb{R}^3$.
Let $A\in\mathbb{R}^{3\times 3}$ be positive definite.
Then the following statements hold:
\begin{equation}\label{e2.10}
\frac{d}{dt}\mathrm{tr}(A(I-R))=2\omega^{\top}\psi(AR),
\end{equation}
\begin{equation}\label{e2.1}
4\lambda_{\min}^{\mathbf{E}(A)}|R|_I^2\leq\mathbf{tr}(A(I-R))\leq 4\lambda_{\max}^{\mathbf{E}(A)}|R|_I^2,
\end{equation}
\begin{equation}\label{e2.5}
\|\mathbf{E}(AR)\|_F\leq\|\mathbf{E}(A)\|_F,
\end{equation}
\begin{equation}\label{e2.3}
\|\psi(AR)\|^2=\alpha(A,R)\mathbf{tr}(\mathbf{E}^{-1}(\mathbf{E}(A)^2)(I-R)),
\end{equation}
\begin{equation}\label{e2.4}
1-|R|_I^2\leq\alpha(A,R)\leq(1-\frac{\lambda_{\min}^{\mathbf{E}(A)}}{\lambda_{\max}^{\mathbf{E}(A)}}|R|_I^2),
\end{equation}
where $\alpha(A,R)=(1-|R|_I^2\cos^2(u,E(A)u))$, $u\in\mathbb{S}^2$ is the aixs of the rotation $R$.
\end{lemma}

Now, we prove Theorem \ref{estimationstability}.
Let us prove the first item. Consider the following Lyapunov function candidate:
\[
\setlength{\arraycolsep}{1pt}
\begin{array}{rcl}
V_1&=&\frac{1}{2}(k_{R^l}\mathbf{tr}(\mathbf{A}^l(I_{3N}-\tilde{R}^l))+(\tilde{\omega}_{0}^l)^{\top}\tilde{\omega}_0^l+\tilde{\sigma}_0^{\top}\tilde{\sigma}_0\\
\ &\ &+k_{\bar{R}}\mathbf{tr}(\bar{\mathbf{A}}(I_{3N}-\bar{R}^r))+\bar{\omega}^{\top}J\bar{\omega}),
\end{array}
\]
where $\mathbf{A}^l=\mathrm{diag}(A_1^l,\ldots,A_{N}^l)$ and $\bar{\mathbf{A}}=\mathrm{diag}(\bar{A}_1,\ldots,\bar{A}_N)$.
Using the fact that $\lambda_{\min}^{\mathbf{E}(A)}=\frac{1}{2}(\mathbf{tr}(A)-\lambda_{\max}^{A^{\top}}),\forall A=A^{\top}>0$, one infers $\lambda_{\min}^{\mathbf{E}(A_k^l)}>0$ and $\lambda_{\min}^{\mathbf{E}(\bar{A}_i)}>0$ since
 $A_k^l$ and $\bar{A}_i$ are positive definite.
This, with \eqref{e2.1}, leads to
$\mathbf{tr}(A_k^l(I_3-\tilde{R}_k^l))\geq4\lambda_{\min}^{\mathbf{E}(A_k^l)}|\tilde{R}_k^l|_I^2$ for $k\in\mathcal{M}_{\overline{\mathcal{G}}_0}$ and $\mathbf{tr}(\bar{A}_i(I_3-\bar{R}_i^r))\geq\lambda_{\min}^{\mathbf{E}(\bar{A}_i)}|\bar{R}_i^r|_I^2$ for $i\in\mathcal{V}$, implying that $V_1$ is positive definite. From \eqref{e2.10}, one has
\begin{align*}
\frac{d}{dt}\mathbf{tr}(A_k^l(I_3-\tilde{R}_k^l))&=2(\hat{R}_j\tilde{\omega}_k^l)^{\top}\psi(A_k^l\tilde{R}_k^l),
~~k\in\mathcal{M}_{\overline{\mathcal{G}}_0},\\
 \frac{d}{dt}\mathbf{tr}(\bar{A}_i(I_3-\bar{R}_i^r))&=2\bar{\omega}_i^{\top}\psi(\bar{A}_i\bar{R}_i^r), ~~i\in\mathcal{V}.
\end{align*}
The time-derivative of $V_1$, along the trajectories of systems \eqref{errorcompact} and \eqref{controlerror1}, is given by
\begin{equation}\label{dotV1}
\begin{array}{rcl}
\dot{V}_1&=&k_{R^l}\sum\limits_{k=1}^{N}(\tilde{\omega}_k^l)^{\top}\hat{R}_j^{\top}\psi(A_k^l\tilde{R}_k^l)+(\tilde{\omega}_0^l)
^{\top}(\tilde{\sigma}_0\\
\ &\ &-k_{R^l}\hat{R}^{\top}\tilde{\mathcal{L}}_2^l\Psi^l-k_{\omega^l}((\Xi+LL^{\top})\otimes I_3)\tilde{\omega}_0^l\\
\ &\ &+\tilde{\sigma}_0^{\top}(I_N\otimes P-k_{\tilde{\sigma}}(\Xi+LL^{\top})\otimes I_3)\tilde{\sigma}_0\\
\ &\ &+k_{\bar{R}}\sum\limits_{i=1}^{N}\bar{\omega}_i^{\top}\psi(\bar{A}_i\bar{R}_i^r)+\bar{\omega}^{\top}(\Sigma\bar{\omega}+\Gamma+\tau),\\
\end{array}
\end{equation}
where $\tau=[\tau_1^{\top}\ \cdots\ \tau_N^{\top}]^{\top}$, 
$\Gamma=[\Gamma_1^{\top}\ \cdots\ \Gamma_N^{\top}]^{\top}$.
Using the facts $\psi(AR)=-\psi(R^{\top}A)$ and $\psi(AR)=R^{\top}\psi(RA)$, $A=A^{\top}\in\mathbb{R}^{3\times 3}, R\in \mathrm{SO}(3)$,
we can express the term $(\hat{R}_j\tilde{\omega}_k^l)^{\top}\psi(A_k^l\tilde{R}_k^l)$ as follows:
\[
\setlength{\arraycolsep}{1pt}
\begin{array}{rcl}
(\hat{R}_j\tilde{\omega}_k^l)^{\top}\psi(A_k^l\tilde{R}_k^l)&=&\frac{1}{2}(\hat{\omega}_i-\hat{\omega}_j)^{\top}\hat{R}_j^{\top}\psi(A_{i,j}^l\tilde{R}_{i,j}^l)\\
\ &\ &+\frac{1}{2}(\hat{\omega}_j-\hat{\omega}_i)^{\top}\hat{R}_i^{\top}\psi(A_{j,i}^l\tilde{R}_{j,i}^l),\\
\ &=&\hat{\omega}_i^{\top}\hat{R}_j^{\top}\psi(A_{i,j}^l\tilde{R}_{i,j}^l)+\hat{\omega}_j^{\top}\hat{R}_i^{\top}\psi(A_{j,i}^l\tilde{R}_{j,i}^l), \\
\end{array}
\]
leading to
\begin{align*}\label{comdottr}
&\sum\limits_{k=1}^{N}(\tilde{\omega}_{k}^l)^{\top}(\hat{R}_j^l)^{\top}\psi(A_k^l\tilde{R}_k^l)
=\sum\limits_{i=1}^n(\tilde{\omega}_{i,0}^l)^{\top}R_0^{\top}\psi(A_{i,0}^l\tilde{R}_{i,0}^l)\\
& +\frac{1}{2}\sum\limits_{i=1}^{N}\sum\limits_{j\in\mathcal{N}_{\mathcal{G}}(i)}(\tilde{\omega}_{i,0}^l-\tilde{\omega}_{j,0}^l)^{\top}\hat{R}_j^{\top}\psi(A_{i,j}^l\tilde{R}_{i,j}^l)\\
&=\sum\limits_{i=1}^n(\tilde{\omega}_{i,0}^l)^{\top}R_0^{\top}\psi(A_{i,0}^l\tilde{R}_{i,0}^l)\\
&+\sum\limits_{i=1}^{N}(\tilde{\omega}_{i,0}^l)^{\top}\sum\limits_{j\in\mathcal{N}_{\mathcal{G}}(i)}\hat{R}_j^{\top}\psi(A_{i,j}^l\tilde{R}_{i,j}^l)\\
&=(\tilde{\omega}_0^l)^{\top}\hat{R}^{\top}\tilde{\mathcal{L}}_2^l\Psi^l.
\end{align*}
Using the previous equality and \eqref{control} in \eqref{dotV1}, one obtains
\[
\setlength{\arraycolsep}{1pt}
\begin{array}{rcl}
\dot{V}_1&=&
-k_{\bar{\omega}}\|\bar{\omega}\|^2-[(\tilde{\omega}_0^l)^{\top}\
\tilde{\sigma}_0^{\top}]
\mathbf{F}
\left[
\begin{array}{c}
\tilde{\omega}_0^l\\
\tilde{\sigma}_0
\end{array}
\right]\\
\end{array}
\]
where
\[
\mathbf{F}=\left[
\begin{array}{cc}
k_{\omega^l}\mathbf{F}_1&-\frac{1}{2}I_{3N}\\
-\frac{1}{2}I_{3N}&\mathbf{F}_2\\
\end{array}
\right]
\] with
$\mathbf{F}_1=(\Xi+LL^{\top})\otimes I_3$ and $\mathbf{F}_2=k_{\tilde{\sigma}}\mathbf{F}_1-\frac{1}{2}(I_N\otimes (P+P^{\top}))
$. It is clear that $\mathbf{F}_1$ is non-singular as per Lemma \ref{le1}.
Then one derives
\[
\overline{\mathbf{F}}=\mathrm{diag}(k_{\omega^l}\mathbf{F}_1,\overline{\mathbf{F}}_2)=\mathbf{H}^{\top}\mathbf{F}\mathbf{H},
\]
where $\overline{\mathbf{F}}_2=-\frac{1}{4k_{\omega^l}}\mathbf{F}_1^{-1}+\mathbf{F}_2$,
\[\mathbf{H}=\left[
\begin{array}{cc}
I_{3N}&\frac{1}{2k_{\omega^l}}\mathbf{F}_1^{-1}\\
\mathbf{0}_{3N\times 3N}&I_{3N}\\
\end{array}
\right]_.\]
In what follows, we will show that $\mathbf{F}$ is positive definite by proving that all eigenvalues of $\overline{\mathbf{F}}$ are positive, leveraging the congruence between $\mathbf{F}$ and $\overline{\mathbf{F}}$.
It follows from Lemma \ref{le1} that
$\lambda_{\min}^{\mathbf{F}_1}>0$ and $\lambda_{\min}^{\mathbf{F}_1^{-1}}>0$.
So we take
\[k_{\tilde{\sigma}}>\max\left\{0,\frac{\lambda_{\min}^{\mathbf{F}_1^{-1}}-4k_{\omega^l}\lambda_{\min}^{-\frac{1}{2}(I_{N}\otimes(P+P^{\top}))}}{4k_{\omega^l}\lambda_{\min}^{\mathbf{F}_1}}\right\}.\]
Using Weyl inequality yields
\[\begin{array}{rcl}
\lambda_{\min}^{\overline{\mathbf{F}}_2}&\geq&k_{\tilde{\sigma}}\lambda_{\min}^{\mathbf{F}_1}-\frac{1}{4k_{\omega^l}}\lambda_{\min}^{\mathbf{F}_1^{-1}}+\lambda_{\min}^{-\frac{1}{2}(I_{N}\otimes(P+P^{\top}))}.\\
\end{array}
\]
If $-\frac{1}{4k_{\omega^l}}\lambda_{\min}^{\mathbf{F}_1^{-1}}+\lambda_{\min}^{-\frac{1}{2}(I_N\otimes (P+P^{\top}))}\geq0$, then $\lambda_{\min}^{\mathbf{F}_2}> 0$.
If $-\frac{1}{4k_{\omega^l}}\lambda_{\min}^{\mathbf{F}_1^{-1}}+\lambda_{\min}^{-\frac{1}{2}(I_N\otimes (P+P^{\top}))}<0$, then $k_{\tilde{\sigma}}\lambda_{\min}^{\mathbf{F}_1}>\frac{1}{4k_{\omega^l}}\lambda_{\min}^{\mathbf{F}_1^{-1}}-\lambda_{\min}^{-\frac{1}{2}(I_N\otimes (P+P^{\top}))}$, leading to $\lambda_{\min}^{\overline{\mathbf{F}}_2}>0$.
Consequently, all eigenvalues of $\overline{\mathbf{F}}_2$ are positive.
Since
the eigenvalues of $\overline{\mathbf{F}}$ consist of the eigenvalues of $k_{\omega^l}\mathbf{F}_1$ and $\overline{\mathbf{F}}_2$, all eigenvalues of $\overline{\mathbf{F}}$ are positive, resulting in $\mathbf{F}$ to be positive definite.
This allows us to conclude that $\dot{V}_1\leq 0$.
Consequently, the desired equilibrium $\mathcal{C}^{l,d}$ is stable.

We invoke Barbalat Lemma to show that the trajectories of the error systems \eqref{errorcompact} and $\eqref{controlerrorcompact}$ converge to $\mathcal{C}^l$.
We first prove that $\dot{V}_1$ converges to $0$ using Barbalat lemma.
It follows from $\dot{V}_1\leq 0$ that $V_1$ is non-increasing. 
This, with the positive definitness of $V_1$, implies that $\lim_{t\rightarrow\infty}V_1$ exists and is finite.
The second time-derivative of $V_1$ is as follows:
\begin{equation}\label{ddot_V1}
\setlength{\arraycolsep}{1pt}
\begin{array}{rcl}
\ddot{V}_1&=&2(k_{\bar{\omega}}k_{\bar{R}}\bar{\omega}^{\top}J^{-1}\bar{\Psi}+k_{\bar{\omega}}^2\bar{\omega}^{\top}J^{-1}\bar{\omega}-k_{\omega^l}(\tilde{\omega}_0^l)^{\top}\mathbf{F}_1\tilde{\sigma}_0\\
\ &\ &
+k_{\omega^l}^2\|\mathbf{F}_1\tilde{\omega}^l\|^2
+k_{\omega^l}k_{R^l}(\tilde{\omega}_0^l)^{\top}\hat{R}^{\top}\tilde{\mathcal{L}}_2^l\Psi^l\\
\ &\ &-(\tilde{\omega}_0^l)^{\top}(I_N\otimes P)\tilde{\sigma}_0+\tilde{\sigma}_0^{\top}\mathbf{F}_2((I_N\otimes P)-k_{\tilde{\sigma}}\mathbf{F}_1)\tilde{\sigma}_0).
\end{array}
\end{equation}
From the fact that $\dot{V}_1$ is negative semi-definite, it is clear that $\tilde{\omega}_0^l,\tilde{\sigma}_0,\bar{\omega}$ are bounded.
Using \eqref{e2.3}-\eqref{e2.4} and the fact that $\mathbf{E}(A_i^l)$ and $\mathbf{E}(\bar{A}_i)$ are positive definite for $i\in\mathcal{V}$, one deduces that $\psi(A_i^l\tilde{R}_i^l)$ and $\psi(\bar{A}_i\bar{R}_i^r)$ are  bounded for $i\in\mathcal{V}$, implying that $\Psi^l$ and $\bar{\Psi}$ are bounded.
It follows from $\|\tilde{\mathcal{L}}_2^l\|_F=\sqrt{\mathbf{tr}((\tilde{\mathcal{L}}_2^l)^{\top}\tilde{\mathcal{L}}_2^l)}=\sqrt{6N-3n}$ that $\tilde{\mathcal{L}}_2^l$ is bounded.
Therefore, one concludes that $\ddot{V}_1$ in \eqref{ddot_V1} is bounded, implying that $\dot{V}_1$ is uniformly continuous.
This, with the fact that $\lim_{t\rightarrow\infty}V_1$ exists and is finite, results in $\lim_{t\rightarrow\infty}\dot{V}_1=0$ as per Barbalat lemma.
As a result, one has $\lim_{t\rightarrow\infty}\bar{\omega}=\mathbf{0}_{3N},\lim_{t\rightarrow\infty}\tilde{\omega}_0^l=\mathbf{0}_{3N},\lim_{t\rightarrow\infty}\tilde{\sigma}_0=\mathbf{0}_{3N}$.

Now, we show that $\dot{\tilde{\omega}}_0^l$ and $\dot{\bar{\omega}}$ converge to $\mathbf{0}_{3N}$ using Barbalat lemma.
The second time-derivatives of $\tilde{\omega}_0^l$ and $\bar{\omega}$ are as follows:
\[
\setlength{\arraycolsep}{1pt}
\begin{array}{rcl}
\ddot{\tilde{\omega}}_0^l
&=&(I_N\otimes P-k_{\tilde{\sigma}}\mathbf{F}_1)\tilde{\sigma}_0+k_{R^l}\hat{W}\hat{R}^{\top}\tilde{\mathcal{L}}_2^l\Psi^l-k_{R^l}\hat{R}^{\top}\dot{\tilde{\mathcal{L}}}_2^l\Psi^l\\
\ &\ &-k_{R^l}\hat{R}^{\top}\tilde{\mathcal{L}}_2^l\dot{\Psi}^l-k_{\omega^l}\mathbf{F}_1(\tilde{\sigma}_0-k_{R^l}\hat{R}^{\top}\tilde{\mathcal{L}}_2^l\Psi^l-k_{\omega^l}\mathbf{F}_1\tilde{\omega}_0^l),\\
\ddot{\bar{\omega}}&=&J^{-1}(k_{\bar{\omega}}(k_{\bar{R}}\bar{\Psi}+k_{\bar{\omega}}\bar{\omega})-k_{\bar{R}}\dot{\bar{\Psi}}),
\end{array}
\]
where $\hat{W}=\mathrm{diag}(\hat{\omega}_1^{\times},\ldots,\hat{\omega}_N^{\times})$, $\dot{\Psi}^l=\mathrm{diag}(\mathbf{E}(A_1^l\tilde{R}_1^l),\ldots,$
$\mathbf{E}(A_N^l\tilde{R}_N^l))[\mathbf{vex}(\tilde{W}_{11}^l)^{\top}\ \cdots\ \mathbf{vex}(\tilde{W}_{NN}^l)^{\top}]^{\top}$,
$\dot{\bar{\Psi}}=\mathrm{diag}(\mathbf{E}(\bar{A}_1\bar{R}_1^r),\ldots,\mathbf{E}(\bar{A}_N\bar{R}_N^r))\bar{\omega}$, $\dot{\tilde{\mathcal{L}}}_2^l$ is obtained by removing the first three rows of $\dot{\tilde{\mathcal{L}}}^l$,
and $\dot{\tilde{\mathcal{L}}}^l\in\mathbb{R}^{3(N+1)\times 3N}$ is obtained by defining each $3$-by-$3$ block $(\dot{\tilde{\mathcal{L}}}^l)_{ik}$ as
\[
(\dot{\tilde{\mathcal{L}}}^l)_{ik}=\left\{
\begin{array}{ll}
\dot{\tilde{R}}_k^l,&k\in(\mathcal{M}_{\overline{\mathcal{G}}_0})_{i-1}^{-},\\
\mathbf{0}_{3\times 3},&\mathrm{otherwise}.
\end{array}
\right.
\]
Since $\tilde{\omega}_{0}^l$ and $\omega_0$ are bounded, one infers that $\hat{\omega}_{i}$ is bounded for $i\in\mathcal{V}$, implying that $\hat{W}$ is bounded.
Using the fact that $\|\dot{\tilde{R}}_{i,j}^l\|_F^2=-\mathbf{tr}((\hat{R}_j\tilde{\omega}_{i,j}^l)^{\times})^2 =2||\tilde{\omega}_{i,j}^l||^2$ and $\|\dot{\bar{R}}_i^r\|_F^2=-\mathbf{tr}(\bar{\omega}_i^{\times})^2= 2||\bar{\omega}_i||^2$, one concludes that $\dot{\tilde{R}}_{i,0}^l$ and $\dot{\bar{R}}_i^r$ are bounded since $\tilde{\omega}_{i,j}^l$ and $\bar{\omega}_i$ are bounded, which implies that $\dot{\tilde{\mathcal{L}}}_2^l$ is bounded.
Using \eqref{e2.5}, it is clear that $\mathbf{E}(A_k^l\tilde{R}_k^l)$ and $\mathbf{E}(\bar{A}_i\bar{R}_i^r)$ are bounded.
This, together with the boundedness of $\tilde{W}^l$ and $\bar{\omega}$, means that $\dot{\Psi}^l$ and $\dot{\bar{\Psi}}$ are bounded.
Therefore, $\ddot{\tilde{\omega}}_0^l$ and $\ddot{\bar{\omega}}$ are bounded, implying that $\dot{\tilde{\omega}}_0^l$ and $\dot{\bar{\omega}}$ are uniformly continuous.
Combining with $\lim_{t\rightarrow\infty}\tilde{\omega}_0^l=\mathbf{0}_{3N},\lim_{t\rightarrow\infty}\bar{\omega}=\mathbf{0}_{3N}$, one can infer that $\lim_{t\rightarrow\infty}\dot{\tilde{\omega}}_0^l=\mathbf{0}_{3N},\lim_{t\rightarrow\infty}\dot{\bar{\omega}}=\mathbf{0}_{3N}$ as per Barbalat lemma.
Consequently, one has $\lim_{t\rightarrow\infty}\hat{R}^{\top}\tilde{\mathcal{L}}_2^l\Psi^l=\mathbf{0}_{3N}$ and $\lim_{t\rightarrow\infty}\bar{\Psi}=\mathbf{0}_{3N}$.
Lemma \ref{le2} shows that
$\tilde{\mathcal{L}}_2^l$ is non-singular, yielding a non-singular matrix $\hat{R}^{\top}\tilde{\mathcal{L}}_2^l$.
This leads to $\lim_{t\rightarrow\infty}\Psi^l=\mathbf{0}_{3N}$, \textit{i.e.}, $\lim_{t\rightarrow\infty}\psi(A_k^l\tilde{R}_k^l)=\mathbf{0}_3$.
From $\lim_{t\rightarrow\infty}\bar{\Psi}=\mathbf{0}_{3N}$, one derives $\lim_{t\rightarrow\infty}\psi(\bar{A}_i\bar{R}_i^r)=\mathbf{0}_3$.
As per Lemma 2 in \cite{Mayhew2011}, one concludes that $\bar{R}_i^r$ and $\tilde{R}_k^l$ converge to $\{I_3\}\cup\{\mathcal{R}(\pi,\bar{u})|\bar{u}\in\mathcal{E}_{v}(\bar{A}_i)\}$ and $\{I_3\}\cup\{\mathcal{R}(\pi,\tilde{u}^l)|\tilde{u}^l\in\mathcal{E}_{v}(A_i^r)\}$ respectively, where $i\in\mathcal{V},k\in\mathcal{M}_{\overline{\mathcal{G}}_0}$.

To sum up, the trajectories of the error systems \eqref{errorcompact} and $\eqref{controlerrorcompact}$ converge to the set of equilibria $\mathcal{C}^l$.

Now we prove the second item by showing that the undesired equilibria in $\mathcal{C}^{l,u}$ are unstable and the stable manifold associated to these undesired equilibria has zero Lebesgue measure.
Since the tracking error subsystem \eqref{controlerrorcompact} is autonomous, we proceed by showing that the equilibria in $\mathcal{C}^{l,u}$
associated with the subsystem \eqref{controlerrorcompact} are unstable and their associated stable manifold has zero Lebesgue measure with respect to the subspace of the tracking error subsystem \eqref{controlerrorcompact}.
To this end, we linearize the tracking error subsystem \eqref{controlerrorcompact}, and show that its
Jacobian matrix has at least one positive eigenvalue.

Consider the undesired equilibrium $(\tilde{R}^{l,\ast},\mathbf{0}_{3N},\mathbf{0}_{3N},\bar{R}^{r,\ast},$
$\mathbf{0}_{3N})$. 
Let $(\bar{R}_i^{r,\ast})^{\top}\bar{R}_i^r=\exp(\bar{\mathbf{r}}_i^{\times})$ for $i\in\mathcal{V}$.
Using the fact that $\exp(\bar{\mathbf{r}}_i^{\times})\approx I_3+\bar{\mathbf{r}}_i^{\times}$ around $\bar{\mathbf{r}}_i\approx 0$, one has  $\bar{R}_i^r\approx\bar{R}_i^{r,\ast}(I_3+\bar{\mathbf{r}}_i^{\times})$ for $i\in\mathcal{V}$.
Thus, the linearization of $\dot{\bar{R}}_i^r=\bar{R}_i^r\bar{\omega}_i^{\times}$ around $\bar{R}_i^{r,\ast}$ is given by $\dot{\bar{\mathbf{r}}}_i=\bar{\omega}_i$, where
$i\in\mathcal{V}$.
It follows that
$
\psi(\bar{A}_i\bar{R}_i^r)\approx\mathbf{E}(\bar{A}_i\bar{R}_i^{r,\ast})\bar{\mathbf{r}}_i$, $i\in\mathcal{V}$.
The linearized angular velocity tracking error dynamics around the undesired equilibrium are given by
$J_i\dot{\bar{\omega}}_i=-k_{\bar{R}}\mathbf{E}(\bar{A}_i\bar{R}_i^{r,\ast})\bar{\mathbf{r}}_i-k_{\bar{\omega}}\bar{\omega}_i$.
Therefore, the linearization of subsystem \eqref{controlerrorcompact} around the undesired equilibrium is given by
\begin{equation}\label{linearizedobserver-1}
\begin{array}{rcl}
\dot{\bar{\mathbf{r}}}&=&\bar{\omega},\\
\dot{\bar{\omega}}&=&-k_{\bar{R}}J^{-1}\mathbf{E}_{\bar{A}\bar{R}^{r,\ast}}\bar{\mathbf{r}}-k_{\bar{\omega}}J^{-1}\bar{\omega},\\
\end{array}
\end{equation}
where
$\mathbf{E}_{\bar{A}\bar{R}^{r,\ast}}=\mathrm{diag}(\mathbf{E}(\bar{A}_1\bar{R}_1^{r,\ast}),\ldots,\mathbf{E}(\bar{A}_N\bar{R}_N^{r,\ast}))$, $\bar{\mathbf{r}}=[\bar{\mathbf{r}}_1^{\top}\ \cdots\ \bar{\mathbf{r}}_N^{\top}]^{\top}$. The Jacobian matrix is then obtained as follows:
\begin{equation}\label{bfM2}
\mathbf{M}_2=\left[
\begin{array}{cc}
\mathbf{0}_{3N\times 3N}&I_{3N}\\
-k_{\bar{R}}J^{-1}\mathbf{E}_{\bar{A}\bar{R}^{r,\ast}}& -k_{\bar{\omega}}J^{-1}\\
\end{array}
\right]_.
\end{equation}

Now, we show that $\mathbf{M}_2$ is non-singular by proving $\mathrm{det}(\mathbf{E}_{\bar{A}\bar{R}^{r,\ast}})\neq 0$, due to the fact that $\mathrm{det}(\mathbf{M}_2)=\mathrm{det}(J^{-1})\mathrm{det}(k_{\bar{R}}\mathbf{E}_{\bar{A}\bar{R}^{r,\ast}})$. Straightforward calculations show that the eigenvalues of $\mathbf{E}(\bar{A}_i\bar{R}_i^{r,\ast})$ are given by $\frac{1}{2}(\mathbf{tr}(\bar{A}_i\bar{R}_i^{r,\ast})-\lambda^{\bar{A}_i\bar{R}_i^{r,\ast}})$. For $i\in\mathcal{M}_{\overline{\mathcal{G}}_0}^{I}$, one has $\bar{R}_i^{r,\ast}=I_3$, which leads us to conclude that the eigenvalues of $\mathbf{E}(\bar{A}_{i}\bar{R}^{r,\ast}_{i})$ are given by $\frac{1}{2}(\mathbf{tr}(\bar{A}_{i})-\lambda^{\bar{A}_{i}})\not=0$.
For $i\in\mathcal{M}_{\overline{\mathcal{G}}_0}^{\pi}$, one has $\bar{A}_{i}\bar{R}_i^{r,\ast}=2\lambda_{\bar{u}_i}^{\bar{A}_i}\bar{u}_i\bar{u}_i^{\top}-\bar{A}_i$, leading to $\bar{A}_{i}\bar{R}_i^{r,\ast}\bar{u}_i^r=2\lambda_{\bar{u}_i}^{\bar{A}_i}\bar{u}_i-\bar{A}_i\bar{u}_i=\lambda_{\bar{u}_i}^{\bar{A}_i}\bar{u}_i$, implying that $\lambda_{\bar{u}_i}^{\bar{A}_i}$ is one of the eigenvalues of $\bar{A}_{i}\bar{R}_i^{r,\ast}$.
In this case, one of the eigenvalues of $\mathbf{E}(\bar{A}_i\bar{R}_i^{r,\ast})$ is given by $\frac{1}{2}(\lambda_{\bar{u}_i}^{\bar{A}_i}-\mathbf{tr}(\bar{A}_i))$.
Let $\lambda_{u_\ell}^{\bar{A}_i}\not=\lambda_{\bar{u}_i}^{\bar{A}_i}$, $\ell \in\{1,2\}$, be the two other distinct eigenvalues of $\bar{A}_i$ with $u_\ell$ being their associated unit eigenvectors. Since $\bar{A}_i$ has distinct eigenvalues,
one has $u_\ell^{\top}\bar{u}_i=0$. Then, one has $\bar{A}_i\bar{R}_i^{r,\ast}u_\ell=(2\lambda_{\bar{u}_i}^{\bar{A}_i}\bar{u}_i\bar{u}_i^{\top}-\bar{A}_i)u_\ell=-\lambda_{u_\ell}^{\bar{A}_i}u_\ell$,
implying that $-\lambda_{u_\ell}^{\bar{A}_i}$, $\ell \in\{1,2\}$, are two eigenvalues of $\bar{A}_i\bar{R}_i^{r,\ast}$.
As a result, one can conclude that $\lambda_{\bar{u}_i}^{\bar{A}_i}+\frac{1}{2}(\lambda_{u_\ell}^{\bar{A}_i}-\mathbf{tr}(\bar{A}_i))$,  $\ell \in\{1,2\}$, are two other eigenvalues of $\mathbf{E}(\bar{A}_i\bar{R}_i^{r,\ast})$. Finally, for $i\in\mathcal{M}_{\overline{\mathcal{G}}_0}^{\pi}$, the eigenvalues of $\mathbf{E}(\bar{A}_i\bar{R}_i^{r,\ast})$ are given by $-\frac{1}{2}(\lambda_{u_1}^{\bar{A}_i}+\lambda_{u_2}^{\bar{A}_i})$ and $\frac{1}{2}(\lambda_{\bar{u}_i}^{\bar{A}_i}-\lambda_{u_\ell}^{\bar{A}_i})$, $\ell \in\{1,2\}$. Consequently, one can conclude that  $\mathrm{det}(\mathbf{E}_{\bar{A}\bar{R}^{r,\ast}})\neq 0$.

Since $\mathbf{M}_2$ is non-singular, one has
$\mathrm{det}(\mathbf{M}_2-\lambda I_{6N})=\mathrm{det}(\lambda^2 I_{3N}+\lambda k_{\bar{\omega}}J^{-1}+k_{\bar{R}}J^{-1}\mathbf{E}_{\bar{A}\bar{R}^{r,\ast}})$.
Let $\tilde{\mathbf{M}}=J(\mathbf{M}_2-\lambda I_{6N})=\lambda^2 J+\lambda k_{\bar{\omega}}I_{3N}+k_{\bar{R}}\mathbf{E}_{\bar{A}\bar{R}^{r,\ast}}$.
Then one can show that there exists $\lambda>0$ such that $\mathrm{det}(\mathbf{M}_2-\lambda I_{6N})=0$ by proving that there exists $\lambda>0$ such that $\mathrm{det}(\tilde{\mathbf{M}})=0$.

Let us consider the following continuous function:
\[f(\lambda)=\min_{\|z\|=1}  z^{\top}\tilde{\mathbf{M}}z.\]
Taking $z=[z_1^{\top}\ \cdots\ z_N^{\top}]^{\top}$ with $z_{p}=\mathbf{0}_3,\forall p\in\mathcal{M}_{\overline{\mathcal{G}}_0}^{I}$ and $z_{q}=\bar{u}_{q},\forall q\in\mathcal{M}_{\overline{\mathcal{G}}_0}^{\pi}$, one has
\[
\begin{array}{l} z^{\top}\tilde{\mathbf{M}} z=k_{\bar{R}}\sum\limits_{q\in\mathcal{M}_{\overline{\mathcal{G}}_0}^{\pi}}\bar{u}_{q}^{\top}\mathbf{E}(\bar{A}_{q}\bar{R}^{r,\ast}_{q})\bar{u}_{q}+\lambda k_{\bar{\omega}}+\lambda^2z^{\top}Jz.
\end{array}
\]
Using the fact that $\bar{A}_{q}\bar{R}^{r,\ast}_{q}=2\lambda^{\bar{A}_{q}}_{\bar{u}_{q}}\bar{u}_{q}\bar{u}_{q}^{\top}-\bar{A}_{q}$, we derive
$\mathbf{E}(\bar{A}_{q}\bar{R}^{r,\ast}_{q})=\lambda^{\bar{A}_{q}}_{\bar{u}_{q}}(I_3-\bar{u}_{q}\bar{u}_{q}^{\top})-\mathbf{E}(\bar{A}_{q})$, yielding $\bar{u}_{q}^{\top}\mathbf{E}(\bar{A}_{q}\bar{R}^{r,\ast}_{q})\bar{u}_{q}=-\frac{1}{2}(\mathbf{tr}(\bar{A}_{q})-\lambda^{\bar{A}_{q}}_{\bar{u}_{q}})<0$.
Consequently, $z^{\top}\tilde{\mathbf{M}}z<0$ for $\lambda=0$, allowing to conclude that $\tilde{\mathbf{M}}$ has at least one negative eigenvalue, and as such, $f(0)=\lambda_{\min}^{\mathbf{E}_{\bar{A}\bar{R}^{r,\ast}}}<0$.
Due to the fact that  $\lambda_{\min}^{\mathbf{E}_{\bar{A}\bar{R}^{r,\ast}}}\leq z^{\top}\mathbf{E}_{\bar{A}\bar{R}^{r,\ast}}z\leq \lambda^{\mathbf{E}_{\bar{A}\bar{R}^{r,\ast}}}_{\max},~\forall z\in\mathbb{S}^{3N-1}$, there must exist a real positive $\lambda_1$ such that $\lambda_1^2z^{\top}Jz+\lambda_1 k_{\bar{\omega}}>-k_{R^r}\lambda_{\min}^{\mathbf{E}_{\bar{A}\bar{R}^{r,\ast}}},~\forall z\in\mathbb{S}^{3N-1}$, guaranteeing that $f(\lambda_1)>0$.
Therefore, there exists $0<\lambda^{\ast}<\lambda_1$ such that $f(\lambda^{\ast})=0$.
Consequently, there exists a real positive eigenvalue $\lambda^{\ast}$ such that $0$ is the minimum eigenvalue of $\tilde{\mathbf{M}}$, implying that $\mathrm{det}(\mathbf{M}_{2}-\lambda^{\ast} I_{6N})=\mathrm{det}(J^{-1})\mathrm{det}(\tilde{\mathbf{M}})=0$.
Therefore, the diagonal block matrix of
the Jacobian matrix $\mathbf{M}_2$ has a real positive eigenvalue.
As a result, the undesired equilibria in $\mathcal{C}^{l,u}$ are unstable. According to the stable manifold theorem \cite{Perko_book}, the stable manifold of the tracking error subsystem \eqref{controlerrorcompact} associated to the undesired equilibria in $\mathcal{C}^{l,u}$ has a zero Lebesgue measure with respect to the subspace of the tracking error system.
Therefore, the stable manifold of the overall error system associated to the undesired equilibria in $\mathcal{C}^{l,u}$ has a zero Lebesgue measure with respect to the entire state space of the overall error system.
Consequently, the desired equilibrium $\mathcal{C}^{l,d}$ is almost globally attractive.
Together with the fact that $\mathcal{C}^{l,d}$ is stable, one concludes that $\mathcal{C}^{l,d}$ is almost globally asymptotically stable.

Finally, we prove the third item.
Due to the fact that the desired equilibrium $\mathcal{C}^{l,d}$ is almost globally asymptotically stable, it is clear that $\lim_{t\rightarrow\infty}\hat{R}_i\hat{R}_j^{\top}=I_3$, $(j,i)\in\overline{\mathcal{E}}_0$,  $\lim_{t\rightarrow\infty}(\hat{\omega}_i-\omega_0)=\mathbf{0}_3$, $\lim_{t\rightarrow\infty}\hat{R}_i^{
\top}R_i=I_3$ and $\lim_{t\rightarrow\infty}(\omega_i-R_i^{\top}\hat{R}_i\hat{\omega}_i)=\mathbf{0}_3$ for $i=1,\ldots,N$.
This yields
\begin{align*}
&\lim_{t\rightarrow\infty}(\omega_i-\omega_0)=
\lim_{t\rightarrow\infty}(\omega_i-\hat{\omega}_i)+\lim_{t\rightarrow\infty}(\hat{\omega}_i-\omega_0)\\
&=\lim_{t\rightarrow\infty}(\omega_i-R_i^{\top}\hat{R}_i\hat{\omega}_i+R_i^{\top}\hat{R}_i\hat{\omega}_i-\hat{\omega}_i)\\
&=\lim_{t\rightarrow\infty}(R_i^{\top}\hat{R}_i-I_3)\hat{\omega}_i=\mathbf{0}_3
\end{align*}
for $i=1,\ldots,N$.
Based on Assumption \ref{assumpgraph_original}, for the observer of the $i$-th rigid body ($i\in\mathcal{V}$), there exists a path, denoted as $(0,j_1,\ldots,j_k,i)$, from the leader $0$ to the $i$-th rigid body,
thereby one has
\begin{align*}
&\lim_{t\rightarrow\infty}\hat{R}_iR_0^{\top}=\lim_{t\rightarrow\infty}\hat{R}_i\hat{R}_{j_k}^{\top}\hat{R}_{j_k}\hat{R}_{j_{k-1}}^{\top}\cdots\hat{R}_{j_1}R_0^{\top}=I_3,
\end{align*}
leading to \[\lim_{t\rightarrow\infty}R_iR_0^{\top}=\lim_{t\rightarrow\infty}R_iR_i^{\top}\hat{R}_iR_0^{\top}=\lim_{t\rightarrow\infty}\hat{R}_iR_0^{\top}=I_3.\]
Consequently, one concludes that all the rigid body attitudes synchronize (almost globally) to the desired attitude $R_0$.

\section{Proof of Theorem \ref{stabilityforestimateconstant}} \label{appenE}

Let us prove the first item.
Consider the following Lyapunov function candidate:
\[
\setlength{\arraycolsep}{1pt}
\begin{array}{rcl}
V_2&=&\frac{1}{2}(k_{R^r}\mathbf{tr}(\mathbf{A}^r(I_{3N}-\tilde{R}^r))+(\tilde{\omega}^r_0)^{\top}\tilde{\omega}^r_0\\
\ &\ &+k_{\bar{R}}\mathbf{tr}(\bar{\mathbf{A}}(I_{3N}-\bar{R}^r))+\bar{\omega}^{\top}J\bar{\omega}),
\end{array}
\]
where $\mathbf{A}^r=\mathrm{diag}(A^r_1,\ldots,A^r_N)$.
Using the fact that $\lambda_{\min}^{\mathbf{E}(A)}=\frac{1}{2}(\mathbf{tr}(A)-\lambda_{\max}^{A^{\top}}),\forall A=A^{\top}>0$, one has $\lambda_{\min}^{\mathbf{E}(A^r_i)}>0$ and $\lambda_{\min}^{\mathbf{E}(\bar{A}_i)}>0$ since $A^r_k$ and $\bar{A}_i$ are positive definite.
Together with \eqref{e2.1}, one has $\mathbf{tr}(A^r_k(I_3-\tilde{R}^r_k))\geq 4\lambda_{\min}^{\mathbf{E}(A^r_k)}|\tilde{R}^r_k|_I^2$ for $k\in\mathcal{M}_{\overline{\mathcal{G}}_0}$ and $\mathbf{tr}(\bar{A}_i(I_3-\bar{R}_i^r))\geq  4\lambda_{\min}^{\mathbf{E}(\bar{A}_i)}|\bar{R}_i^r|_I^2$ for $i\in\mathcal{V}$, leading to $V_2$ positive definite.
Using \eqref{e2.10}, the time-derivative of $V_2$, along the trajectories of systems \eqref{estimatedcompactconstant} and  \eqref{controlerrorcompact1}, is given by
\[\dot{V}_2=-k_{\tilde{\omega}}(\tilde{\omega}^r_0)^{\top}((\Xi+LL^{\top})\otimes I_3)\tilde{\omega}^r_0-k_{\bar{\omega}}\|\bar{\omega}\|^2,\]
As per Lemma \ref{le1}, $\Xi+LL^{\top}$ is a positive definite matrix, resulting in $\dot{V}_2\leq 0$.
Consequently, the desired equilibrium $\tilde{
\mathcal{C}}^{r,d}$ is stable.

Barbalat Lemma is used to show that the trajectories of the error systems \eqref{estimatedcompactconstant} and  \eqref{controlerrorcompact1} converge to $\mathcal{C}^r$.
The proof proceeds by first establishing $\lim_{t\rightarrow\infty}\dot{V}_2=0$.
The function $V_2$ is non-increasing due to the fact that $\dot{V}_2\leq 0$.
Together with the fact that $V_2$ is positive definite, one concludes that $\lim_{t\rightarrow\infty}V_2$ exists and is finite.
The second time-derivative of $V_2$ is given by
\begin{equation}\label{ddot_V2}
\setlength{\arraycolsep}{1pt}
\begin{array}{rcl}
\ddot{V}_2&=&2(k_{\omega^r}k_{R^r}(\tilde{\omega}^r_0)^{\top}\mathbf{F}_1\tilde{\mathcal{L}}_2^r\Psi^r+k_{\omega^r}^2\|\mathbf{F}_1\tilde{\omega}^r_0\|^2\\
\ &\ &-k_{\bar{\omega}}\bar{\omega}^{\top}J^{-1}\Sigma\bar{\omega}+k_{\bar{\omega}}k_{\bar{R}}\bar{\omega}^{\top}J^{-1}\bar{\Psi}+k_{\bar{\omega}}^2\bar{\omega}^{\top}J^{-1}\bar{\omega}).
\end{array}
\end{equation}
Since $\dot{V}_2$ is negative semi-definite, one can conclude that $\tilde{\omega}^r_0$ and $\bar{\omega}$ are bounded. Consequently, through a simple signal chasing, one can show that all the terms in the right-hand side of \eqref {ddot_V2} are bounded. Hence, $\ddot{V}_2$ is bounded and consequently $\dot{V}_2$ is uniformly continuous. Combining this with the fact that $\lim_{t\rightarrow\infty}V_2$ exists and is finite, it follows that $\lim_{t\rightarrow\infty}\dot{V}_2=0$ as per Barbalat lemma.
Therefore, $\lim_{t\rightarrow\infty}\tilde{\omega}^r_0=\mathbf{0}_{3N}$ and $\lim_{t\rightarrow\infty}\bar{\omega}=\mathbf{0}_{3N}$.

Now, we show that $\lim_{t\rightarrow\infty}\dot{\tilde{\omega}}^r_0=\mathbf{0}_{3N}$ and $\lim_{t\rightarrow\infty}\dot{\bar{\omega}}=\mathbf{0}_{3N}$.
The second time-derivatives of $\tilde{\omega}^r_0$ and $\bar{\omega}$ are given by
\[
\begin{array}{l}
\ddot{\tilde{\omega}}^r_0
=-k_{R^r}\dot{\tilde{\mathcal{L}}}^r_2\Psi^r-k_{R^r}\tilde{\mathcal{L}}^r_2\dot{\Psi}^r+k_{\omega^r}k_{R^r}\mathbf{F}_1\tilde{\mathcal{L}}^r_2\Psi^r+k_{\omega^r}^2\mathbf{F}_1^2\tilde{\omega}^r_0,\\
\ddot{\bar{\omega}}=J^{-1}(\dot{\Sigma}\bar{\omega}-k_{\bar{R}}\dot{\bar{\Psi}}+(\Sigma-k_{\bar{\omega}}I_{3N})(\Sigma\bar{\omega}-k_{\bar{R}}\bar{\Psi}-k_{\bar{\omega}}\bar{\omega})),
\end{array}
\]
where $\dot{\Psi}^r=\mathrm{diag}(\mathbf{E}(A^r_1\tilde{R}^r_1),\ldots,\mathbf{E}(A^r_N\tilde{R}^r_N))\tilde{\omega}^r_0$, $\dot{\tilde{\mathcal{L}}}^r_2$ is obtained by removing the first three rows of $\dot{\tilde{\mathcal{L}}}^r$,
and $\dot{\tilde{\mathcal{L}}}^r\in\mathbb{R}^{3(N+1)\times 3N}$ is constructed by defining each $3$ by $3$ block $(\dot{\tilde{\mathcal{L}}}^r)_{ik}$ as
\[
(\dot{\tilde{\mathcal{L}}}^r)_{ik}=\left\{
\begin{array}{ll}
-\dot{\tilde{R}}_k^r,&k\in(\mathcal{M}_{\overline{\mathcal{G}}_0})_{i-1}^{+},\\
\mathbf{0}_{3\times 3},&\mathrm{otherwise},
\end{array}
\right.
\]
$\dot{\Sigma}=\mathrm{diag}(\dot{\Sigma}_1,\ldots,\dot{\Sigma}_N)$ with $\dot{\Sigma}_i$ being definied as
\[\begin{array}{rcl}
\dot{\Sigma}_i&=&(J_i\dot{\bar{\omega}}_i)^{\times}+(J_i\bar{\omega}_i^{\times}(\bar{R}_i^r)^{\top}\hat{\omega}_i)^{\times}-(J_i(\bar{R}_i^r)^{\top}\dot{\hat{\omega}}_i)^{\times}\\
\ &\ &+(\bar{\omega}_i^{\times}(\bar{R}_i^r)^{\top}\hat{\omega}_i)^{\times}J_i-((\bar{R}_i^r)^{\top}\dot{\hat{\omega}}_i)^{\times}J_i\\
\ &\ &+J_i(\bar{\omega}_i^{\times}(\bar{R}_i^r)^{\top}\hat{\omega}_i)^{\times}-J_i((\bar{R}_i^r)^{\top}\dot{\hat{\omega}}_i)^{\times}.
\end{array}\]
Since $\hat{\omega}_i,\bar{\omega}_i,\bar{R}_i^r$ are bounded, $\Sigma_i,\dot{\Sigma}_i$ are bounded, \textit{i.e.}, $\Sigma, \dot{\Sigma}$ are bounded.
As proven in the proof of Theorem \ref{estimationstability}, $\bar{\Psi},\dot{\bar{\Psi}}$ are bounded.
Hence, one can infer that $\ddot{\bar{\omega}}$ is bounded, implying that $\dot{\bar{\omega}}$ is uniformly continuous.
Following the same argument as in the proof of Theorem \ref{estimationstability}, one can conclude that $\dot{\tilde{\mathcal{L}}}_2^r$ and $\dot{\Psi}^r$ are bounded, which in turn implies that $\ddot{\tilde{\omega}}^r_0$ is bounded.
Thus, $\dot{\tilde{\omega}}^r_0$ is uniformly continuous.
This, with $\lim_{t\rightarrow\infty}\tilde{\omega}_0^r=\mathbf{0}_{3N}$ and $\lim_{t\rightarrow\infty}\bar{\omega}=\mathbf{0}_{3N}$, leads to $\lim_{t\rightarrow\infty}\dot{\tilde{\omega}}_0^r=\mathbf{0}_{3N}$ and $\lim_{t\rightarrow\infty}\dot{\bar{\omega}}=\mathbf{0}_{3N}$ according to Barbalat lemma.
It follows that $\lim_{t\rightarrow\infty}\tilde{\mathcal{L}}^r_2\Psi^r=\mathbf{0}_{3N}$ and $\lim_{t\rightarrow\infty}\bar{\Psi}=\mathbf{0}_{3N}$.
As per Lemma \ref{le2}, 
$\tilde{\mathcal{L}}^r_2$ is non-singular, leading to $\lim_{t\rightarrow\infty}\Psi^r=\mathbf{0}_{3N}$, \textit{i.e.}, $\lim_{t\rightarrow\infty}\psi(A^r_k\tilde{R}^r_k)=\mathbf{0}_3$. 
By virtue of \cite[Lemma 2]{Mayhew2011}, one concludes that $\tilde{R}^r_k$ converges to $\{I_3\}\cup\{\mathcal{R}(\pi,\tilde{u}^r)|\tilde{u}^r\in\mathcal{E}_{v}(A_i^r)\}$ for $k\in\mathcal{M}_{\overline{\mathcal{G}}_0}$.
As shown in the proof of Theorem \ref{estimationstability},  $\lim_{t\rightarrow\infty}\bar{\Psi}=\mathbf{0}_{3N}$ implies that $\bar{R}_i$ converges to $\{I_3\}\cup\{\mathcal{R}(\pi,\bar{u})|\bar{u}\in\mathcal{E}_{v}(\bar{A}_i)\}$ for $i\in\mathcal{V}$.
To sum up, the trajectories of the error systems \eqref{estimatedcompactconstant} and \eqref{controlerrorcompact1} converge to the set of equilibria $\mathcal{C}^r$.

Next, we prove the second item by showing that the undesired equilibria in $\mathcal{C}^{r,u}$ are unstable and the stable manifold associated to the undesired equilibria in $\mathcal{C}^{r,u}$ has zero Lebesgue measure.
To this end, we linearize the error systems \eqref{estimatedcompactconstant} and  \eqref{controlerrorcompact1}, and prove that the Jacobian matrix has at least one positive eigenvalue.

Consider the undesired equilibrium $(\tilde{R}^{r,\ast},\mathbf{0}_{3N},\bar{R}^{r,\ast},\mathbf{0}_{3N})$.
The linearization of the tracking error subsystem \eqref{controlerrorcompact1} around the undesired equilibria in $\mathcal{C}^{l,u}$ is given  by \eqref{linearizedobserver-1}, and its Jacobian matrix is $\mathbf{M}_2$ defined in \eqref{bfM2}.
Now, we linearize the estimation error subsystem \eqref{estimatedcompactconstant}.
Let $(\tilde{R}_k^{r,\ast})^{\top}\tilde{R}^r_k=\exp(\tilde{\mathbf{r}}_k^{\times})$, where $k\in\mathcal{M}_{\overline{\mathcal{G}}_0}$.
Using the fact that $\exp(\tilde{\mathbf{r}}_k^{\times})\approx I_3+\tilde{\mathbf{r}}_k^{\times}$
around $\tilde{\mathbf{r}}_k \approx 0$, one has $\tilde{R}^r_k\approx \tilde{R}_k^{r,\ast}(I_3+\tilde{\mathbf{r}}_k^{\times})$, where $k\in\mathcal{M}_{\overline{\mathcal{G}}_0}$.
Thus, the linearization of $\dot{\tilde{R}}^r_k=\tilde{R}^r_k(\tilde{\omega}^r_k)^{\times}$ is given by $\dot{\tilde{\mathbf{r}}}_k=\tilde{\omega}^r_k=\mathrm{Row}_k((\tilde{\mathcal{L}}^{r,\ast}_2)^{\top})\tilde{\omega}_0^r$, where
$\tilde{\mathcal{L}}^{r,\ast}_2$ is obtained by removing the first three rows of $\tilde{\mathcal{L}}^{r,\ast}$ with each $3$ by $3$ block $(\tilde{\mathcal{L}}^{r,\ast})_{ik}$ as
\begin{equation}\label{tildetildeL}
(\tilde{\mathcal{L}}^{r,\ast})_{ik}=\left\{
\begin{array}{ll}
-\tilde{R}_k^{r,\ast}, & k\in(\mathcal{M}_{\overline{\mathcal{G}}_0})^{+}_{i-1},  \\
I_3, & k\in(\mathcal{M}_{\overline{G}_0})^{-}_{i-1},\\
\mathbf{0}_{3\times 3}, &\mathrm{otherwise}.\\
\end{array}
\right.
\end{equation}
It follows from $\tilde{R}^r_k\approx\tilde{R}_k^{r,\ast}(I_3+\tilde{\mathbf{r}}_k^{\times})$
that $\psi(A^r_k\tilde{R}^r_k)\approx\mathbf{E}(A^r_k\tilde{R}_k^{r,\ast})\tilde{\mathbf{r}}_k$,
where $k\in\mathcal{M}_{\overline{\mathcal{G}}_0}$.
This leads to the following linearization of the angular velocity error dynamics around the undesired equilibrium in $\mathcal{C}^{r,u}$
\[
\setlength{\arraycolsep}{1pt}
\begin{array}{rcl}
\dot{\tilde{\omega}}^r_0&=&-k_{R^r}\sum\limits_{k=1}^{N}(\tilde{\mathcal{L}}^{r,\ast}_2)_{ik}\mathbf{E}(A^r_k\tilde{R}_k^{r,\ast})\tilde{\mathbf{r}}_k-k_{\omega^r}\mathrm{Row}_i(\mathbf{F}_1)\tilde{\omega}^r_0,
\end{array}
\]
Denoting
$\mathbf{E}_{A^r\tilde{R}^{r,\ast}}=\mathrm{diag}(\mathbf{E}(A^r_1\tilde{R}_1^{r,\ast}),\ldots,\mathbf{E}(A^r_N\tilde{R}_N^{r,\ast}))$, $\tilde{\mathbf{r}}=[\tilde{\mathbf{r}}_1^{\top}\ \cdots\ \tilde{\mathbf{r}}_N^{\top}]^{\top}$,
the linearization of the estimation error dynamics \eqref{estimatedcompactconstant}
around the undesired equilibrium in $\mathcal{C}^{r,u}$ is given by
\begin{equation}\label{linearizedobserver}
\setlength{\arraycolsep}{1pt}
\begin{array}{rcl}
\dot{\tilde{\mathbf{r}}}&=&(\tilde{\mathcal{L}}^{r,\ast}_2)^{\top}\tilde{\omega}^r_0,\\
\dot{\tilde{\omega}}^r_0&=&-k_{R^r}\tilde{\mathcal{L}}^{r,\ast}_2\mathbf{E}_{A^r\tilde{R}^{r,\ast}}\tilde{\mathbf{r}}-k_{\omega^r}\mathbf{F}_1\tilde{\omega}_0^r,\\
\end{array}
\end{equation}
and the correponding Jacobian matrix is given by
\begin{equation}\label{bfM1}
\mathbf{M}_{1}=\left[
\begin{array}{cc}
\mathbf{0}_{3N\times 3N}&(\tilde{\mathcal{L}}^{r,\ast}_2)^{\top}\\
-k_{R^r}\tilde{\mathcal{L}}^{r,\ast}_2\mathbf{E}_{A^r\tilde{R}^{r,\ast}}&-k_{\omega^r}\mathbf{F}_1\\
\end{array}
\right]_,\end{equation}
Hence, the Jacobian matrix of the linearization of \eqref{estimatedcompactconstant} and  \eqref{controlerrorcompact1} around the undesired equilibria is given by $\mathbf{M}=\mathrm{diag}(\mathbf{M}_1,\mathbf{M}_2)$, where $\mathbf{M}_1$ and $\mathbf{M}_2$ are defined in \eqref{bfM1} and \eqref{bfM2} respectively.

Following the procedure adopted in the proof of Theorem \ref{estimationstability}, one can show that $\mathbf{M}_1$ has at least one positive eigenvalue, implying that $\mathbf{M}$ has at least one positive eigenvalue.
Hence, the undesired equilibria in $\mathcal{C}^{r,u}$ are unstable.
As per the stable manifold theorem \cite{Perko_book}, the stable manifold  associated to the undesired equilibria in $\mathcal{C}^{r,u}$ has a zero Lebesgue measure. Consequently, the desired equilibrium $\mathcal{C}^{r,d}$ is almost globally attractive.
Together with the fact that $\mathcal{C}^{r,d}$ is stable, one concludes that $\mathcal{C}^{r,d}$ is almost globally asymptotically stable.

Finally, we prove the third item.
Since the desired equilibrium $\mathcal{C}^{r,d}$ is almost globally asymptotically stable, one has $\lim_{t\rightarrow\infty}\hat{R}_j^{\top}\hat{R}_i=I_3$ ($(j,i)\in\overline{\mathcal{E}}_0$), $\lim_{t\rightarrow\infty}\hat{\omega}_i=\mathbf{0}_{3}$, $\lim_{t\rightarrow\infty}\hat{R}_i^{\top}R_i=I_3$ and $\lim_{t\rightarrow\infty}(\omega_i-R_i^{\top}\hat{R}_i\hat{\omega}_i)=\mathbf{0}_{3}$ for $i\in\mathcal{V}$.
It follows that
\begin{align*}
&\lim_{t\rightarrow\infty}\omega_i=\lim_{t\rightarrow\infty}(\omega_i-R_i^{\top}\hat{R}_i\hat{\omega}_i)+\lim_{t\rightarrow\infty}R_i^{\top}\hat{R}_i\hat{\omega}_i=\mathbf{0}_{3}.
\end{align*}
As per Assumption \ref{assumpgraph_original}, for the observer of the $i$-th rigid body ($i\in\mathcal{V}$), there exists a path, denoted as $(0,j_1,\ldots,j_k,i)$, from the leader $0$ to the $i$-th rigid body.
This implies that
\begin{align*}
&\lim_{t\rightarrow\infty}R_0^{\top}R_i=\lim_{t\rightarrow\infty}R_0^{\top}\hat{R}_{j_1}\hat{R}_{j_1}^{\top}\hat{R}_{j_2}\cdots\hat{R}_{j_k}^{\top}\hat{R}_i\hat{R}_i^{\top}R_i=I_3,
\end{align*}
\textit{i.e.}, $\lim_{t\rightarrow\infty}R_i=R_0$ for $i\in\mathcal{V}$.
Therefore, all the rigid body attitudes synchronize to the desired attitude $R_0$ for almost all initial conditions.
\section{Proof of Theorem \ref{stabilityobserverless}} \label{appenF}
Let us prove the first item.
Consider the following Lyapunov function candidate:
\[
V_3=\frac{1}{2}(k_{\check{R}}\mathbf{tr}(\check{\mathbf{A}}(I_{3N}-\check{R}^r))+\check{\omega}_0^{\top}J\check{\omega}_0),
\]
where $\check{\mathbf{A}}=\mathrm{diag}(\check{A}_1,\ldots,\check{A}_N)$.
Due to the fact that $\check{A}_k=\check{A}_k^{\top}>0$ and  $\lambda_{\min}^{\mathbf{E}(\check{A}_k)}=\frac{1}{2}(\mathbf{tr}(\check{A}_k)-\lambda_{\max}^{\check{A}_k^{\top}}))$, one sees $\lambda_{\min}^{\mathbf{E}(\check{A}_k)}>0$ for $k\in\mathcal{M}_{\overline{\mathcal{G}}_0}$.
It follows from \eqref{e2.1} that $\mathbf{tr}(\check{A}_k(I_3-\check{R}^r_k))\geq 4\lambda_{\min}^{\mathbf{E}(\check{A}_k)}|\check{R}^r_k|_I^2$ for $k\in\mathcal{M}_{\overline{\mathcal{G}}_0}$, leading to $V_3>0$.
The time-derivative of $V_3$, along the trajectories of the synchronization error system \eqref{compactobserverless}, is given by $\dot{V}_3=-k_{\check{\omega}}\|\check{\omega}_0\|^2\leq 0$.
Hence, the desired equilibrium $\check{\mathcal{C}}^{d}$ is stable.

By virtue of LaSalle's invariance principle, all the trajectories of the error system  \eqref{compactobserverless} must converge to largest invariant set characterized by $\dot{V}_3=0$.
One can infer from $\dot{V}_3=0$ that $\check{\omega}_0=\mathbf{0}_{3N}$, which in turn implies that $\dot{\check{\omega}}_0=\mathbf{0}_{3N}$.
This leads to  $\check{\mathcal{L}}^r_2\check{\Psi}=\mathbf{0}_{3N}$.
Since Lemma \ref{le2} shows that $\check{\mathcal{L}}^r_2$ is non-singular, one has $\check{\Psi}=\mathbf{0}_{3N}$, \textit{i.e.}, $\psi(\check{A}_k\check{R}^r_k)=\mathbf{0}_3$, or equivalently, $\check{A}_k\check{R}^r_k=(\check{R}^r_k)^{\top}\check{A}_k$.
Hence, According to \cite[Lemma 2]{Mayhew2011}, $\check{R}^r_k$ converges to $\{I_3\}\cup\{\mathcal{R}(\pi,\check{u})|\check{u}\in\mathcal{E}_{v}(\check{A}_i)\}$, where $k\in\mathcal{M}_{\overline{\mathcal{G}}_0}$.
Consequently, the trajectories of the error systems \eqref{compactobserverless} converge to the set of equilibria $\check{\mathcal{C}}$.

Let us prove the second item now.
To show that the desired equilibrium $\check{\mathcal{C}}^{d}$ is almost globally asymptotically stable, we prove that the undesired equilibria in $\check{\mathcal{C}}^{u}$ are unstable and their associated stable manifold has zero Lebesgue measure.
To do so, we linearize the error system \eqref{compactobserverless}, and prove that the Jacobian matrix has at least one positive eigenvalue.

Consider the undesired equilibrium  $(\check{R}^{r,\ast},\mathbf{0}_{3N})$. 
Let $(\check{R}_k^{r,\ast})^{\top}\check{R}^r_k=\exp(\check{\mathbf{r}}_k^{\times})$ for $k\in\mathcal{M}_{\overline{\mathcal{G}}_0}$.
Since $\exp(\check{\mathbf{r}}_k^{\times})\approx I_3+\check{\mathbf{r}}_k^{\times}$ for $\check{\mathbf{r}}_k\approx 0$, one has $\check{R}^r_k\approx\check{R}_k^{r,\ast}(I_3+\check{\mathbf{r}}_k^{\times})$ for $k\in\mathcal{M}_{\overline{\mathcal{G}}_0}$.
Then the linearization of $\dot{\check{R}}^r_k=\check{R}^r\check{\omega}_k^{\times}$ around $\check{R}_k^{r,\ast}$ is given by $\dot{\check{\mathbf{r}}}_k=\check{\omega}_k=\mathrm{Row}_k((\check{\mathcal{L}}^{r,\ast}_2)^{\top})\check{\omega}_0$, where $\check{\mathcal{L}}^{r,\ast}_2$ is obtained by removing the first three rows of $\check{\mathcal{L}}^{r,\ast}$ defining by each $3$-by-$3$ block $(\check{\mathcal{L}}^{r,\ast})_{ik}$ to be
\begin{equation}\label{barbarL}
(\check{\mathcal{L}}^{r,\ast})_{ik}=\left\{
\begin{array}{ll}
-\check{R}_k^{r,\ast}, & k\in(\mathcal{M}_{\overline{\mathcal{G}}_0})^{+}_{i-1},  \\
I_3, & k\in(\mathcal{M}_{\overline{\mathcal{G}}_0})^{-}_{i-1},\\
\mathbf{0}_{3\times 3}, &\mathrm{otherwise}.\\
\end{array}
\right.
\end{equation}
One can show that $\psi(\check{A}_k\check{R}^r_k)\approx\mathbf{\mathbf{E}}(\check{A}_k\check{R}_k^{r,\ast})\check{\mathbf{r}}_k$,
where $k\in\mathcal{M}_{\overline{\mathcal{G}}_0}$.
Hence, the linearization \eqref{angularerror} around the undesired equilibrium in $\check{\mathcal{C}}^{u}$ is given by
\[
J_i\dot{\check{\omega}}_{i,0}=-k_{\check{R}}\sum\limits_{k=1}^{N}(\check{\mathcal{L}}^{r,\ast}_2)_{i,k}\mathbf{E}(\check{A}_k\check{R}_k^{r,\ast})\check{\mathbf{r}}_k-k_{\check{\omega}}\check{\omega}_{i,0}.
\]
Let $\check{\mathbf{r}}=[\check{\mathbf{r}}_1^{\top}\ \cdots\ \check{\mathbf{r}}_N^{\top}]^{\top}$, $\mathbf{E}_{\check{A}\check{R}^{r,\ast}}=\mathrm{diag}(\mathbf{E}(\check{A}_1\check{R}_1^{r,\ast}),$
$\ldots,\mathbf{E}(\check{A}_N\check{R}_N^{r,\ast}))$.
The linearization of \eqref{compactobserverless} around the undesired equilibrium in $\check{\mathcal{C}}^{u}$ is given by
\begin{equation}\label{linearizationobserverless}
\setlength{\arraycolsep}{1pt}
\begin{array}{rcl}
\dot{\check{\mathbf{r}}}&=&(\check{\mathcal{L}}^{r,\ast}_2)^{\top}\check{\omega}_0,\\
\dot{\check{\omega}}_0&=&-k_{\check{R}}J^{-1}\check{\mathcal{L}}^{r,\ast}_2\mathbf{E}_{\check{A}\check{R}^{r,\ast}}\check{\mathbf{r}}-k_{\check{\omega}}J^{-1}\check{\omega}_0,
\end{array}
\end{equation}
where the Jacobian matrix is as follows:
\[
\bar{\mathbf{M}}=\left[
\begin{array}{cc}
\mathbf{0}_{3N\times 3N} & (\check{\mathcal{L}}^{r,\ast}_2)^{\top} \\
-k_{{\check{R}^r}}J^{-1}\check{\mathcal{L}}^{r,\ast}_2\mathbf{E}_{\check{A}\check{R}^{r,\ast}}&-k_{\check{\omega}}J^{-1}
\end{array}
\right]_.
\]
Following the procedure adopted in the proof of Theorem \ref{estimationstability}, one can show that $\bar{\mathbf{M}}$ has at least one positive eigenvalue.
This implies that the undesired equilibria in $\check{\mathcal{C}}^{u}$ are unstable.
Consequently, the stable manifold associated to the undesired equilibria in $\check{\mathcal{C}}^{u}$ has zero Lebesgue measure as per the stable manifold theorem. Consequently, the desired equilibrium $\check{\mathcal{C}}^{d}$ is almost globally attractive.
Combining with the fact that $\check{\mathcal{C}}^{d}$ is stable, one concludes that $\check{\mathcal{C}}^{d}$ is almost globally asymptotically stable.

Finally, let us prove the third item.
It is clear that $\lim_{t\rightarrow\infty}R_j^{\top}R_i=I_3$, $\forall (j,i)\in\overline{\mathcal{E}}_0$, $\lim_{t\rightarrow\infty}\omega_i=\lim_{t\rightarrow\infty}\check{\omega}_{i,0}=\mathbf{0}_{3N}$, $\forall i\in\mathcal{V}$ due to the fact that the desired equilibrium $\check{\mathcal{C}}^d$ is almost globally asymptotically stable.
Under Assumption \ref{assumpgraph_original}, for the $i$-th
rigid body ($i\in\mathcal{V}$), there exists a path, denoted as $(0,j_1,\ldots,j_k,i)$ from the leader $0$ to the $i$-th rigid body.
Hence, one has
\begin{align*}
&\lim_{t\rightarrow\infty}R_0^{\top}R_i=\lim_{t\rightarrow\infty}R_0^{\top}R_{j_1}R_{j_1}^{\top}R_{j_2}\cdots R_{j_k}^{\top}R_i=I_3,
\end{align*}
leading to $\lim_{t\rightarrow\infty}R_i=R_0$ for $i\in\mathcal{V}$.
As a result, all the rigid body attitudes synchronize to the desired attitude $R_0$ for almost all initial conditions.


\begin{thebibliography}{10}
\providecommand{\url}[1]{#1}
\csname url@samestyle\endcsname
\providecommand{\newblock}{\relax}
\providecommand{\bibinfo}[2]{#2}
\providecommand{\BIBentrySTDinterwordspacing}{\spaceskip=0pt\relax}
\providecommand{\BIBentryALTinterwordstretchfactor}{4}
\providecommand{\BIBentryALTinterwordspacing}{\spaceskip=\fontdimen2\font plus
\BIBentryALTinterwordstretchfactor\fontdimen3\font minus
  \fontdimen4\font\relax}
\providecommand{\BIBforeignlanguage}[2]{{%
\expandafter\ifx\csname l@#1\endcsname\relax
\typeout{** WARNING: IEEEtran.bst: No hyphenation pattern has been}%
\typeout{** loaded for the language `#1'. Using the pattern for}%
\typeout{** the default language instead.}%
\else
\language=\csname l@#1\endcsname
\fi
#2}}
\providecommand{\BIBdecl}{\relax}
\BIBdecl

\bibitem{Liyiliang2026}
Y.~Li, J.-e. Feng, and A.~Tayebi, ``A leader-follower approach for the attitude
  synchronization of multiple rigid body systems on {S}{O}(3),'' in
  \emph{arXiv:2601.19086}, 2026.

\bibitem{abdessameud2013motion}
A.~Abdessameud and A.~Tayebi, \emph{Motion Coordination for VTOL Unmanned
  Aerial Vehicles: Attitude Synchronisation and Formation Control}.\hskip 1em
  plus 0.5em minus 0.4em\relax Springer Science \& Business Media, 2013.

\bibitem{zhao2021data}
W.~Zhao, H.~Liu, and F.~L. Lewis, ``Data-driven fault-tolerant control for
  attitude synchronization of nonlinear quadrotors,'' \emph{IEEE Transactions
  on Automatic control}, vol.~66, no.~11, pp. 5584--5591, 2021.

\bibitem{guo2022distributed}
Z.~Guo, H.~Li, H.~Ma, and W.~Meng, ``Distributed optimal attitude
  synchronization control of multiple quavs via adaptive dynamic programming,''
  \emph{IEEE Transactions on Neural Networks and Learning Systems}, vol.~35,
  no.~6, pp. 8053--8063, 2022.

\bibitem{tang2022event}
Y.~Tang, X.~Jin, Y.~Shi, and W.~Du, ``Event-triggered attitude synchronization
  of multiple rigid body systems with velocity-free measurements,''
  \emph{Automatica}, vol. 143, p. 110460, 2022.

\bibitem{ren2009distributed}
W.~Ren, ``Distributed cooperative attitude synchronization and tracking for
  multiple rigid bodies,'' \emph{IEEE Transactions on Control Systems
  Technology}, vol.~18, no.~2, pp. 383--392, 2009.

\bibitem{jin2025event}
X.~Jin, Y.~Tang, Y.~Shi, X.~Wu, and W.~Lin, ``Event-triggered attitude
  consensus of multiple rigid body systems with prescribed performance,''
  \emph{IEEE Transactions on Cybernetics}, vol.~55, no.~1, pp. 307--320, 2025.

\bibitem{wei2018finite}
J.~Wei, S.~Zhang, A.~Adaldo, J.~Thunberg, X.~Hu, and K.~H. Johansson,
  ``Finite-time attitude synchronization with distributed discontinuous
  protocols,'' \emph{IEEE Transactions on Automatic Control}, vol.~63, no.~10,
  pp. 3608--3615, 2018.

\bibitem{jin2021event}
X.~Jin, Y.~Shi, Y.~Tang, H.~Werner, and J.~Kurths, ``Event-triggered fixed-time
  attitude consensus with fixed and switching topologies,'' \emph{IEEE
  transactions on automatic control}, vol.~67, no.~8, pp. 4138--4145, 2021.

\bibitem{dimarogonas2009leader}
D.~V. Dimarogonas, P.~Tsiotras, and K.~J. Kyriakopoulos, ``Leader-follower
  cooperative attitude control of multiple rigid bodies,'' \emph{Systems \&
  Control Letters}, vol.~58, no.~6, pp. 429--435, 2009.

\bibitem{meng2010distributed}
Z.~Meng, W.~Ren, and Z.~You, ``Distributed finite-time attitude containment
  control for multiple rigid bodies,'' \emph{Automatica}, vol.~46, no.~12, pp.
  2092--2099, 2010.

\bibitem{zou2016distributed}
A.-M. Zou, A.~H. de~Ruiter, and K.~D. Kumar, ``Distributed finite-time
  velocity-free attitude coordination control for spacecraft formations,''
  \emph{Automatica}, vol.~67, pp. 46--53, 2016.

\bibitem{abdessameud2009attitude}
A.~Abdessameud and A.~Tayebi, ``Attitude synchronization of a group of
  spacecraft without velocity measurements,'' \emph{IEEE Transactions on
  Automatic Control}, vol.~54, no.~11, pp. 2642--2648, 2009.

\bibitem{abdessameud2012attitude}
A.~Abdessameud, A.~Tayebi, and I.~G. Polushin, ``Attitude synchronization of
  multiple rigid bodies with communication delays,'' \emph{IEEE Transactions on
  Automatic Control}, vol.~57, no.~9, pp. 2405--2411, 2012.

\bibitem{huang2021global}
Y.~Huang and Z.~Meng, ``Global finite-time distributed attitude synchronization
  and tracking control of multiple rigid bodies without velocity
  measurements,'' \emph{Automatica}, vol. 132, p. 109796, 2021.

\bibitem{he2025leader}
C.~He and K.~W.~S. Au, ``Leader-following consensus of multiple rigid body
  systems by bounded attitude feedback,'' \emph{IEEE Transactions on Control of
  Network Systems}, vol.~12, no.~2, pp. 1406--1414, 2025.

\bibitem{cai2016leader}
H.~Cai and J.~Huang, ``Leader-following attitude consensus of multiple rigid
  body systems by attitude feedback control,'' \emph{Automatica}, vol.~69, pp.
  87--92, 2016.

\bibitem{koditschek1989application}
D.~Koditschek, ``The application of total energy as a {Lyapunov} function for
  mechanical control systems,'' \emph{Contemporary Mathematics, American
  Mathematical Society, 1989}, vol.~97, 02 1989.

\bibitem{Bhat_SCL2000}
S.~P. Bhat and D.~S. Bernstein, ``A topological obstruction to continuous
  global stabilization of rotational motion and the unwinding phenomenon,''
  \emph{Systems \& Control Letters}, vol.~39, no.~1, pp. 63--70, 2000.

\bibitem{sarlette2009autonomous}
A.~Sarlette, R.~Sepulchre, and N.~E. Leonard, ``Autonomous rigid body attitude
  synchronization,'' \emph{Automatica}, vol.~45, no.~2, pp. 572--577, 2009.

\bibitem{thunberg2016consensus}
J.~Thunberg, J.~Goncalves, and X.~Hu, ``Consensus and formation control on
  ${S}{E} (3)$ for switching topologies,'' \emph{Automatica}, vol.~66, pp.
  109--121, 2016.

\bibitem{deng2021attitude}
J.~Deng, L.~Wang, and Z.~Liu, ``Attitude synchronization and rigid formation of
  multiple rigid bodies over proximity networks,'' \emph{Automatica}, vol. 125,
  p. 109388, 2021.

\bibitem{peng2018specified}
X.~Peng, Z.~Geng, and J.~Sun, ``The specified finite-time distributed
  observers-based velocity-free attitude synchronization for rigid bodies on
  ${S}{O} (3)$,'' \emph{IEEE Transactions on Systems, Man, and Cybernetics:
  Systems}, vol.~50, no.~4, pp. 1610--1621, 2018.

\bibitem{zou2018velocity}
Y.~Zou and Z.~Meng, ``Velocity-free leader--follower cooperative attitude
  tracking of multiple rigid bodies on ${S}{O} (3)$,'' \emph{IEEE Transactions
  on Cybernetics}, vol.~49, no.~12, pp. 4078--4089, 2019.

\bibitem{Tron_TAC2012}
R.~Tron, B.~Afsari, and R.~Vidal, ``Riemannian consensus for manifolds with
  bounded curvature,'' \emph{IEEE Transactions on Automatic Control}, vol.~58,
  no.~4, pp. 921--934, 2013.

\bibitem{markdahl2021synchronization}
J.~Markdahl, ``Synchronization on {Riemannian} manifolds: {Multiply} connected
  implies multistable,'' \emph{IEEE Transactions on Automatic Control},
  vol.~66, no.~9, pp. 4311--4318, 2021.

\bibitem{Bou_Ber_Tayebi_TAC2026}
M.~Boughellaba, S.~Berkane, and A.~Tayebi, ``Global attitude synchronization
  for heterogeneous multiagent systems on $\text{SO}(3)$,'' \emph{IEEE
  Transactions on Automatic Control}, vol.~71, no.~7, pp. 4558--4573, 2026.

\bibitem{Arcak2008}
H.~Bai, M.~Arcak, and J.~T. Wen, ``Rigid body attitude coordination without
  inertial frame information,'' \emph{Automatica}, vol.~44, no.~12, pp.
  3170--3175, 2008.

\bibitem{Bough_Tay_2025}
M.~Boughellaba and A.~Tayebi, ``Distributed attitude estimation for multiagent
  systems on {S}{O}(3),'' \emph{IEEE Transactions on Automatic Control},
  vol.~70, no.~1, pp. 657--664, 2025.

\bibitem{shivakumar1974sufficient}
P.~Shivakumar and K.~H. Chew, ``A sufficient condition for nonvanishing of
  determinants,'' \emph{Proceedings of the American mathematical society}, pp.
  63--66, 1974.

\bibitem{horn2012matrix}
R.~A. Horn and C.~R. Johnson, \emph{Matrix analysis}.\hskip 1em plus 0.5em
  minus 0.4em\relax Cambridge university press, 2012.

\bibitem{bough_Berk_Tay_Arxiv2025}
M.~Boughellaba, S.~Berkane, and A.~Tayebi, ``Global attitude synchronization
  for heterogeneous multi-agent systems on {S}{O}(3),''
  \emph{arXiv:2501.11188}, 2025.

\bibitem{BerkaneSoulaimanephdtheis}
S.~Berkane, ``Hybrid attitude control and estimation on {S}{O}(3),'' Ph.D.
  dissertation, The University of Western Ontario (Canada), 2017.

\bibitem{Mayhew2011}
C.~G. Mayhew and A.~R. Teel, ``Synergistic potential functions for hybrid
  control of rigid-body attitude,'' in \emph{2011 American Control Conference
  (ACC)}, San Francisco, CA, USA, 2011, pp. 875--880.

\bibitem{Perko_book}
L.~Perko, \emph{Differential Equations and Dynamical Systems}, 3rd~ed.\hskip
  1em plus 0.5em minus 0.4em\relax Springer, 2000.

\end{thebibliography}

\end{document}